\documentclass[a4paper, 11pt]{article}
\pdfoutput=1

\usepackage{jheppub}

\usepackage{hyperref}
\hypersetup{
  colorlinks,
  linkcolor={violet},
  citecolor={brown},
  urlcolor={gray}
}

\usepackage{graphicx}
\graphicspath{ {./} } 
\usepackage{mathrsfs}
\usepackage{amsmath}
\usepackage{enumitem}
\usepackage{mathtools}
\usepackage{amssymb}
\usepackage{stmaryrd}
\usepackage{amsthm}
\usepackage{tikz-cd}
\usepackage{tikz}
\usetikzlibrary{fit,shapes.geometric}
\usepackage{xcolor}
\usepackage{environ}
\usepackage{bm}
\usepackage{cleveref}

\input epsf.sty

\newcommand{\re}{{\rm e}}
\newcommand{\ri}{{\rm i}}

\newcommand{\e}{\mathbf{e}_1}
\newcommand{\mb}{{\mathsf{b}}}

\def\IZ{{\mathbb Z}}
\def\IR{{\mathbb R}}
\def\IC{{\mathbb C}}
\def\IH{{\mathbb H}}
\def\IQ{{\mathbb Q}}
\def\IN{{\mathbb N}}
\def\IP{{\mathbb P}}

\def\IF{{\mathbb F}}

\newcommand{\mO}{\mathsf{O}}
\newcommand{\mrho}{\mathsf{\rho}}

\newcommand{\mx}{\mathsf{x}}

\newcommand{\my}{\mathsf{y}}

\newcommand{\CA}{{\cal A}}

\newcommand{\CC}{{\cal C}}

\newcommand{\CK}{{\cal K}}
\newcommand{\CL}{{\cal L}}

\newcommand{\CO}{{\cal O}}

\newcommand{\CS}{{\cal S}}

\newcommand{\CV}{{\cal V}}

\newcommand{\be}{\begin{equation}}
\newcommand{\ee}{\end{equation}}
\newcommand{\ba}{\begin{aligned}}
\newcommand{\ea}{\end{aligned}}

\newcommand{\borel}{\mathcal{B}}

\newtheorem{theorem}{Theorem}[section]
\newtheorem{prop}[theorem]{Proposition}
\newtheorem{lemma}[theorem]{Lemma}
\newtheorem{cor}[theorem]{Corollary}

\newtheorem{definition}{Definition}[section]
\newtheorem{conjecture}{Conjecture}

\newtheorem{rmk}{Remark}[section]

\title{\huge{\textbf Strong-weak symmetry and quantum modularity of resurgent topological strings on local $\IP^2$}}

\author{Veronica Fantini$^a$ and Claudia Rella$^b$}
\affiliation{${}^a$IH\'ES, 91440 Bures-sur-Yvette, France \\ ${}^b$D\'epartement de Physique Th\'eorique, Universit\'e de Gen\`eve, CH-1211 Gen\`eve, Switzerland}

\emailAdd{fantini@ihes.fr}
\emailAdd{claudia.rella@unige.ch}

\abstract{Quantizing the mirror curve to a toric Calabi--Yau threefold gives rise to quantum operators whose fermionic spectral traces produce factorially divergent formal power series in the Planck constant and its inverse. These are conjecturally captured by the Nekrasov--Shatashvili and standard topological string free energies, respectively, via the TS/ST correspondence. 
The resurgent structures of the first fermionic spectral trace of local $\mathbb{P}^2$ in both weak and strong coupling limits were solved exactly by the second author in~\cite{Rella22}.
Here, we argue that a full-fledged strong-weak resurgent symmetry is at play, exchanging the perturbative/non-perturbative contributions to the holomorphic and anti-holomorphic blocks in the factorization of the spectral trace.
This relies on a global net of relations connecting the perturbative series and the discontinuities in the dual regimes, which is built upon the analytic properties of the $L$-functions with coefficients given by the Stokes constants and the $q$-series acting as their generating functions. Then, we show that the latter are holomorphic quantum modular forms for $\Gamma_1(3)$ and are reconstructed by the median resummation of their asymptotic expansions.}

\makeatletter
\gdef\@fpheader{\null}
\makeatother

\begin{document}

\maketitle
\flushbottom

\section{Introduction}\label{sec:intro}
The A-model topological string theory compactified on a local Calabi--Yau (CY) threefold $X$ is defined perturbatively by a worldsheet genus expansion in the string coupling constant $g_s$ and can be expressed in terms of the enumerative invariants of the background geometry. The fixed-genus topological string free energies $F_g(\bm{t})$, $g \ge 0$, are well-defined expansions in a common region of convergence near the large radius limit $\Re(t_i)\gg 1$ in the moduli space of K\"ahler structures of $X$. Yet, at a fixed point $\bm{t}$ in the convergence region, $F_g(\bm{t})$ diverges factorially as $(2g)!$~\cite{lecturesM}, signaling the presence of exponentially small corrections in $g_s$.

A non-perturbative definition of the topological string on $X$ has been proposed in~\cite{GHM, CGM2}. By mirror symmetry, the B-model topological string theory on the mirror $\hat{X}$ is described by the theory of deformations of its complex structures~\cite{M, BKMP}. This is encoded in algebraic equations parametrizing a family of Riemann surfaces $\Sigma$ of genus $g_\Sigma$ embedded in $\IC^* \times \IC^*$, whose quantization leads naturally to quantum operators acting on the real line. A Planck constant $\hbar \in \IR_{>0}$ is introduced as a quantum deformation parameter. The conjectural statement of~\cite{GHM, CGM2}, known as Topological String/Spectral Theory (TS/ST) correspondence, implies a strong-weak coupling duality $\hbar \propto g_s^{-1}$ and leads to exact formulae for the fermionic spectral traces $Z_X(\bm{N},\hbar)$, $\bm{N} \in \IN^{g_\Sigma}$, of these quantum-mechanical operators. These are well-defined functions of $\hbar$ and are expressed in terms of the total grand potential of the topological string on $X$. In particular, the A-model total free energies of both the conventional and Nekrasov--Shatashvili (NS) topological strings on $X$, that is, the generating functions of the Gromov--Witten and refined BPS invariants, are needed to define the total grand potential and can be regarded as non-perturbative corrections of one another in the appropriate regimes~\cite{SWH, GG}.
Thus, the TS/ST correspondence provides a way to access the non-perturbative effects associated with the factorial divergence of the topological string perturbation series in the spirit of large-$N$ gauge/string dualities. After the original proposal of~\cite{M, Marino:2007te, Marino:2008ya}, growing evidence indicates that the theory of resurgence~\cite{EcalleI} can be applied to obtain a systematic understanding of the hidden non-perturbative sectors of topological string theory. 

The machinery of resurgence uniquely associates a divergent formal power series with a collection of non-analytic, exponential-type corrections paired with a non-trivial set of complex numbers, known as Stokes constants, which capture information about the large-order behavior of the original asymptotic series and its additional non-perturbative sectors in a mathematically precise way. See, \emph{e.g.},~\cite{ABS, diver-book, Dorigoni} for reviews on resurgence and~\cite{lecturesM, bookM} for its application to gauge and string theories.
For instance, following the work of~\cite{C-SESV1, C-SESV2}, the resurgent structure of the topological string has recently been investigated by means of an operator formulation of the BCOV holomorphic anomaly equations~\cite{BCOV1, BCOV2} satisfied by the conventional closed topological string free energies on toric CY threefolds~\cite{GuM2}, which has later been extended to arbitrary CY threefolds~\cite{GKKM, IM} and Walcher’s real topological string~\cite{MS}. 
These works produced formal solutions for the multi-instanton trans-series extensions of the free energies that are conjecturally related to the resurgent structure of their perturbative genus expansion in a specific way, although a determination of which of the possible Borel singularities are realized and the values of their Stokes constants are missing. 
The same methods have been applied to the refined topological string and its NS limit on toric CY threefolds~\cite{GuM3, AMP, Gu23, GuGuo}. The Stokes constants in these resurgent structures are generally unknown, and only a few have been computed numerically. 
Yet, growing evidence supports a conjectural identification between them and enumerative invariants counting BPS states of the topological string. The case of the resolved conifold has been studied in~\cite{PS, ASTT, GHN, AHT}.

Simultaneously, the alternative approach of~\cite{GuM, Rella22}, which we follow and promote in this paper, investigates the resurgent structure of the topological string on a toric CY threefold $X$ via the factorially divergent numerical power series obtained by asymptotically expanding the fermionic spectral traces $Z_X(\bm{N}, \hbar)$ at fixed $\bm{N}$ in the limits of $\hbar \rightarrow 0$ and $\hbar \rightarrow \infty$. Crucially, this leads, in particular, to peacock patterns of singularities in the Borel plane and infinite sets of generally calculable and rational Stokes constants. Using the TS/ST correspondence, the dual $\hbar$-regimes of the fermionic spectral traces above are explicitly related to the free energies of the NS and conventional topological string theories on $X$, respectively. 

The power of studying the resurgence of the fermionic spectral traces is manifest in the leading work of the second author on the local $\IP^2$ geometry~\cite{Rella22}, which is the benchmark example among all non-compact CY threefolds. 
Indeed, the complete resurgent structures of the logarithm of the first fermionic spectral trace of local $\IP^2$ in both weak and strong limits in $\hbar$ have been exactly solved, resulting in proven analytic formulae for the Stokes constants.
These have a transparent and straightforward arithmetic meaning as divisor sum functions, while their generating series are known in closed form in terms of $q$-Pochhammer symbols. In addition, the Stokes constants are the coefficients of two explicit $L$-functions, whose analytic number-theoretic properties underlie a global net of relations connecting the resurgent structures at weak and strong coupling. These novel results are the starting ground for the investigation performed in this paper. The arithmetic duality discovered in~\cite{Rella22} is completed and upgraded into a beautiful full-fledged symmetry. 
We call it the \emph{strong-weak resurgent symmetry}. In a nutshell, it is an exact symmetry between the resurgent structures in the strong and weak coupling regimes building upon the interplay of $q$-series and $L$-functions and revolving around the central role played by the Stokes constants. 
This newly found symmetry acts at the level of the perturbative/non-perturbative contributions to the holomorphic and anti-holomorphic blocks in the factorization of the spectral trace in a precise way. As a side effect, we shed some light on one of the formal dissimilarities between topological strings and complex Chern--Simons theory pointed out in~\cite{GuM}. 

The richness of the resurgence of the spectral trace of local $\IP^2$ leads us to investigate further the properties of the generating functions of the weak and strong coupling Stokes constants, which we denote by $f_0$ and $f_\infty$, respectively. 
On the one hand, we prove that $f_0$ and $f_\infty$ and their images under Fricke involution are \emph{holomorphic quantum modular forms} of weight zero for the congruence subgroup $\Gamma_1(3) \subset \mathsf{SL}_2(\IZ)$ (see Theorems~\ref{thm:f0,finf-quantum modular} and~\ref{thm:fricke-quantum modular}). Notice that the group $\Gamma_1(3)$ is deeply related to the geometry of the problem since the moduli space parametrizing the complex structures of the mirror of local $\IP^2$ is the compactification of the quotient $\IH/\Gamma_1(3)$, where $\IH$ denotes the upper half of the complex plane. Moreover, the generating functions of Gromov--Witten invariants of local $\IP^2$ are known to be quasi-modular forms for $\Gamma_1(3)$~\cite{ABK,CI,Bousseau20}.
On the other hand, we study the summability properties of the asymptotic expansions of $f_0$ and $f_\infty$, which are explicitly governed by the perturbative series of the logarithm of the spectral trace in the two regimes in $\hbar$. We prove that the asymptotic expansion of $f_0$ is resummable and precisely reconstructs the generating function through the \emph{median resummation} (see Theorem~\ref{thm:median-sum-Sinf}). A similar result is only conjectured for $f_\infty$ (see Conjecture~\ref{conj:median-sum-S0}) and supported by numerical evidence.

In the companion paper~\cite{FR1maths}, we lift the paramount features of the spectral trace of local $\IP^2$ to a generic paradigm of \emph{modular resurgence}, which revolves around the functional equation satisfied by the resurgent $L$-functions. This leads us to present and discuss new perspectives on the resurgence of divergent formal power series and quantum modularity and to introduce the notion of a \emph{modular resurgent structure}.

This paper is organized as follows. In Section~\ref{sec:background}, we review the basic ingredients in the construction of the fermionic spectral traces of local $\IP^2$ and their conjectural resurgent structures in the weak and strong coupling regimes in $\hbar$. We then focus on the logarithm of the first fermionic spectral trace and review the key results obtained by the second author in~\cite{Rella22}. In Section~\ref{sec:1}, we delve into the analytic number-theoretic characterization of the dual resurgent structures discovered in~\cite{Rella22} and complete it into a full-fledged strong-weak resurgent symmetry. As we take the perspective of the Stokes constants, their generating functions, and their $L$-functions, we find new exact relations intertwining the dual regimes in $\hbar$. 
In Section~\ref{sec:summability-and-QM}, we study the quantum modularity properties of the generating functions of the Stokes constants and the summability of their asymptotic expansions. 
Finally, in Section~\ref{sec:conclusion}, we conclude and mention further perspectives to be addressed in future works. There are two Appendices.

\section{Resurgence of the spectral trace of local \texorpdfstring{$\IP^2$}{P2}}\label{sec:background}
In this section, we review the construction of the fermionic spectral traces from the quantization of the mirror curve of the local $\IP^2$ geometry and their connection to the topological string theory compactified on the same background via the conjectural TS/ST correspondence~\cite{GHM, CGM2}. We refer to~\cite{reviewM} for a thorough overview. We describe the recent proposal of~\cite{GuM, Rella22} on the resurgence of the asymptotic series obtained from the weak and strong coupling expansions of the fermionic spectral traces. Finally, we summarize the exact results of~\cite{Rella22} on the dual resurgent structures of the first fermionic spectral trace and their startling number-theoretic properties.

\subsection{Geometric setup and the fermionic spectral traces}\label{sec:geometry}
The local $\IP^2$ geometry is the simplest example of a toric del Pezzo Calabi--Yau (CY) threefold.\footnote{A toric del Pezzo CY threefold is defined as the total space of the canonical line bundle on a toric del Pezzo surface. Examples of toric del Pezzo surfaces are the projective surface $\IP^2$, the Hirzebruch surfaces $\IF_n$ for $n=0,1,2$, and the blowups of $\IP^2$ at $n$ points, usually denoted by $\mathcal{B}_n$, for $n=1,2,3$.} 
It is defined as the total space of the canonical line bundle on the projective surface $\IP^2$, that is, 
\be
X = \CO(-3) \rightarrow \IP^2 \, .
\ee
Mirror symmetry pairs $X$ with a mirror CY threefold $\hat{X}$ in such a way that the theory of variations of complex structures of the mirror $\hat{X}$ is encoded in the one-parameter family of algebraic equations
\be \label{eq: mcP2}
O_{\IP^2}(x,y) + \kappa = \re^x + \re^y + \re^{-x-y} + \kappa = 0 \, , \quad x,y \in \IC \, , 
\ee
where $\kappa$ is the complex deformation parameter of $\hat{X}$. Note that there are no mass parameters. 
Eq.~\eqref{eq: mcP2} describes a genus-one Riemann surface $\Sigma$ embedded in $\IC^* \times \IC^*$, which is known as the mirror curve of $X$ and determines the B-model topological string theory on $\hat{X}$~\cite{M, BKMP}. 
\begin{rmk}
Toric del Pezzo surfaces are classified by reflexive polyhedra in two dimensions. The polyhedron $\Delta_{\IP^2}$ associated with the projective surface $\IP^2$ is the convex hull of the origin together with the two-dimensional vectors
\be \label{eq: vecP2}
\nu^{(1)} = (1,0)\, , \quad  \nu^{(2)} = (0,1) \, , \quad \nu^{(3)} = (-1,-1) \, ,
\ee
as shown in Fig.~\ref{fig: polyP2}. Note that the function $O_{\IP^2}(x,y)$ in Eq.~\eqref{eq: mcP2} can be written as
\be
O_{\IP^2}(x, y) = \sum_{i=1}^{3} \exp \left( \nu_1^{(i)} x + \nu_2^{(i)} y \right) \, ,
\ee
where $\nu^{(i)} = (\nu_1^{(i)},\nu_2^{(i)})$, $i=1,2,3$.
\begin{figure}[htb!]
\center
 \includegraphics[width=0.3\textwidth]{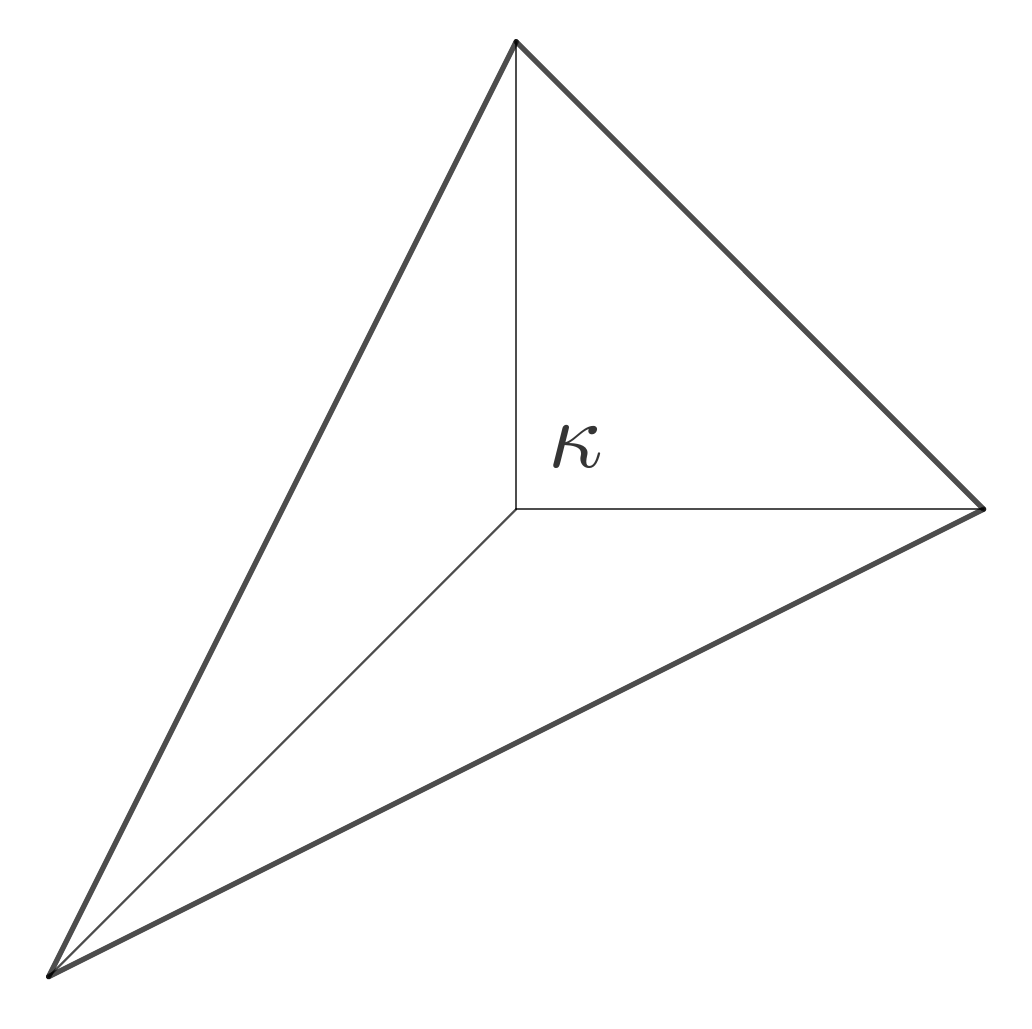}
 \caption{Toric diagram of the local $\IP^2$ geometry. We show the vectors in Eq.~\eqref{eq: vecP2} and the polyhedron $\Delta_{\IP^2}$. The complex modulus $\kappa$ corresponds to the internal vertex of the diagram.}
 \label{fig: polyP2}
\end{figure}
\end{rmk}
The Batyrev coordinate $z = 1/\kappa^3$ is related to the K\"ahler parameter of $X$ via the classical mirror map 
\be \label{eq: MM}
- t(z) = \log z + \tilde{S}(z) \, , 
\ee
where $\tilde{S}(z)$ is a power series in $z$ with finite radius of convergence. 
Following a choice of symplectic basis $\{A, B\}$ of one-cycles on the spectral curve $\Sigma$, the classical periods of the meromorphic differential one-form $\lambda = y(x)dx$, where the function $y(x)$ is locally defined by Eq.~\eqref{eq: mcP2}, satisfy
 \be \label{eq: periods}
 t(z) \propto \oint_{A} \lambda \, , \quad \partial_{t} F_0(z) \propto \oint_{B} \lambda \, ,
  \ee
where the function $F_0(z)$ is the classical prepotential of local $\IP^2$~\cite{CGPO}, representing the genus-zero amplitude of the B-model topological string on $\hat{X}$---or, equivalently, the generating function of the genus-zero Gromow--Witten invariants of $X$ convoluted with the mirror map. 
In the parametrization of the moduli space given by the Batyrev coordinate, the Picard--Fuchs differential equation associated with the mirror of local $\IP^2$ is 
\be \label{eq: PF}
 \CL_z \Pi = \left( \Theta^3-3z(3 \Theta+1)(3 \Theta+2)\Theta \right) \Pi = 0 \, ,
 \ee
where $\Theta = z d/d z$ and $\Pi$ is the full period vector of $\lambda$. The Picard--Fuchs differential operator $\CL_z$ has three singular points, which are the large radius point at $z=0$, the conifold point at $z=-1/27$, and the orbifold point at $1/z=0$. In addition, by looking at the monodromy data of the Picard--Fuchs equation, the moduli space of complex structures on $\hat{X}$ can be identified with the compactified quotient $\overline{\IH/\Gamma_1(3)}$, where $\IH$ denotes the upper half-plane and $\Gamma_1(3)$ is the congruence subgroup
\be \label{eq: G13-definition}
\Gamma_1(3):=\left\lbrace \begin{pmatrix}
    a & b\\
    c & d
\end{pmatrix}\in\mathsf{SL}_2(\IZ)\, \Big\vert\, a,d \equiv_3 1 \,, c\equiv_3 0\right\rbrace \, ,
\ee
which is generated by the elements
\be \label{eq: G13-generators}
T=\begin{pmatrix}
        1 & 1 \\
        0 & 1
    \end{pmatrix} \, , \quad  \gamma_3=\begin{pmatrix}
        1 & 0 \\
        3 & 1
    \end{pmatrix} \,.
\ee
We refer to~\cite[Section~10.1]{CI} for a detailed discussion.

The canonical Weyl quantization scheme of~\cite{GHM, CGM2} applied to the function $O_{\IP^2}(x,y)$ in Eq.~\eqref{eq: mcP2}, together with an appropriate choice of reality conditions for the variables $x,y \in \IC$, produces the Hermitian quantum operator\footnote{The spectral theory of $\mO_{\IP^2}$ has been studied numerically in~\cite{HW}.}
\be \label{eq: opP2}
\mO_{\IP^2}(\mx, \my) = \re^{\mx} + \re^{\my} + \re^{-\mx-\my} \, ,
\ee
acting on $L^2(\IR)$, where $\mx, \, \my$ are self-adjoint Heisenberg operators on the real line satisfying the commutation relation $[ \mx , \, \my ] = \ri \hbar$. A Planck constant $\hbar \in \IR_{>0}$ is introduced as a quantum deformation parameter. 
Consequently, the classical mirror map $t(z)$ in Eq.~\eqref{eq: MM} is promoted to the quantum mirror map
 \be \label{eq: qMM}
 - t(z, \hbar) = \log z + \tilde{S}(z, \hbar) \, ,
 \ee
which reproduces the conventional mirror map in the semiclassical limit $\hbar \rightarrow 0$ and is determined as an A-period of a quantum-corrected version of the differential $\lambda$ obtained via the all-orders perturbative WKB approximation~\cite{MiMo, ACDKV}.
It was conjectured in~\cite{GHM, CGM2} and later proven in~\cite{KM} that the inverse operator 
\be
\rho_{\IP^2} = \mO_{\IP^2}^{-1}
\ee 
is positive-definite and of trace class. 
Its spectral (or Fredholm) determinant
\be
\Xi_{\IP^2}(\kappa, \hbar) = \mathrm{det} \left(1 + \kappa \rho_{\IP^2} \right)
\ee
is an entire function of $\kappa$~\cite{CGM2, CGuM}, while its fermionic spectral traces $Z(N, \hbar)$, where $N$ is a non-negative integer, are well-defined functions of $\hbar \in \IR_{>0}$ and can be computed explicitly. In particular, they appear in the convergent power series expansion of the analytically continued spectral determinant around the orbifold point $\kappa = 0$, that is, 
\be \label{eq: expXi}
\Xi_{\IP^2}(\kappa, \hbar) = 1 + \sum_{N=1}^{\infty} Z(N, \hbar) \, \kappa^{N} \, .
\ee
Besides, classical results in Fredholm's theory provide explicit determinant expressions for the fermionic spectral traces, which can be regarded as multi-cut matrix model integrals~\cite{MZ, KMZ}. In particular, we have the multi-dimensional integral representation
\be \label{eq: traces-def}
\ba
Z(N, \hbar) &= \frac{1}{N!} 
\sum_{\sigma \in \CS_N} (-1)^{\epsilon(\sigma)} \int \prod_{i=1}^N \rho_{\IP^2}(x_i, x_{\sigma(i)}) \, d^N x \\
&=\frac{1}{N!} \int \mathrm{det} \left( \rho_{\IP^2}(x_i, x_j) \right) \, d^N x \, ,
\ea
\ee
where $\CS_N$ is the $N$-th permutation group, $\epsilon(\sigma)$ is the signature of the element $\sigma \in \CS_N$, and $\rho_{\IP^2}(x_i, x_j)$ denotes the integral kernel of the operator $\rho_{\IP^2}$. This is known explicitly as
\be \label{eq: P2kernel}
\rho_{\IP^2}(x, \, y) = \frac{\re^{\pi \mb (x+y) /3}}{2 \mb \cosh(\pi (x-y)/\mb + \ri \pi /6)} \frac{\Phi_{\mb} (y + \ri \mb /3) }{\Phi_{\mb} (x - \ri \mb/3)} \, ,
\ee
where $\mb$ is related to $\hbar$ by 
\be 
2 \pi \mb^2 = 3 \hbar
\ee
and $\Phi_{\mb}$ denotes Faddeev's quantum dilogarithm (see Appendix~\ref{app: Faddeev} for its definition and some of its fundamental properties).
Note that the integral kernel in Eq.~\eqref{eq: P2kernel} is well-defined for 
\be \label{eq: complexhbar}
\hbar \in \IC'=\IC \backslash \IR_{\le 0} \, ,
\ee
since $\Phi_{\mb}$ can be analytically continued to all values of $\mb$ such that $\mb^2 \notin \IR_{\le 0}$. Consequently, the same holds for the fermionic spectral traces in Eq.~\eqref{eq: traces-def}.

The conjectural statement of the TS/ST correspondence of~\cite{GHM, CGM2}, which is now supported by an increasing amount of evidence obtained in applications to concrete examples, implies an exact formula for the fermionic spectral traces of local $\IP^2$ in terms of the topological string amplitudes on the same background.
Namely, we have that\footnote{
Although being initially defined for non-negative integer values of $N$, the Airy-type integral in Eq.~\eqref{eq: contour} naturally extends the fermionic spectral traces to entire functions of $N \in \IC$~\cite{CGM}.}
\be \label{eq: contour}
 Z(N, \hbar) = \frac{1}{2 \pi \ri} \int_{- \ri \infty}^{\ri \infty} d \mu  \, \exp(J(\mu, \hbar) -  \mu N) \, , \quad N \in \IN \, ,
\ee
where the integration contour along the imaginary axis can be suitably deformed to guarantee convergence, the chemical potential $\mu$ is related to the complex modulus by $\kappa = \re^{\mu}$, and $J(\mu, \hbar)$ is the total grand potential of topological string theory on $X$.
We recall that this is given by the sum
\be \label{eq: J_total}
J(\mu, \hbar) = J^{\rm WS}(\mu, \hbar) + J^{\rm WKB}(\mu, \hbar) \, ,
\ee
where the worldsheet component $J^{\rm WS}(\mu, \hbar)$ encodes the non-perturbative corrections in $\hbar$ due to complex instantons, while the WKB component $J^{\rm WKB}(\mu, \hbar)$ contains the perturbative contributions in $\hbar$ to the quantum-mechanical spectral problem in the mirror B-model~\cite{HMMO}.
Precisely, $J^{\rm WS}(\mu, \hbar)$ can be expressed in terms of the all-genus generating functional of Gopakumar--Vafa invariants of $X$~\cite{GV, PT}, which is captured by the total free energy of the A-model standard topological string, while $J^{\rm WKB}(\mu, \hbar)$ is a function of the all-order generating functional of refined BPS invariants of $X$~\cite{NS, CKK, NO}, which is captured by the total free energy of the A-model Nekrasov--Shatashvili (NS) topological string\footnote{The standard and NS topological string partition functions can be engineered as two one-parameter specializations of the partition function of the refined topological string theory compactified on a toric CY background~\cite{N}.}.
As we will discuss in Section~\ref{sec: resurgent_strings}, only one of the two components of the total grand potential above contributes to each of the dual perturbative regimes of $\hbar\to 0$ and $\hbar \to \infty$.
Moreover, and importantly, the TS/ST correspondence implies the identification
\be \label{eq: duality}
g_s = \frac{4 \pi^2}{\hbar} \, ,
\ee
where $g_s$ is the conventional topological string coupling constant. The statement above can be interpreted as a strong-weak coupling duality between the spectral theory of the operator $\rho_{\IP^2}$ arising from the quantization of the mirror curve $\Sigma$ and the standard topological string theory on $X$. We will return to this in Section~\ref{sec: physics_arg}. 

\begin{rmk}
In its original formulation, the TS/ST correspondence of~\cite{GHM, CGM2}, and therefore also the statement in Eq.~\eqref{eq: contour}, applies to $\hbar \in \IR_{>0}$.  
The issue of the complexification of $\hbar$, or, equivalently, of $g_s$, in the context of the TS/ST correspondence, has been addressed in various studies since~\cite{Ka, KS1, KS2, GM2}. As we mentioned above, the integral kernel of local $\IP^2$ in Eq.~\eqref{eq: P2kernel} admits a well-defined analytic continuation to $\hbar \in \IC'$. Moreover, the analytically continued fermionic spectral traces obtained in this way match the natural analytic continuation of the total grand potential of local $\IP^2$ to complex values of $\hbar$ such that $\Re (\hbar) >0$. The TS/ST correspondence is, then, still applicable.
\end{rmk}

\subsection{Resurgent structures in the weak and strong coupling regimes}\label{sec: resurgent_strings}
In light of the strong-weak duality between the topological string coupling $g_s$ and the quantum deformation parameter $\hbar$ in Eq.~\eqref{eq: duality}, the regimes of $\hbar \rightarrow 0$ and $\hbar \rightarrow \infty$ of the spectral theory correspond to the strong and weak coupling limits of the dual topological string theory on local $\IP^2$, respectively. Besides, due to the functional forms of the worldsheet and WKB grand potentials appearing in Eq.~\eqref{eq: J_total}, there are appropriate scaling regimes in the coupling constants in which only one of the two components effectively contributes to the total grand potential. In the standard double-scaling limit~\cite{MP, KKN, GM}
\be \label{eq: SS1}
\hbar \rightarrow \infty \, , \quad \mu \rightarrow \infty \, , \quad \frac{\mu}{\hbar} = \zeta \; \; \mathrm{fixed} \, , 
\ee 
the quantum mirror map $t(z, \hbar)$ in Eq.~\eqref{eq: qMM} becomes trivial, and the total grand potential in Eq.~\eqref{eq: J_total} reduces to its worldsheet component, which has the all-genus asymptotic expansion
\be \label{eq: tHooft}
J^{\rm WS}(\zeta, \hbar) \, \sim \, \sum_{g =0}^{\infty} J_g(\zeta) \, \hbar^{2-2g} \, ,
\ee
where $J_g(\zeta)$ is essentially the genus $g$ free energy of the conventional topological string in the large radius limit $\Re(t) \gg 1$ of the moduli space of $X$ after the so-called B-field has been turned on.
Similarly, in the dual semiclassical regime
\be \label{eq: DS1}
\hbar \rightarrow 0 \, , \quad \mu \; \; \mathrm{fixed} \, ,
\ee
the quantum mirror map $t(z, \hbar)$ in Eq.~\eqref{eq: qMM} becomes classical by construction, and the total grand potential in Eq.~\eqref{eq: J_total} retains only the WKB contribution, which can be formally expanded in powers of $\hbar$ as
\be \label{eq: dual}
J^{\text{WKB}}(\mu, \hbar) \, \sim \, \sum_{n=0}^{\infty} J^{\rm NS}_n(\mu) \, \hbar^{2n-1} \, ,
\ee
where the fixed-order WKB grand potential $J^{\rm NS}_n(\mu)$ is a function of the order $n$ free energy of the NS topological string at large radius and its derivative in the K\"ahler parameter $t$.
Evidence suggests that both asymptotic series in Eqs.~\eqref{eq: tHooft} and~\eqref{eq: dual} diverge factorially~\cite{HMMO}. 
Thus, their non-perturbative completions will contain exponentially small effects in $\hbar$ and $1/\hbar$, respectively. See the discussion in Section~\ref{sec: physics_arg}.

Due to the TS/ST correspondence, there are related 't Hooft limits for the fermionic spectral traces that extract the perturbative all-genus and all-order expansions of the conventional and NS topological strings on local $\IP^2$, respectively. In the standard 't Hooft limit associated with Eq.~\eqref{eq: SS1}, that is, 
\be \label{eq: SS2}
\hbar \rightarrow \infty \, , \quad N \rightarrow \infty \, , \quad \frac{N}{\hbar} = \lambda \; \; \text{fixed} \, , 
\ee
the saddle-point evaluation of the integral in Eq.~\eqref{eq: contour} represents a symplectic transformation from the large radius point in moduli space to the so-called maximal conifold point\footnote{The maximal conifold point can be defined as the unique point in the conifold locus of moduli space where its connected components intersect transversally.}~\cite{CGM2, CGuM}.
Specifically, it follows from the geometric formalism of~\cite{ABK} that the 't Hooft parameter $\lambda$ is the flat coordinate corresponding to the maximal conifold frame of local $\IP^2$. In the recent work of~\cite{Rella22}, a dual WKB double-scaling regime associated with Eq.~\eqref{eq: DS1} is studied, that is, 
\be \label{eq: DS2}
\hbar \rightarrow 0 \, , \quad N \rightarrow \infty \, , \quad N \hbar = \sigma \; \; \text{fixed} \, ,
\ee
so that the symplectic transformation of the total grand potential at large radius encoded in the integral in Eq.~\eqref{eq: contour} can be interpreted again as a suitable change of frame in moduli space. The dual 't Hooft coupling $\sigma$ is shown to be a simple function of the periods in Eq.~\eqref{eq: periods} and the modular coordinate $\partial_t^2 F_0(z)$ which appears in the study of the modular properties and BPS spectrum of local $\IP^2$~\cite{CI, Bousseau20, Bousseau22}. 

Let us go back to the fermionic spectral traces $Z(N, \hbar)$, $N \in \IZ_{>0}$, in Eq.~\eqref{eq: traces-def} and perturbatively expand them in the limits $\hbar \rightarrow 0$ and $\hbar \rightarrow \infty$ with $N$ fixed. Following~\cite{GuM, Rella22}, we construct the two families of asymptotic series
\begin{subequations} \label{eq: pert_exp}
\begin{align} 
\log Z(N, \hbar) &\sim \phi_N(\hbar) \, \quad \quad \, \mathrm{for} \; \; \; \hbar \rightarrow 0 \, , \label{eq: pert_exp_0} \\ 
\log Z(N, \hbar) &\sim \psi_N(\hbar^{-1}) \, \quad \, \mathrm{for} \; \; \; \hbar \rightarrow \infty  \, , \label{eq: pert_exp_inf}
\end{align}
\end{subequations}
indexed by the non-negative integer $N$. Under the assumption that the formal power series in Eq.~\eqref{eq: pert_exp} are Gevrey-1 and simple resurgent, which is explicitly verified in the examples considered in this paper, the theory of resurgence can be applied to give us access to the non-analytic sectors that are hidden in perturbation theory. We briefly recall some of the basic notions of resurgence in Appendix~\ref{app: resurgence}. In order to perform a resurgent analysis of the fermionic spectral traces, we consider their analytic continuation to $\hbar \in \IC'$.
A conjectural proposal for the resurgent structure of the asymptotic series in Eq.~\eqref{eq: pert_exp} has been recently formulated\footnote{The conjecture of~\cite{GuM, Rella22} describes a universal structure underlying the resurgence of the strong and weak coupling perturbative expansions of the fermionic spectral traces for all local CY threefolds. Explicit results in support of the conjecture are obtained in the examples of local $\IP^2$ and local $\IF_0$.} in~\cite{GuM, Rella22}.
Let us briefly recall its main statements. For fixed $N \in \IZ_{>0}$, each of the formal power series in Eq.~\eqref{eq: pert_exp_0}, which emerge in the semiclassical limit $\hbar \rightarrow 0$ of the spectral theory, is associated with a minimal resurgent structure $\mathfrak{B}_{\phi_N}$ of the particular form
\be
\mathfrak{B}_{\phi_N} = \{ \Phi_{\sigma, n ; N}(\hbar) \}_{\sigma = 0, \dots, l_0 ; \, n \in \IZ} \, , \quad \Phi_{\sigma, n ; N}(\hbar) = \phi_{\sigma ; N}(\hbar) \re^{- n \frac{\mathcal{A}_0}{\hbar}} \, ,
\ee
where the non-negative integer $l_0$ and the complex constant $\mathcal{A}_0$ depend on the choice of $N$. Namely, for each value of $\sigma \in \{ 0 , \dots , l_0 \}$, a Gevrey-1 asymptotic series $\phi_{\sigma ; N}(\hbar)$ resurges from the original perturbative expansion $\phi_N(\hbar) = \phi_{0 ; N}(\hbar)$ and the singularities of its Borel transform $\hat{\phi}_{\sigma ; N}(\zeta)$ are located along infinite towers in the Borel $\zeta$-plane so that every two singularities in the same tower are spaced by an integer multiple of $\mathcal{A}_0$. 
The global arrangement of the complete set of Borel singularities of $\mathfrak{B}_{\phi_N}$ is known as a peacock pattern\footnote{Peacock patterns are typically observed in theories controlled by a quantum curve in exponentiated variables, including complex Chern--Simons theory on the complement of a hyperbolic knot~\cite{GGuM, GGuM2}.}. 
In this way, each asymptotic series $\phi_{\sigma ; N}(\hbar)$ gives rise to an infinite family of basic trans-series $\Phi_{\sigma, n ; N}(\hbar)$, $n \in \IZ$.
The corresponding infinite-dimensional matrix of Stokes constants $\mathcal{S}_{\phi_N}$ can be written as
\be
\mathcal{S}_{\phi_N} = \{ S_{\sigma, \sigma', n ; N} \in \IQ \}_{\sigma, \sigma' = 0, \dots, l_0 ; \, n \in \IZ} \, ,
\ee
after fixing a canonical normalization of the asymptotic series $\phi_{\sigma ; N}(\hbar)$. It is conjectured that the Stokes constants $S_{\sigma, \sigma', n ; N}$ are closely related to non-trivial sequences of integers, thus representing a new class of enumerative invariants of the topological string on local $\IP^2$. Furthermore, they can be naturally organized as the coefficients of generating functions that are expressible in the form of $q$-series and are uniquely determined by the original perturbative expansion $\phi_N(\hbar)$. Schematically, 
\be
S_{\sigma \sigma' ; N}(q) = \sum_{n \in \IZ}  S_{\sigma \sigma' , n ; N} \, q^n \, .
\ee
The conjectural proposal of~\cite{GuM, Rella22} for the formal power series in Eq.~\eqref{eq: pert_exp_inf}, which emerge in the strongly-coupled limit $\hbar \rightarrow \infty$ of the spectral theory, is entirely analogous. At fixed $N \in \IZ_{>0}$, the minimal resurgent structure $\mathfrak{B}_{\psi_N}$ is given by
\be
\mathfrak{B}_{\psi_N} = \{ \Psi_{\sigma, n ; N}(\hbar^{-1}) \}_{\sigma = 0, \dots, l_\infty ; \, n \in \IZ} \, , \quad \Psi_{\sigma, n ; N}(\hbar^{-1}) = \psi_{\sigma ; N}(\hbar^{-1}) \re^{- n \mathcal{A}_\infty \hbar} \, ,
\ee
where again $l_\infty \in \IZ_{>0}$ and $\mathcal{A}_\infty \in \IC$ depend on the choice of $N$, resulting in a dual peacock arrangement of singularities in the complex Borel plane. The infinite-dimensional matrix of Stokes constants $\mathcal{S}_{\psi_N}$ is instead
\be
\mathcal{S}_{\psi_N} = \{ R_{\sigma, \sigma', n ; N} \in \IQ \}_{\sigma, \sigma' = 0, \dots, l_\infty ; \, n \in \IZ} \, ,
\ee
after fixing the normalization of the asymptotic series $\psi_{\sigma ; N}(\hbar^{-1})$. Once more, the Stokes constants $R_{\sigma, \sigma', n ; N}$ provide a new non-trivial class of integer invariants of the geometry and can be naturally organized into the generating $q$-series
\be
R_{\sigma \sigma' ; N}(q) = \sum_{n \in \IZ}  R_{\sigma \sigma' , n ; N} \, q^n \, ,
\ee
which are uniquely determined by the original perturbative expansion $\psi_N(\hbar^{-1})$. 

\begin{rmk}
The work of~\cite{GuM} refers to the exponential of the series in Eq.~\eqref{eq: pert_exp_inf}, that is, the perturbative expansion in $g_s$ of the fermionic spectral traces directly. As pointed out in~\cite{Rella22}, however, the arithmetic properties underlying the resurgent structures of the series in Eq.~\eqref{eq: pert_exp} for $N=1$ become less readily accessible after exponentiation. For example, the Stokes constants for $\exp (\phi_1(\hbar))$ are complex numbers with no manifest interpretation beyond the fact that they can be expressed in terms of the Stokes constants of $\phi_1(\hbar)$ through a closed partition-theoretic formula~\cite{Rella22}. Similarly, the exact number-theoretic duality between the weak and strong coupling regimes emerges when looking at the logarithm of the spectral trace, which is akin to considering the free energy instead of the partition function. Let us stress that, as shown explicitly in~\cite{Rella22}, the analytic solution to the resurgent structures of the series in Eq.~\eqref{eq: pert_exp} can be translated into results on the corresponding exponentiated series by applying the tools of alien calculus~\cite{ABS, diver-book, Dorigoni, lecturesM}.   
\end{rmk}

\begin{rmk}\label{rmk:perturbative/non-perturbative}
Let us stress that the TS/ST statement in Eq.~\eqref{eq: contour} and the strong-weak coupling duality in Eq.~\eqref{eq: duality} map the semiclassical $\hbar$-series in Eq.~\eqref{eq: pert_exp_0} into a dual perturbative $\hbar$-expansion that is described by the NS limit of the refined topological string on local $\IP^2$. Analogously, the asymptotic series for $\hbar \rightarrow \infty$ in Eq.~\eqref{eq: pert_exp_inf} is associated with a dual perturbative expansion in the weakly coupled regime $g_s \rightarrow 0$ that is captured by the conventional topological string on the same geometry. Notably, the standard and NS topological string free energies can be regarded as non-perturbative corrections of one another in the appropriate regimes~\cite{SWH, GG}. The interplay of the two one-parameter specializations of the refined topological string theory on a local CY threefold with the perturbative and non-perturbative contributions to the fermionic spectral traces 
in the dual limits in $\hbar$ is discussed in Section~\ref{sec: physics_arg}.
\end{rmk}

\subsection{Exact solutions for the spectral trace}\label{sec:exact-solution}
In the rest of this paper, unless explicitly stated, we will focus on the spectral trace 
\be
Z_{\IP^2}(1, \hbar) = \mathrm{Tr}(\mrho_{\IP^2}) \, ,
\ee 
whose resurgent structures in the semiclassical limit $\hbar \rightarrow 0$ and the dual strongly-coupled limit $\hbar \rightarrow \infty$ have been obtained analytically in closed form in~\cite{Rella22}. This is the only example of a fermionic spectral trace of a local CY threefold to have been solved exactly.
Under the assumption that $\Re (\mb) >0$, it has the integral representation~\cite{MZ}
\be \label{eq: P2int}
\mathrm{Tr}(\mrho_{\IP^2}) = \frac{1}{\sqrt{3} \mb} \int_{\IR} \re^{2 \pi \mb x /3} \frac{\Phi_{\mb} (x + \ri \mb/3)}{\Phi_{\mb} (x - \ri \mb/3)} \, d x \, ,
\ee
which is an analytic function of $\hbar \in \IC'$. The integral above can be evaluated using the integral Ramanujan formula or analytically continuing $x$ to the complex domain, completing the integration contour from above, and summing over residues, yielding the closed formula~\cite{KM}
\be \label{eq: P2close}
\mathrm{Tr}(\mrho_{\IP^2}) = \frac{1}{\sqrt{3} \mb} \re^{- \frac{\pi \ri}{36} (12 c_{\mb}^2+4 \mb^2-3)} \frac{\Phi_{\mb} \left( c_{\mb} - \frac{\ri \mb}{3} \right)^2}{\Phi_{\mb} \left( c_{\mb} - \frac{2 \ri \mb}{3} \right)} = \frac{1}{\mb} \left| \Phi_{\mb} \left( c_{\mb} - \frac{\ri \mb}{3} \right)\right|^3 \, ,
\ee
where $c_{\mb} = \ri (\mb + \mb^{-1})/2$. 
Moreover, following~\cite{GuM}, the expression in Eq.~\eqref{eq: P2close} can be factorized into a product of $q, \tilde{q}$-series by applying the infinite product representation in Eq.~\eqref{eq: seriesPhib}. 
Namely, we have that 
\be
\Phi_{\mb} \left( c_{\mb} - \frac{\ri \mb}{3} \right) =\frac{(q^{2/3} ; \, q)_{\infty}}{(w^{-1} ; \, \tilde{q})_{\infty}}  \, , \quad \Phi_{\mb} \left( c_{\mb} - \frac{2 \ri \mb}{3} \right) = \frac{(q^{1/3} ; \, q)_{\infty}}{(w ; \, \tilde{q})_{\infty}} \, ,
\ee
where $(x q^{\alpha}; \, q)_{\infty}$ is the quantum dilogarithm defined in Eq.~\eqref{eq: dilog}. Therefore,
\be \label{eq: P2fact}
\mathrm{Tr}(\mrho_{\IP^2}) = \frac{1}{\sqrt{3} \mb} \re^{- \frac{\pi \ri}{36} \mb^2 + \frac{\pi \ri}{12} \mb^{-2} + \frac{\pi \ri}{4}}  \frac{(q^{2/3} ; \, q)_{\infty}^2}{(q^{1/3} ; \, q)_{\infty}} \frac{(w ; \, \tilde{q})_{\infty}}{(w^{-1} ; \, \tilde{q})_{\infty}^2} \, ,
\ee
where we have introduced
\be \label{eq: q-var}
q = \re^{2 \pi \ri \mb^2} = \re^{3 \ri \hbar} \, , \quad \tilde{q} = \re^{- 2 \pi \ri / \mb^{2}} = \re^{- 4 \pi^2 \ri /(3 \hbar)} \, , \quad w = \re^{2 \pi \ri /3} \, .
\ee
Note that the factorization into holomorphic/anti-holomorphic blocks in Eq.~\eqref{eq: P2fact} is not symmetric in $q, \, \tilde{q}$.
We further assume that $\Im(\mb^2) > 0$, and therefore $\Im(\hbar) > 0$, which implies $|q|, |\tilde{q}| < 1$, so that the $q, \tilde{q}$-series converge. 

Let us now recall some of the novel results of~\cite{Rella22} on the exact resurgent structures of the perturbative expansions of $\mathrm{Tr}(\mrho_{\IP^2})$ in both weak and strong $\hbar$-regimes and the analytic number-theoretic duality relating them, which is the starting point for the investigation performed in this paper. In the following, we will summarize part of~\cite[Section 4]{Rella22}. We invite the interested reader to consult the original reference for a complete and detailed description with proofs.

\subsubsection*{The limit $\hbar \rightarrow 0$}
The all-orders perturbative expansion in the limit $\hbar \rightarrow 0$ of the spectral trace of local $\IP^2$ in Eq.~\eqref{eq: P2fact} is obtained by applying the asymptotic expansion formula for the quantum dilogarithm in Eq.~\eqref{eq: logPhiK} to the holomorphic $q$-series component\footnote{The $\tilde{q}$-series giving the anti-holomorphic block in the factorized expression in Eq.~\eqref{eq: P2fact} converge for $\hbar \rightarrow 0$.}. Namely, 
\be \label{eq: expP2}
\mathrm{Tr}(\mrho_{\IP^2}) \sim \frac{\Gamma\left(\frac{1}{3}\right)^3}{6 \pi \hbar} \exp \left( 3 \sum_{n = 1}^{\infty} (-1)^{n-1} \frac{B_{2n} B_{2n+1}(2/3)}{2n (2n+1)!} (3 \hbar)^{2n} \right) \, ,
\ee
where $B_n(x)$ is the $n$-th Bernoulli polynomial, $B_n = B_n(0)$ is the $n$-th Bernoulli number, and $\Gamma(x)$ is the gamma function.
Let us denote by $\phi(\hbar)$ the formal power series appearing in the exponent in Eq.~\eqref{eq: expP2}, that is, 
\be \label{eq: phiP2}
\phi(\hbar) = \sum_{n=1}^{\infty} a_{2n} \hbar^{2n} \in \IQ[\![\hbar]\!] \, , \quad a_{2n} = (-1)^{n-1} \frac{B_{2n} B_{2n+1}(2/3)}{2n (2n+1)!} 3^{2n+1}  \quad n \ge 1 \, .
\ee
Its perturbative coefficients satisfy the factorial growth
\be
|a_{2n}| \le (2n)! \mathcal{A}_0^{-2n}  \quad n \gg 1 \, , \quad \mathcal{A}_0 = \frac{4 \pi^2}{3} \, ,
\ee
and $\phi(\hbar)$ is thus a Gevrey-1 asymptotic series. The Borel transform $\hat{\phi}(\zeta)$ can be explicitly resummed into a well-defined, exact function of $\zeta$ that is manifestly simple resurgent. Its singularities are simple poles\footnote{We apply the conventions of Appendix~\ref{app: resurgence}, which differ slightly from those adopted in~\cite{Rella22} due to a change in the definition of the Borel and Laplace transforms. In particular, with the conventions of~\cite{Rella22}, the singularities of the Borel transforms of the series $\phi(\hbar), \psi(\tau)$ are logarithmic branch points.} located along the imaginary axis at
\be \label{eq: zetan}
\zeta_n = \mathcal{A}_0 \ri n \, , \quad n \in \IZ_{\ne 0} \, ,
\ee 
tracing two Stokes lines at the angles $\pm \pi/2$. This is the simplest occurrence of the peacock pattern described in Section~\ref{sec: resurgent_strings}.

The local expansion of the Borel transform $\hat{\phi}(\zeta)$ at $\zeta = \zeta_n$, $n \in \IZ_{\ne 0}$, is given by
\be \label{eq: explocalP2}
\hat{\phi}(\zeta) =  - \frac{S_n}{2 \pi \ri (\zeta - \zeta_n)} + \text{regular in $\zeta-\zeta_n$} \, ,
\ee
where $S_n$ is the corresponding Stokes constant.
It follows that all the Stokes constants can be accessed analytically.
In particular, the normalized Stokes constant $S_n/S_1$, where $S_1 = 3 \sqrt{3} \ri$, is given by the divisor sum function
\be \label{eq: formulaS1}
\frac{S_n}{S_1} = \sum_{\substack{d | n \\ d \equiv_3 1}} \frac{1}{d} - \sum_{\substack{d | n \\ d \equiv_3 2}} \frac{1}{d} = \sum_{\substack{d | n}} \frac{1}{d} \, \chi_{3,2}(d) \in \IQ_{>0} \, , \quad n \in \IZ_{\ne 0} \, , 
\ee
where $d | n$ denotes the positive integer divisors $d$ of $n$ and $\chi_{3,2}(n)$ is the unique non-principal Dirichlet character modulo $3$, that is, 
\be \label{eq: character32}
\chi_{3,2} (n) = 
\begin{cases}
0 & \quad \text{if} \quad n \equiv_3 0 \\
1 & \quad \text{if} \quad n \equiv_3 1 \\
-1 & \quad \text{if} \quad n \equiv_3 2 
\end{cases} \, , \quad n \in \IZ \, .
\ee
It follows that $S_n/S_1$ is a multiplicative arithmetic function and satisfies $S_{-n} = S_{n}$.
Note that we can introduce the sequence of integers
\be \label{eq: intStokesP2zero}
\alpha_n = \frac{S_n}{S_1} n \, , \quad n \in \IZ_{\ne 0} \, , 
\ee
which satisfies $\alpha_n \in \IZ_{>0}$ for $n \in \IZ_{> 0}$ and $\alpha_{-n} = -\alpha_{n}$. The pattern of singularities in the Borel plane and the associated $\alpha_n$, $n \in \IZ_{\ne 0}$ are shown in Fig.~\ref{fig: peacockP2} on the left.

\begin{figure}[htb!]
\center
\includegraphics[width=0.2\textwidth]{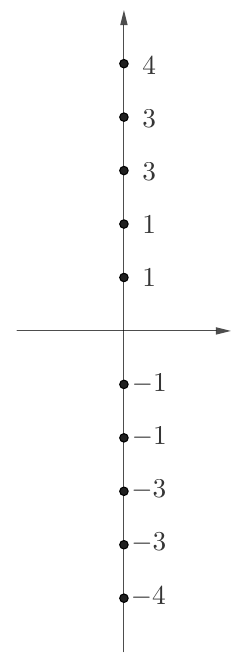} \quad \quad \quad
\includegraphics[width=0.2\textwidth]{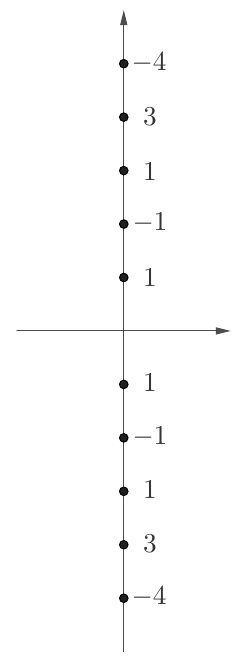}
\caption{On the left, the first few singularities of the Borel transform of the asymptotic series $\phi(\hbar)$, defined in Eq.~\eqref{eq: phiP2}, and the associated integer constants $\alpha_n$, $n \in \IZ_{\ne 0}$, defined in Eq.~\eqref{eq: intStokesP2zero}. On the right, the first few singularities of the Borel transform of the asymptotic series $\psi(\hbar)$, defined in Eq.~\eqref{eq: phiP2infty}, and the associated integer constants $\beta_n$, $n \in \IZ_{\ne 0}$, defined in Eq.~\eqref{eq: intStokesP2infty}.}
\label{fig: peacockP2}
\end{figure}

\subsubsection*{The limit $\hbar \rightarrow \infty$}
The all-orders perturbative expansion in the limit $\hbar \rightarrow \infty$ of the spectral trace of local $\IP^2$ in Eq.~\eqref{eq: P2fact} is obtained by applying the asymptotic expansion formula for the quantum dilogarithm in Eq.~\eqref{eq: logPhiNC} to the anti-holomorphic $\tilde{q}$-series component\footnote{The $q$-series giving the holomorphic block in the factorized expression in Eq.~\eqref{eq: P2fact} converge for $\hbar \rightarrow \infty$.}. Namely, 
\be \label{eq: expP2infty}
\mathrm{Tr}(\mrho_{\IP^2}) \sim \sqrt{\frac{2 \pi}{3^{5/2} \hbar}} \re^{-\frac{3 V}{4 \pi^2} \hbar} \exp \left( \sqrt{3} \sum_{n = 1}^{\infty} (-1)^{n-1} \frac{B_{2n} B_{2n-1}(2/3)}{(2n)! (2n-1)} \left( \frac{4 \pi^2}{\hbar} \right)^{2n-1} \right) \, ,
\ee
where $V=2 \Im(\mathsf{Li}_2(\re^{\pi \ri/3}))$.
Let us introduce the parameter 
\be \label{eq: tau}
\tau = -\frac{1}{\mb^2} = - \frac{2 \pi}{3\hbar}
\ee
and note that we can write $\tilde{q}= \re^{2 \pi \ri \tau}$ and $q = \re^{-2 \pi \ri /\tau}$.
We denote by $\psi(\tau)$ the formal power series appearing in the exponent in Eq.~\eqref{eq: expP2infty}, that is, 
\be \label{eq: phiP2infty}
\psi(\tau) = \sum_{n=1}^{\infty} b_{2n} \tau^{2n-1} \in \IQ[\pi, \sqrt{3}] [\![\tau]\!] \, , \quad b_{2n} = (-1)^{n} \sqrt{3} \frac{B_{2n} B_{2n-1}(2/3)}{(2n)! (2n-1)} (6 \pi)^{2n-1}  \quad n \ge 1 \, .
\ee
Its perturbative coefficients satisfy the factorial growth
\be
| b_{2n} | \le (2n)! \mathcal{A}_\infty^{-2n} \quad n \gg 1 \, , \quad \mathcal{A}_\infty = \frac{2 \pi}{3} = \frac{\mathcal{A}_0}{2 \pi} \, ,
\ee
and $\psi(\tau)$ is thus a Gevrey-1 asymptotic series. Analogously to the case of $\hbar \rightarrow 0$, the Borel transform $\hat{\psi}(\zeta)$ can be explicitly resummed into an exact function of $\zeta$ which is manifestly simple resurgent. Its singularities are simple poles located along the imaginary axis at
\be \label{eq: zetan-infty}
\eta_n = \mathcal{A}_\infty \ri n \, , \quad n \in \IZ_{\ne 0} \, ,
\ee
again tracing the only two Stokes lines at the angles $\pm \pi/2$.

The local expansion of the Borel transform $\hat{\psi}(\zeta)$ at $\zeta = \eta_n$, $n \in \IZ_{\ne 0}$, is given by
\be \label{eq: explocalP2infty}
\hat{\psi}(\zeta) = - \frac{R_n}{2 \pi \ri (\zeta - \eta_n)} + \text{regular in $\zeta-\eta_n$} \, ,
\ee
where $R_n$ is the corresponding Stokes constant. Once more, all the Stokes constants can be derived analytically in closed form. Indeed, the normalized Stokes constant $R_n/R_1$, where $R_1 = 3$, is given by the divisor sum function
\be \label{eq: formulaS1infty}
\frac{R_n}{R_1} = \sum_{\substack{d | n \\ d \equiv_3 1}} \frac{d}{n} - \sum_{\substack{d | n \\ d \equiv_3 2}} \frac{d}{n} = \sum_{\substack{d | n}} \frac{d}{n}\, \chi_{3,2}(d) \in \mathbb{Q}_{\ne 0} \, , \quad n \in \IZ_{\ne 0} \, ,
\ee
which is a multiplicative arithmetic function satisfying the property $R_{-n} = -R_{n}$. Note the strikingly simple arithmetic symmetry between the formulae in Eqs.~\eqref{eq: formulaS1} and~\eqref{eq: formulaS1infty}. 
Moreover, we can introduce a dual sequence of integers
\be \label{eq: intStokesP2infty}
\beta_n = \frac{R_n}{R_1} n \, , \quad n \in \IZ_{\ne 0} \, , 
\ee
which satisfies $\beta_n \in \IZ_{\ne 0}$ for $n \in \IZ_{> 0}$ and $\beta_{-n} = \beta_{n}$.
The pattern of singularities in the Borel plane and the associated $\beta_n$, $n \in \IZ_{\ne 0}$, are shown in Fig.~\ref{fig: peacockP2} on the right. 

\subsubsection*{Exact formulae for the discontinuities}
A notable consequence of the closed formulae in Eqs.~\eqref{eq: formulaS1} and~\eqref{eq: formulaS1infty} is that the Stokes constants $S_n, R_n$, $n \in \IZ_{>0}$, can be naturally organized as coefficients of exact generating functions given by the discontinuities of the corresponding asymptotic series $\phi(\hbar), \psi(\tau)$ across the positive imaginary axis. Explicitly, 
\begin{subequations} \label{eq: discontinuities-exact}
\begin{align}
\mathrm{disc}_{\frac{\pi}{2}} \, \phi(\hbar) &= \sum_{n=1}^{\infty} S_n \, \tilde{q}^n = - 3 \log \frac{(w ; \, \tilde{q})_{\infty}}{(w^{-1} ; \, \tilde{q})_{\infty}} -\ri \pi   \, ,  \label{eq: disc-weak}\\
\mathrm{disc}_{\frac{\pi}{2}} \, \psi(\tau) &= \sum_{n=1}^{\infty} R_n \, q^{n/3} = 3 \log \frac{(q^{2/3} ; \, q)_{\infty}}{(q^{1/3} ; \, q)_{\infty}} \, . \label{eq: disc-strong}
\end{align}
\end{subequations}
Observe that the quantum dilogarithms appearing in the generating function in Eq.~\eqref{eq: disc-weak} associated with the weak coupling asymptotic expansion $\phi(\hbar)$ are the same $\tilde{q}$-series that occur in the anti-holomorphic block of the spectral trace in Eq.~\eqref{eq: P2fact}. Similarly, the $q$-series in the generating function in Eq.~\eqref{eq: disc-strong} associated with the strong coupling asymptotic expansion $\psi(\tau)$ are those appearing in the dual holomorphic block of the spectral trace in Eq.~\eqref{eq: P2fact}\footnote{Disregarding constant prefactors, the $q, \tilde{q}$-series in Eq.~\eqref{eq: discontinuities-exact} differ from the logarithms of the holomorphic/anti-holomorphic blocks in Eq.~\eqref{eq: P2fact} by powers of two acting on the numerator/denominator, respectively.}. 

\begin{rmk}
The $q, \tilde{q}$-series in the factorized expression of the spectral trace of local $\IP^2$ encode the perturbative information in one asymptotic limit in $\hbar$, while being invisible to perturbation theory in the dual limit. Simultaneously, the same $q, \tilde{q}$-series supply the discontinuities in the opposite regimes.
Thus, the holomorphic block resurges from the asymptotics of the anti-holomorphic one and vice versa.
This two-way exchange of perturbative/non-perturbative content between the holomorphic/anti-holomorphic blocks is one of several exact cross-relations connecting the weak and strong coupling resurgent structures and paving the way towards the formulation of a global strong-weak resurgent symmetry, which we will describe in Section~\ref{sec: strong-weak}. 
\end{rmk}

\subsubsection*{A number-theoretic resurgent duality}
The standard large-$n$ asymptotics of the perturbative coefficients $a_{2n}, b_{2n}$, $n \in \IZ_{>0}$, in Eqs.~\eqref{eq: phiP2} and~\eqref{eq: phiP2infty} can be upgraded by systematically including the
contributions from all sub-dominant singularities in the Borel plane to give the exact large-order relations 
\begin{subequations} \label{eq: exactlarge}
\begin{align}
a_{2n} &= \frac{\Gamma(2n)}{\pi \ri \left(\CA_0 \ri \right)^{2n}} \sum_{m=1}^{\infty} \frac{S_m}{m^{2n}} \, , \label{eq: exactlarge-0} \\
b_{2n} &= \frac{\Gamma(2n-1)}{\pi \ri \left(\CA_{\infty} \ri \right)^{2n-1}} \sum_{m=1}^{\infty} \frac{R_m}{m^{2n-1}} \, , \label{eq: exactlarge-infty}
\end{align}
\end{subequations}
that is, up to the simple prefactors above, the perturbative coefficients of the series $\phi(\hbar), \psi(\tau)$ are given by the Dirichlet series encoding the corresponding weak and strong coupling Stokes constants evaluated at even and odd integer points, respectively. 

Amazingly, the Dirichlet series defined by the Stokes constants satisfy an Euler product expansion indexed by the set of prime numbers, are absolutely convergent in the right half-plane $\Re (s) >1$, and can be analytically continued to meromorphic functions on the whole complex $s$-plane. The resulting weak and strong coupling $L$-functions $L_0(s), L_{\infty}(s)$, $s \in \IC$, further admit the factorization
\begin{subequations} \label{eq: convolution2}
\begin{align}
L_0(s) &= \sum_{m=1}^{\infty} \frac{S_m}{m^{s}} =  S_1 L(s+1, \chi_{3,2}) \zeta(s) \, , \label{eq: convolution2-0} \\
L_{\infty}(s) &= \sum_{m=1}^{\infty} \frac{R_m}{m^{s}} = R_1 L(s, \chi_{3,2}) \zeta(s+1) \, , \label{eq: convolution2-infty}
\end{align}
\end{subequations}
where $L(s, \chi_{3,2})$ is the Dirichlet $L$-function of the primitive character $\chi_{3,2}$ and $\zeta(s)$ is the Riemann zeta function. 
Observe that the arithmetic duality\footnote{The symmetric decomposition of the Stokes constants into the Dirichlet convolution of simple arithmetic functions dictates the dual factorizations of the corresponding $L$-functions~\cite[Section 4.4]{Rella22}.} relating the weak and strong coupling Stokes constants $S_n, R_n$, $n \in \IZ_{> 0}$, translates at the level of the $L$-functions encoded in the perturbative coefficients into a symmetric unitary shift in the arguments of the factors in the RHS (right-hand-side) of Eq.~\eqref{eq: convolution2}. 
\begin{rmk}
If we extend the discrete index $n \in \IZ_{>0}$ of the sequences of perturbative coefficients $\{a_{2n}\}, \{b_{2n}\}$ to a continuous variable $s \in \IC$, then the exact large-order relations in Eq.~\eqref{eq: exactlarge} allow us, in principle, to analytically continue the perturbative coefficients to meromorphic functions throughout the complex $s$-plane as
\be
a_s = \frac{\Gamma(s)}{\pi \ri \left(\CA_0 \ri \right)^{s}} L_0(s) \, , \quad b_s = \frac{\Gamma(s-1)}{\pi \ri \left(\CA_{\infty} \ri \right)^{s-1}} L_{\infty}(s-1) \, .
\ee
A similar observation has been made in the recent work of~\cite{Vonk23} on the exact large-order relations for general resurgent trans-series.
\end{rmk}

\section{Resurgent \texorpdfstring{$L$}{L}-functions}\label{sec:1}
In this section, we complete the analytic number-theoretic duality of~\cite{Rella22}, which connects the exact resurgent structures of $\log \mathrm{Tr}(\mrho_{\IP^2})$ in the dual regimes of $\hbar \rightarrow 0$ and $\hbar \rightarrow \infty$, and upgrade it into a full-fledged exact \emph{strong-weak resurgent symmetry}. This unique global construction rests on the interplay of $q$-series and $L$-functions and revolves around the central role played by the weak and strong coupling Stokes constants. 

\subsection{Generating functions of the Stokes constants}\label{sec: generating-functs}
As it will be useful in the rest of this paper, let us now write the generating series of the weak coupling Stokes constants $S_n$, $n \in \IZ_{\ne 0}$, in Eq.~\eqref{eq: formulaS1} in the form 
\begin{equation}\label{eq:f_0}
    f_0(y):=\begin{cases}
    \displaystyle\sum_{n>0} S_n \, \re^{2\pi \ri n y} & \text{ if } \Im(y)>0 \\
    & \\
    -\displaystyle\sum_{n<0} S_n \, \re^{2\pi \ri n y} & \text{ if } \Im(y)<0
    \end{cases} \, ,
\end{equation}
which defines a holomorphic function on $\IC \setminus \IR$ with the periodicity and parity properties
\be\label{eq:f_0-symm}
f_0(y+1)=f_0(y) \, , \quad f_0(-y) = -f_0(y) \, .
\ee 
It follows that $f_0(y)$, $y \in \IC \setminus \IR$, need only be specified in the upper half of the complex $y$-plane, where it is known in closed form by means of the exact discontinuity formula in Eq.~\eqref{eq: disc-weak}. Namely, we have that~\cite[Corollary 4.8]{Rella22}
\be \label{eq:f_0-closed}
f_0(y) = - 3 \log \frac{(w ; \, \re^{2\pi \ri y})_{\infty}}{(w^{-1} ; \, \re^{2\pi \ri y})_{\infty}} -\ri \pi \, , \quad y \in \IH \, ,
\ee
where $w = \re^{2 \pi \ri /3}$ as before and therefore also
\be \label{eq: f0disc}
f_0\left(-\frac{2 \pi}{3 \hbar} \right) = \mathrm{disc}_{\frac{\pi}{2}} \, \phi(\hbar) \, , \quad \hbar \in \IH \, .
\ee
Analogously, we write the generating series of the strong coupling Stokes constants $R_n$, $n \in \IZ_{\ne 0}$, in Eq.~\eqref{eq: formulaS1infty} in the form 
\begin{equation}\label{eq:f_inf}
    f_\infty(y):=\begin{cases}
    \displaystyle\sum_{n>0} R_n \, \re^{2\pi \ri n y} & \text{ if } \Im(y)>0 \\
    & \\
    -\displaystyle\sum_{n<0} R_n \, \re^{2\pi \ri n y} & \text{ if } \Im(y)<0
    \end{cases} \, , 
\end{equation}
which defines a holomorphic function on $\IC \setminus \IR$ with the periodicity and parity properties\footnote{Note the exchange of even/odd parity between the Stokes constants and the corresponding generating functions in both weak and strong coupling limits.}
\be\label{eq:f_inf-symm}
f_\infty(y+1)=f_\infty(y) \, , \quad f_\infty(-y) = f_\infty(y) \, . 
\ee
Again, it follows that $f_\infty(y)$, $y \in \IC \setminus \IR$, need only be specified in the upper half of the complex $y$-plane, where it is known in closed form through the exact discontinuity formula in Eq.~\eqref{eq: disc-strong}. Explicitly, we have that~\cite[Corollary 4.19]{Rella22}
\be \label{eq:f_inf-closed}
f_\infty(y/3) = 3 \log \frac{(\re^{4\pi \ri y/3} ; \, \re^{2\pi \ri y})_{\infty}}{(\re^{2\pi \ri y/3} ; \, \re^{2\pi \ri y})_{\infty}} \, , \quad y \in \IH \, ,
\ee
and therefore also
\be \label{eq: finftydisc}
f_\infty\left(-\frac{1}{3 \tau} \right) = \mathrm{disc}_{\frac{\pi}{2}} \, \psi(\tau) \, , \quad \tau \in \IH \, .
\ee
\begin{rmk}
Let us now consider the discontinuities $\mathrm{disc}_{\frac{\pi}{2}}\, \phi(\hbar)$ and $\mathrm{disc}_{\frac{\pi}{2}}\, \psi(\tau)$ as functions of the variable $\tau \in \IH$ using Eqs.~\eqref{eq: disc-weak},~\eqref{eq: disc-strong}, and~\eqref{eq: tau}. Namely, we take the functions
    \be
        \left(\mathrm{disc}_{\frac{\pi}{2}}\, \phi\right)(\tau) = - 3 \log \frac{(w ; \, \tilde{q})_{\infty}}{(w^{-1} ; \, \tilde{q})_{\infty}} -\ri \pi \, , \quad 
        \left(\mathrm{disc}_{\frac{\pi}{2}}\, \psi\right)(\tau) = 3 \log \frac{(q^{2/3} ; \, q)_{\infty}}{(q^{1/3} ; \, q)_{\infty}} \, ,
    \ee
    where $\tilde{q}= \re^{2 \pi \ri \tau}$ and $q = \re^{-2 \pi \ri /\tau}$ as before. It follows that
    \be
    \left(\mathrm{disc}_{\frac{\pi}{2}}\, \phi\right)(\tau+1) = \left(\mathrm{disc}_{\frac{\pi}{2}}\, \phi\right)(\tau) \, , \quad \left(\mathrm{disc}_{\frac{\pi}{2}}\, \psi\right)\left(\frac{\tau}{3\tau+1}\right) = \left(\mathrm{disc}_{\frac{\pi}{2}}\, \psi\right)(\tau)\, ,
    \ee
    that is, the discontinuities are invariant under the action of the generators $T: \, \tau \mapsto \tau + 1$ and $\gamma_3: \, \tau \mapsto \tau/(3 \tau + 1)$ of the congruence subgroup $\Gamma_1(3)$, respectively. Recall that we have introduced $\Gamma_1(3)$ in Eq.~\eqref{eq: G13-definition} as it dictates the symmetries of the moduli space of complex structures of the mirror of local $\IP^2$. The same $\Gamma_1(3)$-structure crucially occurs in Section~\ref{sec:quantum-resurgent}, where the modular considerations above are upgraded to full quantum modularity statements on the generating functions of the Stokes constants $f_0(y)$ and $f_\infty(y)$, $y\in\IH$. 
\end{rmk}

\subsubsection*{Asymptotic series for the generating functions} 
Let us compute the all-orders perturbative expansions in the limit $y \rightarrow 0$ of the weak and strong coupling generating functions $f_0(y)$, $f_\infty(y)$, $y \in \IH$, in Eqs.~\eqref{eq:f_0-closed} and~\eqref{eq:f_inf-closed}, which we denote by $\tilde{f}_0(y)$, $\tilde{f}_\infty(y)$, respectively.
We apply the known asymptotic expansion formula for the quantum dilogarithm in Eq.~\eqref{eq: logPhiNC} to the closed expression for the generating function of the weak coupling Stokes constants in Eq.~\eqref{eq:f_0-closed} and explicitly evaluate the special functions that appear.
We find that 
\be \label{eq:asymp-f0}
\tilde{f}_0(y) = -\frac{\pi \ri}{2} - \frac{3}{2 \pi \ri y} \left(\mathrm{Li}_2(w) - \mathrm{Li}_2(w^{-1}) \right) - 2 \sqrt{3} \ri \sum_{n=1}^\infty (6\pi \ri y)^{2n-1}\frac{B_{2n} B_{2n-1}(2/3)}{(2n)!(2n-1)}\, ,
\ee
which simplifies into
\be \label{eq: f0-psi}
\tilde{f}_0(y) = -\frac{\pi \ri}{2} - \frac{3 \mathcal{V}}{2 \pi \ri y} - 2 \psi(y) \, , 
\ee
where\footnote{Observe that $\frac{27}{8 \pi^2} \CV \simeq 0.462758$ is the quantum volume of local $\IP^2$, which is given by the value of its K\"ahler parameter at the conifold~\cite{CKYZ}. This is related to the constant $V$ in Eq.~\eqref{eq: expP2infty} by $\mathcal{V}=-2 \pi^2/3+ 2 \ri V/3$.} $\mathcal{V}= 2 \Im\left(\mathrm{Li}_2(w) \right)$ and $\psi(y)$ is the asymptotic series occurring in the strong coupling perturbative expansion of $\mathrm{Tr}(\mrho_{\IP^2})$ in Eq.~\eqref{eq: phiP2infty}.
Similarly, we apply the known asymptotic expansion formula for the quantum dilogarithm in Eq.~\eqref{eq: logPhiK}, with the choice of $\alpha=1/3, 2/3$, to the closed expression for the generating function of the strong coupling Stokes constants in Eq.~\eqref{eq:f_inf-closed} and explicitly evaluate the special functions that appear. We derive in this way that
\be \label{eq:asymp-finf}
\tilde{f}_\infty(y) = - 3 \log \frac{\Gamma(2/3)}{\Gamma(1/3)} -\log(- 6 \pi \ri y) -6 \sum_{n=1}^{\infty} (6 \pi \ri y)^{2n} \frac{B_{2n} B_{2n+1}(2/3)}{2n (2n+1)!} \, ,
\ee
which equals to
\be \label{eq: finf-phi}
\tilde{f}_\infty(y) = - 3 \log \frac{\Gamma(2/3)}{\Gamma(1/3)} -\log(- 6 \pi \ri y) +2 \phi(2 \pi y) \, ,
\ee
where $\phi(y)$ is the asymptotic series occurring in the weak coupling perturbative expansion of $\mathrm{Tr}(\mrho_{\IP^2})$ in Eq.~\eqref{eq: phiP2}. As a consequence of Eqs.~\eqref{eq: f0-psi} and~\eqref{eq: finf-phi}, the exact resurgent structures of the asymptotic series $\tilde{f}_0(y)$ and $\tilde{f}_\infty(y)$ follow straightforwardly from the work of~\cite{Rella22} on the perturbative expansions $\psi(\tau)$ and $\phi(\hbar)$, respectively, which we have summarized in Section~\ref{sec:exact-solution}.

\begin{rmk}
    Putting together Eqs.~\eqref{eq: f0disc}  and~\eqref{eq: finftydisc} with Eqs.~\eqref{eq: f0-psi} and~\eqref{eq: finf-phi}, we obtain the asymptotic expansions of the discontinuities of the original perturbative series $\phi(\hbar)$ and $\psi(\tau)$ in the large limits of the variables $\hbar$ and $\tau$, respectively. Namely, we have that
    \begin{subequations}
    \begin{align}
        \mathrm{disc}_{\frac{\pi}{2}} \, \phi(\hbar) &\sim -\frac{\pi \ri}{2} - \frac{9 \mathcal{V}}{4 \pi^2 \ri} \hbar - 2 \psi\left(-\frac{2 \pi}{3 \hbar}\right) \, , \quad \hbar \rightarrow \infty \, , \label{eq: asym-disc-phi} \\
        \mathrm{disc}_{\frac{\pi}{2}} \, \psi(\tau) &\sim - 3 \log \frac{\Gamma(2/3)}{\Gamma(1/3)} -\log\left(\frac{2 \pi \ri}{\tau}\right) +2 \phi\left(-\frac{2 \pi}{3\tau}\right) \, , \quad \tau \rightarrow \infty \, . \label{eq: asym-disc-psi}
    \end{align}
    \end{subequations}
\end{rmk}

\subsection{Strong-weak resurgent symmetry}\label{sec: strong-weak}
As we described in Section~\ref{sec:exact-solution}, summarizing the work of~\cite{Rella22}, and in Section~\ref{sec: generating-functs}, the resurgent structures of the logarithm of the spectral trace of local $\IP^2$ in both weak and strong coupling limits display a rich analytic number-theoretic fabric centered around the key role played by the Stokes constants and their generating functions. The interaction of $q$-series and $L$-functions governs a network of relations encompassing and intertwining the regimes of $\hbar \rightarrow 0$ and $\hbar \rightarrow \infty$.
Let us give a concise overview.  
\begin{itemize}
    \item The discontinuity $\mathrm{disc}_{\frac{\pi}{2}}\psi(\tau)$, which equals the generating function $f_\infty\big(-\tfrac{1}{3\tau}\big)$ in Eq.~\eqref{eq:f_inf-closed} of the strong coupling Stokes constants $R_n$, $n \in \IZ_{\ne 0}$, in Eq.~\eqref{eq: formulaS1infty}, reproduces the weak coupling perturbative series $\phi(\hbar)$ in Eq.~\eqref{eq: phiP2} when expanded in the limit $\tau \rightarrow \infty$ (or equivalently $\hbar \rightarrow 0$), as we have seen in Eq.~\eqref{eq: asym-disc-psi}. 
    \item Resurgence associates $\phi(\hbar)$ with the collection $\{\zeta_n, \, S_n\}$, $n \in \IZ_{\ne 0}$, of simple poles and Stokes constants, which satisfy $\zeta_n=\CA_0 \ri n$ and $S_{-n}=S_n$.
    \item The exact large-$n$ relations in Eq.~\eqref{eq: exactlarge-0} for the perturbative coefficients $a_{2n}$, $n \in \IZ_{>0}$, of $\phi(\hbar)$ uniquely determine the weak coupling $L$-function $L_0(s)$ in Eq.~\eqref{eq: convolution2-0}. The latter reproduces, in turn, the perturbative coefficients when evaluated at $s = 2n$.
    \item The discontinuity $\mathrm{disc}_{\frac{\pi}{2}}\phi(\hbar)$, which equals the generating function $ f_0 \big(-\tfrac{2 \pi}{3 \hbar} \big)$ in Eq.~\eqref{eq:f_0-closed} of the weak coupling Stokes constants $S_n$, $n \in \IZ_{\ne 0}$, in Eq.~\eqref{eq: formulaS1}, reproduces the strong coupling perturbative series $\psi(\tau)$ in Eq.~\eqref{eq: phiP2infty} when expanded in the limit $\hbar \rightarrow \infty$ (or equivalently $\tau \rightarrow 0$), as we have seen in Eq.~\eqref{eq: asym-disc-phi}.
    \item Resurgence associates $\psi(\tau)$ with the collection $\{\eta_n, \, R_n\}$, $n \in \IZ_{\ne 0}$, of simple poles and Stokes constants, which satisfy $\eta_n=\CA_\infty \ri n$ and $R_{-n}=-R_n$.
    \item The exact large-$n$ relations in Eq.~\eqref{eq: exactlarge-infty} for the perturbative coefficients $b_{2n}$, $n \in \IZ_{>0}$, of $\psi(\tau)$ uniquely determine the strong coupling $L$-function $L_\infty(s)$ in Eq.~\eqref{eq: convolution2-infty}. The latter reproduces, in turn, the perturbative coefficients when evaluated at $s = 2n-1$.
\end{itemize}
Assembling the above relations together, we find that the resurgent structures of the dual perturbative expansions of $\log \mathrm{Tr}(\mrho_{\IP^2})$ embed into a unique global construction as they compose the symmetric diagram below.
\begin{equation} \label{diag: strong-weak1}
\begin{tikzcd}[column sep=2.8em, row sep=2.8em]
& & \arrow[dd,sloped, shift left=.5ex,"\text{Mellin}"] \mathrm{disc}_{\frac{\pi}{2}} \, \psi(\tau) \arrow[rr, "\hbar \rightarrow 0"] & & \phi(\hbar) \arrow[dd,sloped,above,shift left=.5ex, "\text{exact large-$n$}"]\arrow[ddrr,sloped,"\text{resurgence}"] & & 
\\ \\ 
\arrow[sloped]{uurr}{\text{generating}}[swap]{\text{series}} \{\eta_n, \, R_n\} \arrow{rr}{\text{Dirichlet}}[swap]{\text{series}} & & \{\CA_\infty, \, L_\infty\} \arrow[dd,sloped,above,shift left=.5ex, "\text{evaluation}"] \arrow[uu,sloped,shift left=.5ex, "\text{inverse Mellin}"] & &  \{\CA_0, \, L_0\}\arrow[uu,sloped,above,shift left=.5ex, "\text{evaluation}"] \arrow[dd,sloped,shift left=.5ex, "\text{inverse Mellin}"] & & \arrow{ll}{\text{series}}[swap]{\text{Dirichlet}} \{\zeta_n, \, S_n\} \arrow[sloped]{ddll}{\text{generating}}[swap]{\text{series}} 
\\ \\
& & \psi(\tau)\arrow[uu,sloped,above,shift left=.5ex, "\text{exact large-$n$}"] \arrow[uull,sloped,"\text{resurgence}"] & & \arrow[uu,sloped, shift left=.5ex,"\text{Mellin}"] \mathrm{disc}_{\frac{\pi}{2}} \, \phi(\hbar) \arrow[ll,swap, "\tau \rightarrow 0"]
 & &
\end{tikzcd}
\end{equation}

Notably, starting from any fixed vertex, all other vertices of this commutative diagram are spanned by following its directed arrows---that is, the information contained in each vertex reconstructs the whole diagram and is therefore equivalent to the information content of every other vertex.
Moreover, two additional arrows are drawn in the diagram in Eq.~\eqref{diag: strong-weak1} that were not previously stated. They originate from the well-established relation between power series and Dirichlet series sharing the same coefficients, which makes use of the Mellin transform.
\begin{lemma}\label{lemma:mellin}
    The weak and strong coupling $L$-functions $L_0(s)$, $L_\infty(s)$, $s \in \IC$, in Eqs.~\eqref{eq: convolution2-0} and~\eqref{eq: convolution2-infty} are given by the Mellin transform of the generating functions $f_0(y)$, $f_\infty(y)$, $y \in \IH$, of the corresponding Stokes constants, that is,
    \begin{subequations}
    \begin{align}
    L_0(s) &= \frac{(2 \pi)^{s}}{\Gamma(s)} \int_0^{\infty} t^{s-1} f_0(\ri t) \, d t  \, , \label{eq: Mellin-0} \\
    L_\infty(s) &= \frac{(2 \pi)^{s}}{\Gamma(s)} \int_0^{\infty} t^{s-1} f_\infty(\ri t)\,  d t \, . \label{eq: Mellin-inf}
    \end{align}
    \end{subequations}
\end{lemma}
\begin{proof}
    Let us start with the weak coupling regime of $\hbar \rightarrow 0$ and explicitly compute the Mellin transform of the generating function of the Stokes constants in Eq.~\eqref{eq:f_0}. We find that
    \be \label{eq: MT-0-1}
    \int_0^{\infty} t^{s-1} f_0(\ri t) \, d t = \int_0^{\infty} t^{s-1} \left( \sum_{k=1}^{\infty} S_k \re^{-k 2 \pi t}\right) \, d t = \sum_{k=1}^{\infty} S_k \int_0^{\infty} t^{s-1} \re^{-k 2 \pi t} \, d t  \, ,
    \ee
    where we have permuted sum and integration due to absolute convergence. The simplified integral in the RHS is then easily computed via the well-known formula
    \be \label{eq: MT-0-2}
    \int_0^{\infty} x^{s-1} \re^{-k x} \, d x = \frac{\Gamma(s)}{k^{s}} \, , \quad k \in \IZ_{>0} \, , \quad s \in \IC \, ,
    \ee
    which gives the desired expression in Eq.~\eqref{eq: Mellin-0}.
    The same procedure above is straightforwardly applied in the strong coupling regime of $\hbar \rightarrow \infty$ to explicitly compute the Mellin transform of the generating function of the Stokes constants in Eq.~\eqref{eq:f_inf}, yielding the expression in Eq.~\eqref{eq: Mellin-inf}.  
\end{proof}
By the Mellin inversion theorem, we recover the generating functions of the Stokes constants via the inverse Mellin transform of the corresponding $L$-functions.
Composing Lemma~\ref{lemma:mellin} with the exact large-order relations, we find a direct way of obtaining the perturbative expansion of the generating function in one asymptotic limit from the generating function in the dual limit, which amounts to performing an operation that is the formal inverse of taking the discontinuity of the perturbative series.  
\begin{prop}\label{prop:inv_Mellin}
    The perturbative coefficients $a_{2n}$, $b_{2n}$, $n \in \IZ_{>0}$, in Eqs.~\eqref{eq: phiP2} and~\eqref{eq: phiP2infty} are given by the Mellin transform of the discontinuities $\mathrm{disc}_{\frac{\pi}{2}} \, \phi(\hbar)$, $\mathrm{disc}_{\frac{\pi}{2}} \, \psi(\tau)$, respectively, that is,
     \begin{subequations}
    \begin{align}
    \label{eq: inv_disc_0} a_{2n} &= \frac{(-1)^n}{\pi \rm{i}} \int_0^{\infty} \hbar^{-2n -1} \mathrm{disc}_{\frac{\pi}{2}} \phi(- \ri \hbar) \, d \hbar \, , \\
    \label{eq: inv_disc_inf} b_{2n} &= \frac{(-1)^n}{\pi} \int_0^{\infty} \tau^{-2n} \mathrm{disc}_{\frac{\pi}{2}} \psi(- \ri \tau) \, d \tau\,.
    \end{align}
    \end{subequations}
\end{prop}
\begin{proof}
    Using Eq.~\eqref{eq: f0disc} and applying the change of variable $2 \pi t=  \CA_0/\hbar$, the integral in the RHS of Eq.~\eqref{eq: inv_disc_0} becomes
    \be 
    \int_0^{\infty} \hbar^{-2n -1} f_0\left(\frac{\CA_0 \ri}{2 \pi \hbar} \right) \, d \hbar = \frac{(2 \pi)^{2n}}{\CA_0^{2n}} \int_0^{\infty} t^{2n-1} f_0(\ri t) \, d t \, .
    \ee
    Using Eq.~\eqref{eq: Mellin-0} for $s=2n$ and the exact large-$n$ relations for the weak coupling perturbative coefficients in Eq.~\eqref{eq: exactlarge-0}, we conclude. Similarly, using Eq.~\eqref{eq: finftydisc} and applying the change of variable $2 \pi t = \CA_{\infty} /\tau$, the integral in the RHS of Eq.~\eqref{eq: inv_disc_inf} becomes
    \be 
    \int_0^{\infty} \tau^{-2n} f_\infty\left(\frac{\CA_\infty \ri}{2 \pi \tau} \right) \, d \tau = \frac{(2 \pi)^{2n-1}}{\CA_\infty^{2n-1}} \int_0^{\infty} t^{2n-2} f_\infty(\ri t) \, d t \, .
    \ee
    Using Eq.~\eqref{eq: Mellin-inf} for $s=2n-1$ and the exact large-$n$ relations for the strong coupling perturbative coefficients in Eq.~\eqref{eq: exactlarge-infty} yields the desired expression.
\end{proof}

Consider the commutative diagram in Eq.~\eqref{diag: strong-weak1} once more. The internal vertices and arrows composing the left-side sub-diagram that connects the $q$-series $\mathrm{disc}_{\frac{\pi}{2}}\psi(\tau)$ with the strong coupling perturbative series $\psi(\tau)$ reduce to one vertical two-headed arrow according to the formulae in Eqs.~\eqref{eq: disc-strong} and~\eqref{eq: inv_disc_inf}. Analogously, the set of internal vertices and arrows composing the right-side sub-diagram that connects the $\tilde{q}$-series $\mathrm{disc}_{\frac{\pi}{2}}\phi(\hbar)$ with the weak coupling perturbative series $\phi(\hbar)$ is equivalent to one vertical two-headed arrow according to the dual formulae in Eqs.~\eqref{eq: disc-weak} and~\eqref{eq: inv_disc_0}.
Thus, the two-way exchange of perturbative/non-perturbative information between the dual regimes in $\hbar$ takes the form of a mathematically precise mechanism, which we refer to as \emph{strong-weak resurgent symmetry}.
We represent it schematically in the box diagram below, where we stress the contribution of the Stokes constants to all quantities of interest. 
\begin{equation} \label{diag: strong-weak2}
\begin{tikzcd}[column sep = 2.8em, row sep=2.8em]
\arrow[ddd,red] \mathrm{disc}_{\frac{\pi}{2}}\psi(\tau)= \sum\limits_{n=1}^{\infty}R_n q^{n/3} \arrow[r, "\hbar \rightarrow 0"] & \phi(\hbar)=\sum\limits_{n=1}^{\infty} \left(\frac{\Gamma(2n)}{\pi \ri \left(\CA_0 \ri \right)^{2n}} \sum\limits_{m=1}^{\infty} \frac{S_m}{m^{2n}} \right) \hbar^{2n} \arrow[red,sloped]{ddd}{\text{strong-weak}}[swap]{\text{symmetry}}
\\ \\ \\
 \arrow[red,sloped]{uuu}{\text{strong-weak}}[swap]{\text{symmetry}} \psi(\tau)=\sum\limits_{n=1}^{\infty} \left(\frac{\Gamma(2n-1)}{\pi \ri \left(\CA_\infty \ri \right)^{2n-1}} \sum\limits_{m=1}^{\infty} \frac{R_m}{m^{2n-1}} \right) \tau^{2n-1} & \arrow[uuu,red]  \mathrm{disc}_{\frac{\pi}{2}}\phi(\hbar)=\sum\limits_{n=1}^{\infty}S_n \tilde{q}^{n} \arrow[l,swap, "\tau \rightarrow 0"]
\end{tikzcd}
\end{equation}
\begin{rmk}
The strong-weak resurgent symmetry of the spectral trace of local $\IP^2$, which we have presented here, fundamentally arises from the interplay of the Stokes constants, their generating functions in the form of $q$, $\tilde{q}$-series, and their $L$-functions. This is the starting point in a wider mathematical program that aims at linking the resurgent properties of $q$-series with the analytic number-theoretic properties of $L$-functions and makes contact with the notion of quantum modularity. We will present this new paradigm of resurgence in detail and full generality in the companion paper~\cite{FR1maths}.
\end{rmk}

As pointed out in~\cite{GuM, MR}, the TS/ST correspondence for a toric CY threefold $X$ shares many formal similarities with complex Chern--Simons theory on the complement of a hyperbolic knot $\CK$ in the three-sphere. In particular, the spectral determinant $\Xi_X(\kappa, \hbar)$, which is an entire function of the complex deformation parameter $\kappa$ of the mirror $\hat{X}$, corresponds to the state integral, or Andersen--Kashaev invariant, of the knot $Z_\CK(u, \mb^2)$~\cite{DGLZ, AK}, which is an entire function of the holonomy $u$ around the knot. Here, $\mb^2$ is a complex coupling parameter. Thus, the spectral traces $Z_X(N, \hbar)$ are dual to the coefficients of $Z_\CK(u, \mb^2)$ in a Taylor expansion around $u = 0$. Importantly, the state integral can be expressed as a sum of products of holomorphic and anti-holomorphic blocks given by $q$, $\tilde{q}$-series~\cite{Beem:2012mb,Dimofte:2014zga}, where $q$ and $\tilde{q}$ are related to $\mb^2$ by the same formulae in Eq.~\eqref{eq: q-var}, and the blocks are exchanged under the transformation $\mb^2 \mapsto -\mb^{-2}$, which swaps $q$ and $\tilde{q}$. 
In the TS/ST correspondence, a similar factorization property is observed in closed expressions for the fermionic spectral traces $Z_X(N, \hbar)$, although the holomorphic and anti-holomorphic blocks are generally different functions~\cite{GuM}. Yet, the resurgent study of the spectral trace of local $\IP^2$ shows how such different functions are nonetheless deeply connected. Indeed, as we have already remarked in Section~\ref{sec:exact-solution}, the holomorphic block resurges from the perturbative expansion of the anti-holomorphic block, and vice versa. This statement assumes a concrete meaning in light of the proven, exact results on the generating functions of the weak and strong coupling Stokes constants assembled into the strong-weak resurgent symmetry. To sum up, we display a minimal scheme below.
\begin{equation} \label{diag: strong-weak3}
\begin{tikzcd}[column sep = 2.8em, row sep=2.8em]
\mathrm{disc}_{\frac{\pi}{2}}\psi(\tau)= \sum_{n=1}^{\infty}R_n q^{n/3} \arrow[rrrr,red, shift left=.5ex]  & & & & \mathrm{disc}_{\frac{\pi}{2}}\phi(\hbar)=\sum_{n=1}^{\infty}S_n \tilde{q}^{n} \arrow[llll, red, shift left=.5ex]
\end{tikzcd}
\end{equation}
Here, the red arrows represent the diagonals joining the top-left and bottom-right vertices of the diagram in Eq.~\eqref{diag: strong-weak2}.

\subsubsection*{On the physics of the strong-weak resurgent symmetry}\label{sec: physics_arg}
The strong-weak coupling duality in Eq.~\eqref{eq: duality}, which is at the heart of the TS/ST correspondence, allows us to describe the strongly coupled conventional topological string compactified on local $\IP^2$ in terms of the semiclassical regime of the quantum-mechanical spectral problem defined by the quantization of its mirror curve. Analogously, the semiclassical limit in the topological string coupling constant determines the strong dynamics of the spectral theory. Namely, the strong regime of each theory is approachable through the weak regime of the other. We point out that our strong-weak symmetry relating the exact resurgent structures of the spectral trace of local $\IP^2$ in the limits of $\hbar \rightarrow 0$ and $\hbar \rightarrow \infty$ is a mathematically precise realization of this intuition on the exchange of perturbative and non-perturbative contributions. 

The discussion on the physical mechanism underlying the strong-weak resurgent symmetry can be pushed further. Indeed, recall the definition of the total grand potential $J(\mu, \hbar)$ in Eq.~\eqref{eq: J_total}. As we have briefly described in Section~\ref{sec:background}, the worldsheet generating functional $J^{\rm WS}(\mu, \hbar)$, which determines the perturbative expansion of $J(\mu, \hbar)$ in the weakly coupled limit $g_s \rightarrow 0$, encodes the non-perturbative contributions in $\hbar$ from the conventional topological string partition function. At the same time, the WKB generating functional $J^{\rm WKB}(\mu, \hbar)$, which determines the perturbative expansion of $J(\mu, \hbar)$ in the semiclassical limit $\hbar \rightarrow 0$ instead, contains the non-perturbative, exponentially small effects in $g_s$ from the NS limit of the refined topological string partition function. 
By means of the TS/ST statement in Eq.~\eqref{eq: contour} and the strong-weak coupling duality in Eq.~\eqref{eq: duality}, we can translate the above argument into the following. The non-perturbative $\hbar$-corrections to the semiclassical perturbative expansion of the fermionic spectral trace $Z(N, \hbar)$, $N \in \IZ_{>0}$, in Eq.~\eqref{eq: pert_exp_0}, determine its perturbative expansion in the strong coupling regime $\hbar \rightarrow \infty$. Conversely, the non-perturbative $\hbar^{-1}$-corrections to the strongly coupled perturbative expansion of the fermionic spectral trace $Z(N, \hbar)$, $N \in \IZ_{>0}$, in Eq.~\eqref{eq: pert_exp_inf}, determine its semiclassical perturbative $\hbar$-expansion. 

Take now $N=1$ and consider the first fermionic spectral trace $Z(1, \hbar)$, as we have done in this paper. We know that the non-analytic, exponential-type contributions at strong and weak coupling are governed by two dual sequences of Stokes constants $S_n$, $R_n$, $n \in \IZ_{\ne 0}$. Equivalently, we can repackage the same information in the corresponding generating functions $f_0(y)$, $f_\infty(y)$, $y \in \IC \setminus \IR$, or the corresponding $L$-functions $L_0(s)$, $L_{\infty}(s)$, $s \in \IC$. 
Thus, the argument above implies that the Stokes constants/generating functions/$L$-functions in one regime must dictate the perturbative coefficients in the other, which we have proven true. Yet, substantially more information is discovered by analytically determining the complete resurgent structures of the spectral trace of local $\IP^2$ at strong and weak coupling, as we have summarized in the commutative diagram in Eq.~\eqref{diag: strong-weak1} and will further detail in the rest of this work.
Finally, we note that our results recall the traditional notion of S-duality in string theory~\cite{Sduality}\footnote{In the work of~\cite{Sduality}, it was shown that the S-duality of type IIB superstrings in ten dimensions implies the existence of an S-duality relating the A-model and B-model topological string theories on the same CY background. In particular, the D-instantons of one model correspond to the perturbative amplitudes of the other.}.

\begin{rmk}
 As the ideas above apply in full generality, our conceptual argument underlying the strong-weak resurgent symmetry of the spectral trace of local $\IP^2$ supports the existence of an appropriate generalization of our results to other toric CY threefolds and higher-order fermionic spectral traces. The detailed exploration of the physical principles of the strong-weak resurgent symmetry will be the subject of future work. 
\end{rmk}

\subsection{Functional equation from a unified perspective}
Let us go back to the $L$-functions $L_0(s)$, $L_\infty(s)$, $s \in \IC$, which are defined by the meromorphic continuation to the complex $s$-plane of the Dirichlet series with coefficients given by the weak and strong coupling Stokes constants $S_n$, $R_n$, $n \in \IZ_{>0}$, respectively, as we have briefly recalled in Section~\ref{sec:exact-solution}. These weak and strong coupling $L$-functions satisfy the factorizations in Eqs.~\eqref{eq: convolution2-0} and~\eqref{eq: convolution2-infty} in terms of the Dirichlet $L$-function $L(s, \chi_{3,2})$ and the Riemann zeta function $\zeta(s)$, whose meromorphic completions
\be \label{eq: completed}
\Lambda_{3,2}(s) = \frac{3^{\frac{s}{2}}}{\pi^{\frac{s+1}{2}}} \Gamma\left( \frac{s+1}{2} \right) L(s, \chi_{3,2}) \, , \quad \Lambda_\zeta(s) = \frac{1}{\pi^{\frac{s}{2}}} \Gamma\left( \frac{s}{2} \right) \zeta(s) \, , \quad s \in \IC \, ,
\ee
satisfy the well-known functional equations
\be \label{eq: funct-eqs}
\Lambda_{3,2}(s) = \Lambda_{3,2}(1-s) \, , \quad \Lambda_\zeta(s) = \Lambda_\zeta(1-s) \, , 
\ee
which are centered at $s=1/2$ and relate $s$ with $1-s$.
Following the factorization formulae in Eqs.~\eqref{eq: convolution2-0} and~\eqref{eq: convolution2-infty} and using the completed $L$-functions in Eq.~\eqref{eq: completed}, we define
\begin{subequations}
\begin{align}
    \Lambda_0(s) &= \Lambda_{3,2}(s+1) \Lambda_\zeta(s) = \frac{3^{\frac{s+1}{2}}}{S_1 \pi^{s+1}} \Gamma\left( \frac{s}{2} \right) \Gamma\left( \frac{s}{2}+1 \right) L_0(s) \, , \label{eq: defLambda0} \\
    \Lambda_\infty(s) &= \Lambda_{3,2}(s) \Lambda_\zeta(s+1) = \frac{3^{\frac{s}{2}}}{R_1 \pi^{s+1}} \Gamma\left( \frac{s+1}{2} \right)^2 L_\infty(s) \, , \label{eq: defLambdainfty} 
\end{align}
\end{subequations}
which are meromorphic functions for $s \in \IC$. Observe, however, that they do not individually satisfy a functional equation of the form in Eq.~\eqref{eq: funct-eqs}. Indeed, these completions of the weak and strong coupling $L$-functions are analytically continued to the whole complex $s$-plane through each other. The following statement is a consequence of the remarkable symmetry that relates the factorizations of $L_0(s)$ and $L_\infty(s)$ via a unitary cross-shift in the arguments of the factors.
\begin{theorem} \label{th: duality}
The weak and strong coupling completed $L$-functions $\Lambda_0(s)$ and $\Lambda_{\infty}(s)$, $s \in \IC$, in Eqs.~\eqref{eq: defLambda0} and~\eqref{eq: defLambdainfty}, respectively, satisfy the functional equation
\begin{equation} \label{eq: dualityLambda}
     \Lambda_0(s) = \Lambda_{\infty}(-s) \, .
\end{equation}
\end{theorem}
\begin{proof}
Using the definitions in Eqs.~\eqref{eq: defLambda0} and~\eqref{eq: defLambdainfty} and the functional equations in Eq.~\eqref{eq: funct-eqs}, the statement follows from the explicit computation
\be
\Lambda_0(1-s)=\Lambda_{3,2}(2-s) \Lambda_\zeta(1-s) = \Lambda_{3,2}(s-1) \Lambda_\zeta(s) = \Lambda_\infty(s-1) \, ,
\ee
after the change of variable $s \rightarrow 1-s$.
\end{proof}
As the weak and strong coupling $L$-functions are naturally paired, their combined functional equation in Eq.~\eqref{eq: dualityLambda} provides an additional two-headed arrow joining them directly in the global commutative diagram in Eq.~\eqref{diag: strong-weak1}, which we explicitly reproduce below.\footnote{For completeness, we remind that the fundamental constants $\CA_0$, $\CA_\infty$ are related by $\CA_0 = 2 \pi \CA_\infty$.}
\begin{equation} \label{diag: strong-weak-L}
\begin{tikzcd}[column sep=2.8em, row sep=2.8em]
& & \arrow[dd,sloped, shift left=.5ex,"\text{Mellin}"] \mathrm{disc}_{\frac{\pi}{2}}\psi(\tau) \arrow[rr, "\hbar \rightarrow 0"] & & \phi(\hbar) \arrow[dd,sloped,above,shift left=.5ex, "\text{exact large-$n$}"]\arrow[ddrr,sloped,"\text{resurgence}"] & & 
\\ \\ 
\arrow[sloped]{uurr}{\text{generating}}[swap]{\text{series}} \{\eta_n, \, R_n\} \arrow{rr}{\text{Dirichlet}}[swap]{\text{series}} & & \{\CA_\infty, \, L_\infty\} \arrow[dd,sloped,above,shift left=.5ex, "\text{evaluation}"] \arrow[uu,sloped,shift left=.5ex, "\text{inverse Mellin}"] \arrow[blue]{rr}{\text{functional}}[swap]{\text{equation}} & &  \{\CA_0, \, L_0\}\arrow[uu,sloped,above,shift left=.5ex, "\text{evaluation}"] \arrow[ll,blue] \arrow[dd,sloped,shift left=.5ex, "\text{inverse Mellin}"] & & \arrow{ll}{\text{series}}[swap]{\text{Dirichlet}} \{\zeta_n, \, S_n\} \arrow[sloped]{ddll}{\text{generating}}[swap]{\text{series}} 
\\ \\
& & \psi(\tau)\arrow[uu,sloped,above,shift left=.5ex, "\text{exact large-$n$}"] \arrow[uull,sloped,"\text{resurgence}"] & & \arrow[uu,sloped, shift left=.5ex,"\text{Mellin}"] \mathrm{disc}_{\frac{\pi}{2}}\phi(\hbar) \arrow[ll,swap, "\tau \rightarrow 0"]
 & &
\end{tikzcd}
\end{equation}

\subsection{Arithmetic twist of the Stokes constants}
Recall that the weak and strong coupling Stokes constants $S_n$, $R_n$, $n \in \IZ_{\ne 0}$, are explicitly given by the divisor sum functions in Eqs.~\eqref{eq: formulaS1} and~\eqref{eq: formulaS1infty}, respectively. As we have already observed, an enticingly simple arithmetic symmetry relates them. Namely, the $n$-th Stokes constants in the two asymptotic limits are obtained from one another by exchanging the positive integer divisors $d$ and $n/d$ that multiply the Dirichlet character inside the sums. We can say something more. 
Let us start by introducing the notation
\be \label{eq: n-primes}
n = \text{sign}(n) \prod_{p \in \mathbb{P}} p^{n_p} \in \IZ_{\ne 0} \, , \quad n_p \in \IN \, ,
\ee
where $\IP$ is the set of prime numbers.
We then define the arithmetic function\footnote{The function in Eq.~\eqref{eq: theta} can be interpreted as the Dirichlet convolution $\mathfrak{f}(n) = \sum_{d | n} f_3(d) \, \chi_{3,2}(n/d)$, where $f_3(n) = 3^{n_3}$, $n \in \IZ_{\ne 0}$.}
\be \label{eq: theta}
\mathfrak{f}(n) = 3^{n_3} \, \chi_{3,2}(n/3^{n_3}) \, , \quad n \in \IZ_{\ne 0} \, ,
\ee
where $\chi_{3,2}(n)$ is the non-principal Dirichlet character modulo $3$ in Eq.~\eqref{eq: character32}.
In particular, $\mathfrak{f}(n) = \chi_{3,2}(n)$ when $3 \nmid n$. Note that $\mathfrak{f}(n)$ is completely multiplicative and satisfies $\mathfrak{f}(-n)=-\mathfrak{f}(n)$.
\begin{prop} \label{prop:twist}
The normalized weak and strong coupling Stokes constants $S_n/S_1$, $R_n/R_1$, $n \in \IZ_{\ne 0}$, in Eqs.~\eqref{eq: formulaS1} and~\eqref{eq: formulaS1infty} are related by the arithmetic twist
\be \label{eq: stokes-twist}
\frac{S_n}{S_1}= \mathfrak{f}(n) \frac{R_n}{R_1} \, ,
\ee
where $\mathfrak{f}(n)$ is the arithmetic function in Eq.~\eqref{eq: theta}.
\end{prop}
\begin{proof}
    Let us fix $n \in \IZ_{\ne 0}$ and write $n = 3^{n_3} m$, where $n_3 \in \IN$ and $m \in \IZ_{\ne 0}$ with $3 \nmid m$.  
    As a straightforward consequence of their closed formulae as divisor sum functions, the weak and strong coupling Stokes constants satisfy the simple properties
    \be \label{eq: stokes-nm}
    S_n = S_m \, , \quad R_n = 3^{-n_3} R_m \, .
    \ee
    Using the closed formulae in Eqs.~\eqref{eq: formulaS1},~\eqref{eq: formulaS1infty}, and~\eqref{eq: theta}, we find that
    \be
        \mathfrak{f}(n) \frac{R_n}{R_1} = \sum_{d|n} \frac{d \, 3^{n_3}}{n} \chi_{3,2}(d m) = \sum_{d|m} \frac{d}{m} \chi_{3,2}(m/d) = \sum_{d|m} \frac{1}{d} \chi_{3,2}(d) = \frac{S_m}{S_1} \, ,
    \ee
    where we have applied the property 
    \be
    \chi_{3,2}(d m) = \chi_{3,2}(d)^2 \chi_{3,2}(m/d) =\chi_{3,2}(m/d) \, ,
    \ee
    which makes use of $\chi_{3,2}(d)^2=1$ for $3 \nmid d \mid m$. We conclude by substituting $S_m= S_n$.
\end{proof}
The arithmetic twist above plays a role in the global net of relations among the dual resurgent structures of $\mathrm{Tr}(\mrho_{\IP^2})$ that is analogous to the one played by the functional equation linking the two $L$-functions. Indeed, Proposition~\ref{prop:twist} supplies a two-headed arrow joining the weak and strong coupling Stokes constants directly in the commutative diagram in Eq.~\eqref{diag: strong-weak1}. We reproduce it below in an equivalent form that allows us to visualize the contribution from Eq.~\eqref{eq: stokes-twist}.\footnote{For completeness, we remind that the locations of the singularities $\zeta_n$, $\eta_n$, $n \in \IZ_{\ne 0}$, are related by $\zeta_n = 2 \pi \eta_n$.}
\begin{equation} \label{diag: strong-weak-S}
\begin{tikzcd}[column sep=2.8em, row sep=2.8em]
& & \arrow[ddll,sloped, shift right=.5ex,"\text{Mellin}"] \mathrm{disc}_{\frac{\pi}{2}}\psi(\tau) \arrow[rr, "\hbar \rightarrow 0"] & & \phi(\hbar) \arrow[ddrr,sloped,swap,shift right=.5ex, "\text{exact large-$n$}"]\arrow[dd,sloped,"\text{resurgence}"] & & 
\\ \\ 
 \{\CA_\infty, \, L_\infty\}\arrow[ddrr,sloped,swap,shift right=.5ex, "\text{evaluation}"] \arrow[uurr,sloped,swap, shift right=.5ex, "\text{inverse Mellin}"] & & \{\eta_n, \, R_n\} \arrow{ll}{\text{series}}[swap]{\text{Dirichlet}}\arrow[sloped]{uu}{\text{generating}}[swap]{\text{series}}  \arrow[blue]{rr}{\text{arithmetic}}[swap]{\text{twist}} & &   \{\zeta_n, \, S_n\} \arrow{rr}{\text{Dirichlet}}[swap]{\text{series}} \arrow[sloped]{dd}{\text{generating}}[swap]{\text{series}} \arrow[ll,blue]  & & \arrow[uull,sloped,shift right=.5ex, "\text{evaluation}"] \{\CA_0, \, L_0\}  \arrow[ddll,sloped,shift right=.5ex, "\text{inverse Mellin}"]
\\ \\
& & \psi(\tau)\arrow[uull,sloped,shift right=.5ex, "\text{exact large-$n$}"] \arrow[uu,sloped,"\text{resurgence}"] & & \arrow[uurr,sloped, shift right=.5ex,swap,"\text{Mellin}"] \mathrm{disc}_{\frac{\pi}{2}}\phi(\hbar) \arrow[ll,swap, "\tau \rightarrow 0"]
 & &
\end{tikzcd}
\end{equation}

\subsubsection*{$L$-functions and character twists}
The arithmetic twist of the Stokes constants described in Eq.~\eqref{eq: stokes-twist} appropriately translates in the language of the weak and strong coupling $L$-functions $L_0(s)$, $L_\infty(s)$, $s \in \IC$, in Eqs.~\eqref{eq: convolution2-0} and~\eqref{eq: convolution2-infty}. In particular, Eq.~\eqref{eq: stokes-twist} straightforwardly implies that  
\be
    L_0(s) = \sum_{n=1}^{\infty} \frac{S_n}{n^{s}} = \frac{S_1}{R_1} \sum_{n=1}^{\infty}\frac{R_n}{n^{s}}\mathfrak{f}(n) \, .
\ee
We can reformulate this arithmetic twist in a way that makes its meaning and implications more explicit. Before doing so, we introduce the principal Dirichlet character modulo $3$, that is, 
\be \label{eq: character31}
\chi_{3,1} (n) = 
\begin{cases}
0 & \quad \text{if} \quad n \equiv_3 0 \\
1 & \quad \text{else} 
\end{cases} \, , \quad n \in \IZ \, ,
\ee
which leads us to the following corollary to Proposition~\ref{prop:twist}.
\begin{cor} \label{cor: twist-cor1}
The normalized weak and strong coupling Stokes constants $S_n/S_1$, $R_n/R_1$, $n \in \IZ_{\ne 0}$, in Eqs.~\eqref{eq: formulaS1} and~\eqref{eq: formulaS1infty} satisfy
\be \label{eq: stokes-twist-v2}
\chi_{3,1}(n) \frac{S_n}{S_1}= \chi_{3,2}(n) \frac{R_n}{R_1} \, ,
\ee
where $\chi_{3,1}(n)$, $\chi_{3,2}(n)$ are the Dirichlet characters modulo $3$ in Eqs.~\eqref{eq: character32} and~\eqref{eq: character31}, respectively.
\end{cor}
\begin{proof}
Observe that Eq.~\eqref{eq: stokes-twist} implies
\be \label{eq: twist-cor1}
\frac{S_m}{S_1}= \chi_{3,2}(m) \frac{R_m}{R_1} \, , \quad m \in \IZ_{\ne 0} \, , \quad 3 \nmid m \, .
\ee
Multiplying both sides of Eq.~\eqref{eq: twist-cor1} by the Dirichlet character $\chi_{3,1}(m)$ in Eq.~\eqref{eq: character31} and applying 
\be
\chi_{3,1}(m) \chi_{3,2}(m) = \chi_{3,2}(m) \, ,
\ee
we obtain the desired statement.
\end{proof}
Note that Corollary~\ref{cor: twist-cor1} can be equivalently stated as
\be \label{eq: stokes-twist-v3}
\chi_{3,2}(n) \frac{S_n}{S_1}= \chi_{3,1}(n) \frac{R_n}{R_1} \, .
\ee
Indeed, multiplying both sides of Eq.~\eqref{eq: stokes-twist-v2} by $\chi_{3,2}(n)$ and using that $\chi_{3,2}(n)^2 = \chi_{3,1}(n)$, we find the expression in Eq.~\eqref{eq: stokes-twist-v3}.
Finally, let us introduce the following notation. We denote by $L(s,\chi)$ the $L$-function that results from twisting an arbitrary $L$-function $L(s)$, $s \in \IC$, by an arbitrary Dirichlet character $\chi(n)$, $n \in \IZ$. Namely,\footnote{Note that the Dirichlet $L$-function $L(s, \chi_{3,2})$, which appears in the factorization of the weak and strong coupling $L$-functions in Eqs.~\eqref{eq: convolution2-0} and~\eqref{eq: convolution2-infty}, is obtained by twisting the Riemann zeta function $\zeta(s)$ by the Dirichlet character $\chi_{3,2}$ in Eq.~\eqref{eq: character32}.}
\be \label{eq: twist-notation}
L(s,\chi) =  \sum_{n=1}^{\infty} \frac{A_n}{n^s} \chi(n) \, ,
\ee
where the complex numbers $A_n$ are the coefficients in the $L$-series representation of $L(s)$.
\begin{theorem} \label{theorem: L-twist}
The weak and strong coupling $L$-functions $L_0(s)$, $L_\infty(s)$, $s \in \IC$, in Eqs.~\eqref{eq: convolution2-0} and~\eqref{eq: convolution2-infty} are obtained from each other by applying a character twist. Specifically, they satisfy  
\begin{subequations}
    \begin{align}
        L_0(s) &= \frac{S_1}{R_1}\frac{1}{1-3^{-s}} L_\infty(s, \chi_{3,2}) \, , \label{eq: L-twist-zero} \\ 
        L_\infty(s) &= \frac{R_1}{S_1}\frac{1}{1-3^{-(s+1)}} L_0(s, \chi_{3,2}) \, , \label{eq: L-twist-inf}
    \end{align}
\end{subequations}
where $\chi_{3,2}(m)$ is the unique non-principal Dirichlet character modulo $3$ in Eq.~\eqref{eq: character32}.
\end{theorem}
\begin{proof}
    We write an arbitrary integer $n \in \IZ_{>0}$ in the form $n=3^{n_3} m$, where $n_3 \in \IN$ and $m \in \IZ_{>0}$ with $3 \nmid m$, as before. 
    It follows from the definition of the weak coupling $L$-function in Eq.~\eqref{eq: convolution2-0} that it factorizes as
    \be \label{eq: L-twist-proof-zero}
    L_0(s) = \sum_{n_3=0}^{\infty} 3^{- s n_3} \sum_{\substack{m=1 \\ 3 \nmid m}}^{\infty} \frac{S_m}{m^{s}}=\frac{1}{1-3^{-s}} \sum_{\substack{m=1 \\ 3 \nmid m}}^{\infty} \frac{S_m}{m^{s}} \, ,
    \ee
    where we have used the first formula in Eq.~\eqref{eq: stokes-nm} and resummed the geometric series over the index $n_3$. As a consequence of Corollary~\ref{cor: twist-cor1}, the series in the RHS of Eq.~\eqref{eq: L-twist-proof-zero} can then be written as
    \be \label{eq: L-twist-proof-zero2}
    \sum_{\substack{m=1 \\ 3 \nmid m}}^{\infty} \frac{S_m}{m^{s}} = \sum_{n=1}^{\infty} \frac{S_n}{n^{s}} \chi_{3,1}(n) = \frac{S_1}{R_1} \sum_{n=1}^{\infty} \frac{R_n}{n^{s}} \chi_{3,2}(n) \, .
    \ee
    Substituting Eq.~\eqref{eq: L-twist-proof-zero2} into Eq.~\eqref{eq: L-twist-proof-zero} and using the notation introduced in Eq.~\eqref{eq: twist-notation} for $L_\infty(s)$ yields the statement in Eq.~\eqref{eq: L-twist-zero}.
    Analogously, the definition of the strong coupling $L$-function in Eq.~\eqref{eq: convolution2-infty} implies the factorization
    \be \label{eq: L-twist-proof-inf}
    L_\infty(s) = \sum_{n_3=0}^{\infty} 3^{- (s+1)n_3} \sum_{\substack{m=1 \\ 3 \nmid m}}^{\infty} \frac{R_m}{m^{s}} = \frac{1}{1-3^{-(s+1)}} \sum_{\substack{m=1 \\ 3 \nmid m}}^{\infty} \frac{R_m}{m^{s}} \, ,
    \ee
    where we have used the second formula in Eq.~\eqref{eq: stokes-nm} and resummed the geometric series over the index $n_3$. As a consequence of Corollary~\ref{cor: twist-cor1} in the equivalent form of Eq.~\eqref{eq: stokes-twist-v3}, the series in the RHS of Eq.~\eqref{eq: L-twist-proof-inf} can then be written as
    \be \label{eq: L-twist-proof-inf2}
    \sum_{\substack{m=1 \\ 3 \nmid m}}^{\infty} \frac{R_m}{m^{s}} = \sum_{n=1}^{\infty} \frac{R_n}{n^{s}} \chi_{3,1}(n) = \frac{R_1}{S_1} \sum_{n=1}^{\infty} \frac{S_n}{n^{s}} \chi_{3,2}(n) \, .
    \ee
    Substituting Eq.~\eqref{eq: L-twist-proof-inf2} into Eq.~\eqref{eq: L-twist-proof-inf} and using the notation introduced in Eq.~\eqref{eq: twist-notation} for $L_0(s)$ yields the statement in Eq.~\eqref{eq: L-twist-inf}.
\end{proof}

\section{Summability and quantum modularity}\label{sec:summability-and-QM}
In this section, we present our results on the summability and modularity properties of the generating functions $f_0$, $f_\infty$ of the Stokes constants at weak and strong coupling, defined in Eqs.~\eqref{eq:f_0} and~\eqref{eq:f_inf}, respectively. In particular, we show how the median resummation allows us to effectively reconstruct $f_0(y)$ and (conjecturally) $f_\infty(y)$ from their asymptotic expansions in the limit $y\to0$. Moreover, we prove that they are holomorphic quantum modular forms of weight zero under the congruence subgroup $\Gamma_1(3) \subset \mathsf{SL}_2(\IZ)$.      
\subsection{Borel--Laplace sums}\label{sec:Borel-Laplace}
As a preliminary step toward the summability study we perform in Section~\ref{sec:median_resummation}, we explicitly compute the Borel--Laplace sums of the asymptotic series $\phi(\hbar)$ and $\psi(\tau)$ in Eqs.~\eqref{eq: phiP2} and~\eqref{eq: phiP2infty} capturing the perturbative expansions of the spectral trace of local $\IP^2$ in the limits $\hbar\to 0$ and $\hbar\to\infty$, respectively.
We start by introducing the following notation. Let $\e\colon\IC\setminus \ri\IR_{\geq 0}\to\IC$ be the function defined by\footnote{The notation in Eq.~\eqref{eq: e1-def} was suggested by M. Kontsevich as the function $\e$ plays a similar role in other examples (\emph{e.g.}, the case of $\sigma, \sigma^*$ in~\cite[from min.~21.55]{Fantini-talk-IHES} and the case of even/odd Maass cusp forms in~\cite{FR1maths}).}  
\begin{equation} \label{eq: e1-def}
    \e(y):=\frac{1}{2\pi \ri}\int_0^{\infty} \re^{-2\pi t} \frac{dt}{t+\ri y} \, .
\end{equation}
Note that, applying the simple change of coordinates $t \mapsto t/(2 \pi) - \ri y$, the exponential integral $\e(y)$ can be written in the equivalent form
\be
    \e(y)=\frac{1}{2\pi \ri} \re^{2\pi \ri y} \, \Gamma(0,2\pi \ri y) \, ,
\ee
where $\Gamma(s, 2\pi \ri y)$ denotes the upper incomplete gamma function. It follows from the analyticity properties of $\Gamma(0,2\pi \ri y)$ that the function $\e(y)$ is analytic for $y \in \IC\setminus \ri\IR_{\geq 0}$ and its jump across the branch cut along the positive imaginary axis is given by 
\be
\lim_{\delta \rightarrow 0^+} \e(\ri x - \delta) - \e(\ri x+ \delta) = \re^{-2 \pi x} \, ,
\ee
where we have fixed $x \in \IR_{\ge 0}$.
Let us then go back to the Gevrey-1 asymptotic series $\phi(\hbar)$, $\psi(\tau)$ in Eqs.~\eqref{eq: phiP2} and~\eqref{eq: phiP2infty}. As recalled in Section~\ref{sec:exact-solution}, their perturbative coefficients $a_{2n}$, $b_{2n}$, $n \in \IZ_{>0}$, satisfy the exact large-order relations in Eq.~\eqref{eq: exactlarge}, which can be straightforwardly written as
\begin{subequations}
\begin{align}
a_{2n} &= \frac{\Gamma(2n)}{\pi \ri} \sum_{m = 1}^{\infty} \frac{S_m}{\zeta_m^{2n}} \, , \label{eq: exact2-zero} \\
b_{2n} &= \frac{\Gamma(2n-1)}{\pi \ri} \sum_{m = 1}^{\infty} \frac{R_m}{\eta_m^{2n-1}} \, , \label{eq: exact2-infty}
\end{align}
\end{subequations}
where the Stokes constants $S_m$, $R_m$ are given explicitly in Eqs.~\eqref{eq: formulaS1} and~\eqref{eq: formulaS1infty}, while the locations $\zeta_m$, $\eta_m$ of the singularities in the complex $\zeta$-plane are in Eqs.~\eqref{eq: zetan} and~\eqref{eq: zetan-infty}.
We now prove that both $\phi(\hbar)$ and $\psi(\tau)$ are Borel--Laplace summable at the angles $0$ and $\pi$ and their Borel--Laplace sums can be expressed in terms of the exponential integral in Eq.~\eqref{eq: e1-def}.
\begin{prop}\label{prop:Borel_sum_0}
    The Borel--Laplace sum of the weak coupling perturbative series $\phi(\hbar)$, defined in Eq.~\eqref{eq: phiP2}, along the positive real axis is 
    \begin{equation}
        s_0(\phi)(\hbar)=- \sum_{m \in \IZ_{\ne 0}} S_m \mathbf{e}_1\left(-\frac{2 \pi m}{3 \hbar} \right) \, ,
    \end{equation}
    which is analytic for $\Re(\hbar)>0$. 
\end{prop}
\begin{proof}
As summarized in Appendix~\ref{app: resurgence}, the Borel--Laplace resummation requires two successive steps. First, we compute the Borel transform $\hat{\phi}(\zeta)$, which gives
\be \label{eq: BTzero}
\begin{aligned}
\hat{\phi}(\zeta) &= \sum_{n=1}^{\infty} \frac{a_{2n}}{\Gamma(2n)} \zeta^{2n-1} \\
&=\frac{1}{\pi \ri} \sum_{m=1}^{\infty} S_m \zeta^{-1} \sum_{n=1}^{\infty} \left( \frac{\zeta}{\zeta_m} \right)^{2n} = \frac{1}{\pi \ri} \sum_{m=1}^{\infty} S_m \frac{\zeta}{\zeta_m^2-\zeta^2} \\
&= - \frac{1}{2 \pi \ri} \sum_{m=1}^{\infty} \frac{S_m}{\zeta_m+\zeta} + \frac{1}{2 \pi \ri} \sum_{m=1}^{\infty} \frac{S_m}{\zeta_m-\zeta} = \frac{1}{2 \pi \ri} \sum_{m \in \IZ_{\ne 0}} \frac{S_m}{\zeta_m-\zeta} \, ,
\end{aligned}
\ee
where we have substituted Eq.~\eqref{eq: exact2-zero}, applied the identity
\begin{equation} \label{eq: identity-frac}
    \frac{1}{(x-y)(x+y)}= \frac{1}{2x(x-y)}+\frac{1}{2x(x+y)} \, ,
\end{equation}
and used the properties $\zeta_{-m}=-\zeta_m$ and $S_{-m}=S_m$, $m \in \IZ_{\ne 0}$.
Then, assuming $\hbar\in\IR_{\ge 0}$, we perform the Laplace transform along the positive real axis, that is, 
\be \label{eq: LTzero}
\begin{aligned}
s_0(\phi)(\hbar) &= \int_0^{\infty} d\zeta \re^{-\zeta/\hbar} \hat{\phi}(\zeta) \\
&=\frac{1}{2 \pi \ri} \int_0^{\infty} d\zeta \sum_{m \in \IZ_{\ne 0}} S_m \frac{\re^{-\zeta/\hbar}}{\zeta_m-\zeta} \\
&= \frac{1}{2 \pi \ri}  \int_0^{\infty} dt \sum_{m \in \IZ_{\ne 0}} S_m \frac{\re^{-2 \pi t}}{\frac{\zeta_m}{2 \pi \hbar}-t} = - \sum_{m \in \IZ_{\ne 0}} S_m \mathbf{e}_1\left(\frac{\ri \zeta_m}{2 \pi \hbar} \right) \, ,    
\end{aligned}
\ee
where we have substituted Eq.~\eqref{eq: BTzero}, applied the change of variable $2 \pi t = \zeta/\hbar$, permuted sum and integral due to absolute convergence, and used the notation introduced in Eq.~\eqref{eq: e1-def}. Finally, applying Eq.~\eqref{eq: zetan} in Eq.~\eqref{eq: LTzero}, we find the desired expression for the Borel--Laplace sum of $\phi(\hbar)$, which can be analytically continued to $\Re(\hbar)>0$ thanks to the properties of the function $\e$.
\end{proof}
\begin{prop}\label{prop:Borel_sum_inf}
    The Borel--Laplace sum of the strong coupling perturbative series $\psi(\tau)$, defined in Eq.~\eqref{eq: phiP2infty}, along the positive real axis is 
    \begin{equation}
        s_0(\psi)(\tau)=- \sum_{m \in \IZ_{\ne 0}} R_m \mathbf{e}_1\left(-\frac{m}{3 \tau} \right) \, ,
    \end{equation}
    which is analytic for $\Re(\tau)>0$.
\end{prop}
\begin{proof}
We follow the same steps of the proof of Proposition~\ref{prop:Borel_sum_0}. First, we compute the Borel transform $\hat{\psi}(\tau)$ and obtain
\be \label{eq: BTinfty}
\begin{aligned}
\hat{\psi}(\tau) &= \sum_{n=1}^{\infty} \frac{b_{2n}}{\Gamma(2n-1)} \zeta^{2n-2}\\
&=\frac{1}{\pi \ri} \sum_{m=1}^{\infty} R_m \eta_m^{-1} \sum_{n=1}^{\infty} \left( \frac{\zeta}{\eta_m} \right)^{2n-2} = \frac{1}{\pi \ri} \sum_{m=1}^{\infty} R_m \frac{\eta_m}{\eta_m^2-\zeta^2} \\
&= \frac{1}{2 \pi \ri} \sum_{m=1}^{\infty} \frac{R_m}{\eta_m+\zeta} + \frac{1}{2 \pi \ri} \sum_{m=1}^{\infty} \frac{R_m}{\eta_m-\zeta} = \frac{1}{2 \pi \ri} \sum_{m \in \IZ_{\ne 0}} \frac{R_m}{\eta_m-\zeta} \, ,
\end{aligned}
\ee
where we have substituted Eq.~\eqref{eq: exact2-infty}, applied the fractional identity in Eq.~\eqref{eq: identity-frac}, and used the properties $\eta_{-m}=-\eta_m$ and $R_{-m}=-R_m$, $m \in \IZ_{\ne 0}$.
Then, assuming $\tau\in\IR_{\ge 0}$, we compute the Laplace transform along the positive real axis, which gives
\be \label{eq: LTinfty}
\begin{aligned}
s_0(\psi)(\tau) &= \int_0^{\infty} d\zeta \re^{-\zeta/\tau} \hat{\psi}(\zeta) \\
&=\frac{1}{2 \pi \ri} \int_0^{\infty} d\zeta \sum_{m \in \IZ_{\ne 0}} R_m\re^{-\zeta/\tau} \frac{1}{\eta_m-\zeta} \\
&= \frac{1}{2 \pi \ri}  \int_0^{\infty} dt \sum_{m \in \IZ_{\ne 0}} R_m \re^{-2 \pi t} \frac{1}{\frac{\eta_m}{2 \pi \tau}-t} = - \sum_{m \in \IZ_{\ne 0}} R_m \mathbf{e}_1\left(\frac{\ri \eta_m}{2 \pi \tau} \right)  \, ,   
\end{aligned}
\ee
where we have substituted Eq.~\eqref{eq: BTinfty}, performed the change of variable $2 \pi t = \zeta/\tau$, permuted sum and integral due to absolute convergence, and again applied the notation in Eq.~\eqref{eq: e1-def}. We conclude by substituting Eq.~\eqref{eq: zetan-infty} into Eq.~\eqref{eq: LTinfty}. Notice that $s_0(\psi)(\tau)$ can be analytically continued to $\Re(\tau)>0$ thanks to the properties of the function $\e$.
\end{proof}
In addition, following the same steps in the proofs of Propositions~\ref{prop:Borel_sum_0} and~\ref{prop:Borel_sum_inf}, we obtain the analogous expressions for the Borel--Laplace sums of $\phi(\hbar)$ and $\psi(\tau)$ at angle $\pi$. Namely,  
\begin{subequations}
    \begin{align}
        s_\pi(\phi)(\hbar)&=- \sum_{m \in \IZ_{\ne 0}} S_m \mathbf{e}_1\left(-\frac{2 \pi m}{3 \hbar} \right) \,, \\
        s_\pi(\psi)(\tau)&=- \sum_{m \in \IZ_{\ne 0}} R_m \mathbf{e}_1\left(-\frac{ m}{3 \tau} \right) \,,
    \end{align}
\end{subequations}
which are analytic for $\Re(\hbar)<0$ and $\Re(\tau)<0$, respectively.

Recall that the asymptotic behavior in the limit $y \rightarrow 0$ of the generating functions $f_0(y)$ and $f_\infty(y)$ of the weak and strong coupling Stokes constants in Eqs.~\eqref{eq:f_0} and~\eqref{eq:f_inf} is dictated by the original perturbative series $\psi(\tau)$ and $\phi(\hbar)$, respectively, according to the formulae in Eqs.~\eqref{eq: f0-psi} and~\eqref{eq: finf-phi}. Therefore, the Borel--Laplace sums of the asymptotic series $\tilde{f}_0(y)$ and $\tilde{f}_\infty(y)$ easily follow from Propositions~\ref{prop:Borel_sum_0} and~\ref{prop:Borel_sum_inf}.
For completeness, we write their explicit formulae in the two corollaries below.
\begin{cor}\label{cor:BL-f0_R}
The Borel--Laplace sums of the perturbative expansion $\tilde{f}_0(y)$ in Eq.~\eqref{eq:asymp-f0} along the positive and negative real axes are written in terms of the function $\e$ in Eq.~\eqref{eq: e1-def} as 
\begin{subequations}
    \begin{align}
        s_0(\tilde{f}_0)(y)&= -\frac{\pi \ri}{2}-\frac{3\mathcal{V}}{2\pi \ri y} +2\sum_{m\in\IZ_{\neq 0}} R_m\, \e\left(-\frac{m}{3y}\right) \, , \quad \Re(y)>0 \, , \label{eq: sumzero-cor-0} \\
        s_\pi(\tilde{f}_0)(y)&= -\frac{\pi \ri}{2}-\frac{3\mathcal{V}}{2\pi \ri y} +2\sum_{m\in\IZ_{\neq 0}} R_m\, \e\left(-\frac{m}{3y}\right) \, , \quad \Re(y)<0 \, , \label{eq: sumzero-cor-pi} 
    \end{align}
\end{subequations}
where $\CV = 2 \Im\left(\mathrm{Li}_2(\re^{2 \pi \ri/3})\right)$ as before.
\end{cor}
\begin{proof}
    The proof of the statement follows directly from Eq.~\eqref{eq: f0-psi} and Proposition~\ref{prop:Borel_sum_inf}.
\end{proof}
\begin{cor}\label{cor:BL-finf_S}
The Borel--Laplace sums of the perturbative expansion $\tilde{f}_\infty(y)$ in Eq.~\eqref{eq:asymp-finf} along the positive and negative real axes are written in terms of the function $\e$ in Eq.~\eqref{eq: e1-def} as 
\begin{subequations}
    \begin{align}
        s_0(\tilde{f}_\infty)(y)&= -3\log\frac{\Gamma(2/3)}{\Gamma(1/3)}-\log(-6\pi \ri y)-2\sum_{m\in\IZ_{\neq 0}}S_m\, \e\left(-\frac{m}{3y}\right) \, , \quad \Re(y)>0 \, , \label{eq: suminf-cor-0} \\
        s_\pi(\tilde{f}_\infty)(y)&= -3\log\frac{\Gamma(2/3)}{\Gamma(1/3)}-\log(-6\pi \ri y)-2\sum_{m\in\IZ_{\neq 0}}S_m\, \e\left(-\frac{m}{3y}\right) \, , \quad \Re(y)<0 \, . \label{eq: suminf-cor-pi} 
    \end{align}
\end{subequations}
\end{cor}
\begin{proof}
    The proof of the statement follows directly from Eq.~\eqref{eq: finf-phi} and Proposition~\ref{prop:Borel_sum_0}.
\end{proof}

\subsection{Median resummation}\label{sec:median_resummation} 
We will now show that the median resummation of the asymptotic expansion $\tilde{f}_0(y)$ of the generating function of the weak coupling Stokes constants $S_n$, $n \in \IZ_{\ne 0}$, appropriately reconstructs the generating function $f_0(y)$ itself. As a consequence, our results showcase the effectiveness of the median resummation as a summability tool.
We begin with the preliminary Lemma~\ref{lemma:summability-logPhiNC}, where we prove that the Borel--Laplace sums at angles $0$ and $\pi$ of the asymptotic expansion $\tilde{f}_0(y)$ in Eq.~\eqref{eq:asymp-f0} reproduce the original function $f_0(y)$ with a correction that is suitably encoded in the dual function $f_\infty(y)$. 

\begin{lemma}\label{lemma:summability-logPhiNC}
    Let $\tilde{f}_0(y)$ be the asymptotic expansion in the limit $y\to 0$ of the generating function $f_0(y)$ in Eq.~\eqref{eq:f_0} with $\Im(y)>0$, which we have written explicitly in Eq.~\eqref{eq:asymp-f0}. 
    Its Borel--Laplace sums along the positive and negative real axes can be expressed as
    \begin{subequations}
        \begin{align}
            s_0(\tilde{f}_0)(y)&= f_0(y) +f_\infty\big(-\tfrac{1}{3y}\big) \, , \quad \Re(y)>0 \, , \label{eq: BL-zero-0}\\
            s_\pi(\tilde{f}_0)(y)&=f_0(y) -f_\infty\big(-\tfrac{1}{3y}\big) \, , \quad \Re(y)<0 \, . \label{eq: BL-zero-pi}
        \end{align}
    \end{subequations}
\end{lemma}
\begin{proof}
Applying Eq.~\eqref{eq: logPhiNC} to expand the two quantum dilogarithms in Eq.~\eqref{eq:f_0-closed} separately, $\tilde{f}_0(y)$ can be written as
\be \label{eq: sum0-proof-start}
\begin{aligned}
    \frac{1}{3}\tilde{f}_0(y)
    =&-\frac{\pi \ri}{3}-\frac{1}{2} \left(\log(1-w)-\log(1-w^{-1})\right) \\
    &-\frac{1}{2\pi \ri y} \left(\mathrm{Li}_2(w) - \mathrm{Li}_2(w^{-1})\right) -\left(\tilde{\varphi}_+(y)-\tilde{\varphi}_-(y)\right)\, ,
\end{aligned}
\ee
where $w= \re^{2 \pi \ri/3}$ as before, and we have introduced the two formal power series
\be
\tilde{\varphi}_\pm(y):=\sum_{k=1}^\infty (2\pi \ri y)^{2k-1}\frac{B_{2k}}{(2k)!}\mathrm{Li}_{2-2k}(w^{\pm 1})\, . 
\ee
Let us compute the Borel transform of $\tilde{\varphi}_\pm(y)$ in the variable $\zeta$ conjugate to $y$. We find that 
\be
\begin{aligned}
    \borel\left[\tilde{\varphi}_\pm \right](\zeta)&=\sum_{k=1}^\infty (2\pi i)^{2k-1}\frac{B_{2k}}{(2k)!}\mathrm{Li}_{2-2k}(w^{\pm})\frac{\zeta^{2k-2}}{(2k-2)!}\\
    &=\left(\sum_{k=1}^\infty \frac{B_{2k}}{(2k)!}\zeta^{2k-2} \right) \diamond \left(\sum_{k=1}^\infty (2\pi i)^{2k-1}\frac{\mathsf{Li}_{2-2k}(w^{\pm})}{(2k-2)!} \zeta^{2k-2} \right)\\
    &=:\tilde{g}(\zeta)\diamond \tilde{f}_\pm(\zeta)\,,
\end{aligned}
\ee
where $\diamond$ denotes the Hadamard product~\cite{Hadamard}. The formal power series $\tilde{g}(\zeta)$, $\tilde{f}_\pm(\zeta)$ have a finite radius of convergence at $\zeta=0$ and can be resummed explicitly into the functions\footnote{We impose that $g(0)=1/12$ to eliminate the removable singularity of $g(\zeta)$ at the origin.}
\begin{subequations}
    \begin{align}
        g(\zeta)&=-\frac{1}{\zeta^2}+\frac{1}{2 \zeta}\coth\left(\frac{\zeta}{2}\right) \, , \quad |\zeta|<2\pi \, , \\
    f_\pm(\zeta)&=-\pi \ri\left(\frac{1}{1+\re^{\pm\frac{\pi \ri}{3}-2\pi \ri \zeta}}+\frac{1}{1+\re^{\pm\frac{\pi \ri}{3}+2\pi \ri \zeta}}\right) \, , \quad |\zeta|<\frac{1}{3} \, , 
    \end{align}
\end{subequations}
respectively. After being analytically continued to the whole complex $\zeta$-plane, the function $g(\zeta)$ has poles of order one along the imaginary axis at
\be
\mu_m = 2 \pi \ri m \, , \quad m \in \IZ_{\ne 0} \, ,
\ee
while the functions $f_\pm(\zeta)$ have poles of order one along the real axis at 
\begin{subequations}
    \begin{align}
        \nu_k^{(+,1)}&= \frac{1}{3}+k \, , \quad \nu_k^{(+,2)}= -\frac{1}{3}+k\, ,  \\
        \nu_k^{(-,1)}&=\frac{2}{3}+k \, , \quad \nu_k^{(-,2)}=-\frac{2}{3}+k \, , 
    \end{align}
\end{subequations}
for $k \in \IZ$, respectively.
We consider a circle $\gamma$ in the complex $s$-plane with center $s=0$ and radius $0 < r < 2 \pi$ and apply Hadamard's multiplication theorem. The Borel transform of $\tilde{\varphi}_\pm(y)$ can be written as the integral
\be \label{eq: intBorelinfty}
\begin{aligned}
    \borel \left[\tilde{\varphi}_\pm \right](\zeta) &=\frac{1}{2\pi \ri}\int_\gamma g(s)\, f_\pm\left(\frac{\zeta}{s}\right) \frac{ds}{s} \\
    &= \frac{1}{2}\int_\gamma \left(\frac{1}{s}-\frac{1}{2}\coth\left(\frac{s}{2}\right) \right)\left(\frac{1}{1+\re^{\pm\frac{\pi \ri}{3}-2\pi \ri \frac{\zeta}{s}}}+\frac{1}{1+\re^{\pm\frac{\pi \ri}{3}+2\pi \ri \frac{\zeta}{s}}}\right) \frac{ds}{s^2} \, ,
\end{aligned}
\ee
for $|\zeta|<r/3$. For such values of $\zeta$, the function $s \mapsto f_\pm(\zeta/s)$ has singular points at $s = \zeta/\nu^{(\pm, i)}_k$, $k \in \IZ$, which sit inside the contour of integration $\gamma$ and accumulate at the origin, and no singularities for $|s| > r$. The function $g(s)$ has simple poles at the points $s= \mu_m$ with residues
\be
\underset{s= \pm 2 \pi \ri m}{\text{Res}} g(s) = \pm \frac{1}{2 \pi \ri} \, , \quad m \in \IZ_{> 0} \, . 
\ee
By Cauchy's residue theorem, the integral in Eq.~\eqref{eq: intBorelinfty} can be evaluated by summing the residues at the poles of the integrand which lie outside $\gamma$, allowing us to express the Borel transform as an exact function of $\zeta$. More precisely, we find that
\be \label{eq: intBorel2infty-proof}
\begin{aligned}
    \borel \left[\tilde{\varphi}_\pm \right](\zeta)&=-\sum_{m\in\IZ_{\neq 0}} \underset{s=2 \pi \ri m}{\text{Res}} \left[ g(s) f_{\pm}\left(\frac{\zeta}{s}\right)\frac{1}{s}\right] \\
    &=\sum_{m\in\IZ_{\neq 0}}\frac{1}{4 \pi \ri m^2}\left(\frac{1}{1+\re^{\pm\frac{\pi \ri}{3}-\frac{\zeta}{m}}}+\frac{1}{1+\re^{\pm\frac{\pi \ri}{3}+ \frac{\zeta}{m}}}\right)\\
    &=\frac{1}{2 \pi \ri} \sum_{m=1}^{\infty}\frac{1}{m^2}\left(\frac{1}{1+\re^{\pm\frac{\pi \ri}{3}-\frac{\zeta}{m}}}+\frac{1}{1+\re^{\pm\frac{\pi \ri}{3}+ \frac{\zeta}{m}}}\right) \, ,
\end{aligned}
\ee
which is a well-defined, exact function of $\zeta$ for $|\zeta| < 2 \pi /3$.
In fact, the Borel transforms in Eq.~\eqref{eq: intBorel2infty-proof} converge when $\frac{\zeta}{2 \pi \ri} \neq \frac{m}{3} + n m$ and $\frac{\zeta}{2 \pi \ri} \neq \frac{2 m}{3} + n m$  with $m \in \IZ_{\ne 0}$ and $n\in\IZ$, respectively.\footnote{The convergence of the infinite sum in the RHS of Eq.~\eqref{eq: intBorel2infty-proof} can be easily verified by, \emph{e.g.}, the limit comparison test. Indeed, where defined, the generic term of the series is dominated by $1/m^2$.}

We can now compute the Laplace transform of $\borel \left[\tilde{\varphi}_\pm \right](\zeta)$ along the positive and negative real axes. When $\Re (y), \Im (y)>0$, we have that
\be
\begin{aligned}
    s_0(\tilde{\varphi}_\pm)(y)&=\frac{1}{2\pi \ri}\int_0^\infty \re^{-\zeta/y} \left[\sum_{m=1}^{\infty}\frac{1}{m^2}\left(\frac{1}{1+\re^{\pm\frac{\pi \ri}{3}-\frac{\zeta}{m}}}+\frac{1}{1+\re^{\pm\frac{\pi \ri}{3}+ \frac{\zeta}{m}}}\right)\right] d\zeta \\
    &=\frac{1}{2\pi \ri}\int_0^\infty \left(\sum_{m=1}^{\infty}\frac{\re^{-mt/y}}{m} \right) \left(\frac{1 }{1+\re^{\pm\frac{\pi \ri}{3}-t}}+\frac{1}{1+\re^{\pm\frac{\pi \ri}{3}+t}}\right) dt \, ,
\end{aligned}
\ee
where we have applied the change of variable $\zeta = m t$. Resumming the series in $m$, we find that
\be \label{eq: sum0-proof1}
\begin{aligned}
    s_0(\tilde{\varphi}_\pm)(y)
    &=-\frac{1}{2\pi \ri}\int_0^\infty \log(1-\re^{-t/y})\left(\frac{1}{1+\re^{\pm\frac{\pi i}{3}-t}}+\frac{1}{1+\re^{\pm\frac{\pi \ri}{3}+t}}\right) dt \\
    &=-\frac{1}{2\pi \ri}\int_{-\infty}^\infty \frac{\log(1-\re^{-t/y})}{1+\re^{\pm\frac{\pi \ri}{3}-t}} dt +\frac{1}{2\pi \ri}\int_0^{-\infty}\frac{\log(-\re^{t/y})}{1+\re^{\pm\frac{\pi \ri}{3}-t}} dt \, ,
\end{aligned}
\ee
where we have divided the integral into the sum of two contributions for simplicity.
Indeed, we observe that the second term in the RHS of Eq.~\eqref{eq: sum0-proof1} gives
\be \label{eq: sum0-proof1-p1}
\frac{1}{2\pi \ri}\int_0^{-\infty}\frac{\log(-\re^{t/y})}{1+\re^{\pm\frac{\pi \ri}{3}-t}} dt = -\frac{\mathrm{Li}_2(w^{\pm 1})}{2\pi \ri y}-\frac{1}{2}\log(1-w^{\pm 1}) \, ,
\ee
while the first term gives
\begin{subequations}
    \begin{align}
        -\frac{1}{2\pi \ri}\int_{-\infty}^\infty \frac{\log(1-\re^{-t/y})}{1+\re^{\frac{\pi \ri}{3}-t}} dt &= \log\frac{(w; \, \re^{2\pi \ri y})_\infty}{(\re^{-4\pi \ri/(3y)}; \, \re^{-2\pi \ri/y})_\infty} \, , \label{eq: sum0-proof1-p2a}\\
        -\frac{1}{2\pi \ri}\int_{-\infty}^\infty \frac{\log(1-\re^{-t/y})}{1+\re^{-\frac{\pi \ri}{3}-t}} dt &= \log\frac{(w^{-1}; \, \re^{2\pi \ri y})_\infty}{(\re^{-2\pi \ri/(3y)}; \, \re^{-2\pi \ri/y})_\infty}\, . \label{eq: sum0-proof1-p2b}
    \end{align}
\end{subequations}
The last two formulae follow from~\cite[Theorem~A-29]{wheeler-thesis} as shown in~\cite[Eq.~(5.108)]{wheeler-book}.\footnote{We thank C. Wheeler for sharing the preliminary version of his book in preparation~\cite{wheeler-book}.}
Therefore, substituting Eqs.~\eqref{eq: sum0-proof1-p1},~\eqref{eq: sum0-proof1-p2a}, and~\eqref{eq: sum0-proof1-p2b} into the formula for the Borel--Laplace sum in Eq.~\eqref{eq: sum0-proof1}, we find that 
\begin{subequations}
    \begin{align}
    s_0(\tilde{\varphi}_+)(y)=& -\frac{\mathrm{Li}_2(w)}{2\pi \ri y}-\frac{1}{2}\log(1-w) +\log\frac{(w; \, \re^{2\pi \ri y})_\infty}{(\re^{-4\pi \ri/(3y)}; \, \re^{-2\pi \ri/y})_\infty} \, , \label{eq: sum0-proof1-p3a} \\
    s_0(\tilde{\varphi}_-)(y)=& -\frac{\mathrm{Li}_2(w^{-1})}{2\pi \ri y}-\frac{1}{2}\log(1-w^{-1}) 
    + \log\frac{(w^{-1}; \, \re^{2\pi \ri y})_\infty}{(\re^{-2\pi \ri/(3y)}; \, \re^{-2\pi \ri/y})_\infty} \, . \label{eq: sum0-proof1-p3b}
    \end{align}
\end{subequations}
Now, it follows directly from Eqs.~\eqref{eq: sum0-proof-start},~\eqref{eq: sum0-proof1-p3a}, and~\eqref{eq: sum0-proof1-p3a} that
\be
\begin{aligned}
    s_0(\tilde{f}_0)(y)&= -\pi \ri -3 \log\frac{(w; \, \re^{2\pi \ri y})_\infty}{(\re^{-4\pi \ri/(3y)}; \, \re^{-2\pi \ri/y})_\infty} +3 \log\frac{(w^{-1}; \, \re^{2\pi \ri y})_\infty}{(\re^{-2\pi \ri/(3y)}; \, \re^{-2\pi \ri/y})_\infty} \\
    &=f_0(y)+f_\infty(-1/3y) \, ,
\end{aligned}
\ee
where we have applied the closed formulae for the generating functions in Eqs.~\eqref{eq:f_0-closed} and~\eqref{eq:f_inf-closed}.

Similarly, when $\Re (y)<0$ and $\Im(y)>0$, the Borel--Laplace sum of the formal power series $\tilde{\varphi}_{\pm}(y)$ along the negative real axis is given by
\be
\begin{aligned}
    s_\pi(\tilde{\varphi}_\pm)(y)&=\frac{1}{2\pi \ri}\int_0^{-\infty} \re^{-\zeta/y} \left[\sum_{m=1}^{\infty}\frac{1}{m^2}\left(\frac{1}{1+\re^{\pm\frac{\pi \ri}{3}-\frac{\zeta}{m}}}+\frac{1}{1+\re^{\pm\frac{\pi \ri}{3}+ \frac{\zeta}{m}}}\right) \right] d\zeta \\
    &=\frac{1}{2\pi \ri}\int_0^{-\infty} \left(\sum_{m=1}^{\infty}\frac{\re^{-mt/y}}{m} \right) \left(\frac{1 }{1+\re^{\pm\frac{\pi \ri}{3}-t}}+\frac{1}{1+\re^{\pm\frac{\pi \ri}{3}+t}}\right) dt \, ,
\end{aligned}
\ee
where we have applied the change of variable $\zeta = m t$. Resumming the series in $m$, we find that
\be \label{eq: sumpi-proof1}
\begin{aligned}
    s_\pi(\tilde{\varphi}_\pm)(y)
    &=-\frac{1}{2\pi \ri}\int_0^{-\infty} \log(1-\re^{-t/y})\left(\frac{1}{1+\re^{\pm\frac{\pi i}{3}-t}}+\frac{1}{1+\re^{\pm\frac{\pi \ri}{3}+t}}\right) dt \\
    &=\frac{1}{2\pi \ri}\int_0^{\infty} \log(1-\re^{t/y})\left(\frac{1}{1+\re^{\pm\frac{\pi i}{3}+t}}+\frac{1}{1+\re^{\pm\frac{\pi \ri}{3}-t}}\right) dt \\
    &=\frac{1}{2\pi \ri}\int_{-\infty}^\infty \frac{\log(1-\re^{t/y})}{1+\re^{\pm\frac{\pi \ri}{3}-t}} dt -\frac{1}{2\pi \ri}\int_0^{-\infty}\frac{\log(-\re^{-t/y})}{1+\re^{\pm\frac{\pi \ri}{3}-t}} dt \, ,
\end{aligned}
\ee
where we have divided the integral into the sum of two contributions for simplicity. 
Let us now set $x=-1/y$, so that $\Re(x)>0$ and $\Im (x)>0$, then the second term in the RHS of Eq.~\eqref{eq: sumpi-proof1} becomes
\be \label{eq: sumpi-proof1-p1}
-\frac{1}{2\pi \ri}\int_0^{-\infty}\frac{\log(-\re^{t x})}{1+\re^{\pm\frac{\pi \ri}{3}-t}} dt = \frac{x} {2\pi \ri}\mathrm{Li}_2(w^{\pm 1}) + \frac{1}{2}\log(1-w^{\pm 1}) \, ,
\ee
while the first term gives
\begin{subequations}
    \begin{align}
        \frac{1}{2\pi \ri}\int_{-\infty}^\infty \frac{\log(1-\re^{-tx})}{1+\re^{\frac{\pi \ri}{3}-t}} dt &= -\log\frac{(w; \, \re^{2\pi \ri /x})_\infty}{(\re^{-4\pi \ri x/3}; \, \re^{-2\pi \ri x})_\infty} \, , \label{eq: sumpi-proof1-p2a}\\
        \frac{1}{2\pi \ri}\int_{-\infty}^\infty \frac{\log(1-\re^{-tx})}{1+\re^{-\frac{\pi \ri}{3}-t}} dt &= -\log\frac{(w^{-1}; \, \re^{2\pi \ri /x})_\infty}{(\re^{-2\pi \ri x/3}; \, \re^{-2\pi \ri x})_\infty}\, ,\label{eq: sumpi-proof1-p2b}
    \end{align}
\end{subequations}
following~\cite[Eq.~(5.108)]{wheeler-book}.
Therefore, substituting Eqs.~\eqref{eq: sumpi-proof1-p1},~\eqref{eq: sumpi-proof1-p2a}, and~\eqref{eq: sumpi-proof1-p2b}, after we change back $x$ to $-1/y$, into the formula for the Borel--Laplace sum in Eq.~\eqref{eq: sumpi-proof1}, we find that 
\begin{subequations}
    \begin{align}
    s_\pi(\tilde{\varphi}_+)(y)=& -\frac{\mathrm{Li}_2(w)}{2\pi \ri y}+\frac{1}{2}\log(1-w) -\log\frac{(w; \, \re^{-2\pi \ri y})_\infty}{(\re^{4\pi \ri/(3y)}; \, \re^{2\pi \ri/y})_\infty} \, , \label{eq: sumpi-proof1-p3a} \\
    s_\pi(\tilde{\varphi}_-)(y)=& -\frac{\mathrm{Li}_2(w^{-1})}{2\pi \ri y}+\frac{1}{2}\log(1-w^{-1}) 
    - \log\frac{(w^{-1}; \, \re^{-2\pi \ri y})_\infty}{(\re^{2\pi \ri/(3y)}; \, \re^{2\pi \ri/y})_\infty} \, . \label{eq: sumpi-proof1-p3b}
    \end{align}
\end{subequations}
Putting together Eqs.~\eqref{eq: sum0-proof-start},~\eqref{eq: sumpi-proof1-p3a}, and~\eqref{eq: sumpi-proof1-p3b} yields
\be \label{eq: sumpi-proof2}
\begin{aligned}
    s_\pi(\tilde{f}_0)(y)&=3 \log\frac{(w; \, \re^{-2\pi \ri y})_\infty}{(\re^{4\pi \ri/(3y)}; \, \re^{2\pi \ri/y})_\infty} -3 \log\frac{(w^{-1}; \, \re^{-2\pi \ri y})_\infty}{(\re^{2\pi \ri/(3y)}; \, \re^{2\pi \ri/y})_\infty} \\
    &=-3\log\frac{(w \, \re^{2\pi \ri y}; \, \re^{2\pi \ri y})_\infty}{(w^{-1} \, \re^{2\pi \ri y}; \, \re^{2\pi \ri y})_\infty} - 3\log\frac{(\re^{4\pi \ri/(3y)}; \, \re^{2\pi \ri/y})_\infty}{(\re^{2\pi \ri/(3y)}; \, \re^{2\pi \ri/y})_\infty} \, ,
\end{aligned}
\ee
where we have applied the relation in Eq.~\eqref{eq: qdilog-rel}. Let us now recall that
\be
(w^{\pm 1} \, \re^{2\pi \ri y}; \, \re^{2\pi \ri y})_\infty = \left(1- w^{\pm 1} \right)^{-1} (w^{\pm 1}; \, \re^{2\pi \ri y})_\infty  \, ,
\ee
which we substitute into Eq.~\eqref{eq: sumpi-proof2} to get
\be
\begin{aligned}
    s_\pi(\tilde{f}_0)(y)&=-\pi \ri -3\log\frac{(w; \, \re^{2\pi \ri y})_\infty}{(w^{-1}; \, \re^{2\pi \ri y})_\infty} - 3\log\frac{(\re^{4\pi \ri/(3y)}; \, \re^{2\pi \ri/y})_\infty}{(\re^{2\pi \ri/(3y)}; \, \re^{2\pi \ri/y})_\infty} \\
    &=f_0(y)-f_\infty(1/3y) \\
    &=f_0(y)-f_\infty(-1/3y) \, ,
\end{aligned}
\ee
where we have first applied the closed formulae for the generating functions in Eqs.~\eqref{eq:f_0-closed} and~\eqref{eq:f_inf-closed} and then used the parity property in Eq.~\eqref{eq:f_inf-symm}.\footnote{We note that the proof in the case of $\Re(y)>0$ can be adapted to allow for an arbitrary choice of $w=\exp(z)$, $z\in\IC^*$. However, the proof in the case of $\Re(y)<0$ explicitly uses the symmetry properties of the Stokes constants and their generating functions, which cannot be straightforwardly generalized.}
\end{proof}

We can now prove that the median resummation of the asymptotic expansion $\tilde{f}_0(y)$ reproduces the generating function $f_0(y)$ itself.
\begin{theorem}\label{thm:median-sum-Sinf}
    Let $\tilde{f}_0(y)$ be the asymptotic expansion in the limit $y\to 0$ of the generating function $f_0(y)$ in Eq.~\eqref{eq:f_0} with $\Im(y)>0$, which we have written explicitly in Eq.~\eqref{eq:asymp-f0}. Then, the median resummation of $\tilde{f}_0(y)$ reconstructs the original function $f_0(y)$, that is,
    \begin{equation}\label{eqn:median-sum-inf}
        \CS^\mathbf{med}_{\frac{\pi}{2}}\tilde{f}_0(y)=f_0(y) \, , \quad y\in\IH \, .
    \end{equation}
\end{theorem}
\begin{proof}
    Following the definition in Eq.~\eqref{eq: median2}, the median resummation of $\tilde{f}_0(y)$ along the positive imaginary axis is given by
    \begin{equation} \label{eq: resum-proof-zero}
               \mathcal{S}_{\frac{\pi}{2}}^{\mathbf{med}}\tilde{f}_0(y)=
    \begin{cases}
        s_0(\tilde{f}_0)(y)+\frac{1}{2}\,\mathrm{disc}_{\frac{\pi}{2}}\tilde{f}_0(y) \, , \quad & \Re(y)>0 \, ,\\
        & \\
        s_\pi(\tilde{f}_0)(y)-\frac{1}{2}\,\mathrm{disc}_{\frac{\pi}{2}}\tilde{f}_0(y) \, , \quad & \Re(y)<0 \, ,
    \end{cases}
    \end{equation}
    where we take $\Im(y)>0$. As a consequence of Eqs.~\eqref{eq: finftydisc} and~\eqref{eq: f0-psi}, we have that
    \begin{equation}
        f_\infty\left(-\tfrac{1}{3y}\right) = -\frac{1}{2} \mathrm{disc}_{\frac{\pi}{2}}\tilde{f}_0(y) \, , 
    \end{equation}
    and substituting this in Eq.~\eqref{eq: resum-proof-zero} yields
    \begin{equation}
    \begin{aligned}
               \mathcal{S}_{\frac{\pi}{2}}^{\mathbf{med}}\tilde{f}_0(y)=\begin{cases}
        s_0(\tilde{f}_0)(y)-\, {f}_\infty\big(-\tfrac{1}{3y}\big) \, , \quad
& \Re(y)>0 \, , \\
        & \\
      s_\pi(\tilde{f}_0)(y)+\, {f}_\infty\big(-\tfrac{1}{3y}\big)  \, , \quad & \Re(y)<0 \, , 
    \end{cases} 
        \end{aligned}
    \end{equation}
    where again $\Im(y)>0$. Then, the conclusion follows from Lemma~\ref{lemma:summability-logPhiNC}.
\end{proof}

Let us now consider the strong coupling asymptotic expansion $\tilde{f}_\infty(y)$ in Eq.~\eqref{eq:asymp-finf}. Despite its similarities with the dual formal series $\tilde{f}_0(y)$ in Eq.~\eqref{eq:asymp-f0}, the analog to the analytic argument of Lemma~\ref{lemma:summability-logPhiNC} is still missing. Therefore, relying on the support of numerical tests, we limit ourselves to conjecture the effectiveness of the median resummation for the generating functions of the strong coupling Stokes constants. 
\begin{conjecture}\label{conj:median-sum-S0}
    Let $\tilde{f}_\infty(y)$ be the asymptotic expansion in the limit $y\to 0$ of the generating function $f_\infty(y)$ in Eq.~\eqref{eq:f_inf} with $\Im(y)>0$, which we have written explicitly in Eq.~\eqref{eq:asymp-finf}. Then, the median resummation of $\tilde{f}_\infty(y)$ reconstructs the original function $f_\infty(y)$, that is,
    \begin{equation}\label{eqn:median-sum0}
        \CS^\mathbf{med}_{\frac{\pi}{2}}\tilde{f}_\infty(y)=f_\infty(y) \, , \quad y\in\IH \, .
    \end{equation}
\end{conjecture}
\begin{rmk}
    As a consequence Eqs.~\eqref{eq: f0disc} and~\eqref{eq: finf-phi}, we have that  
    \begin{equation}
        f_0\left(-\tfrac{1}{3y}\right) = \frac{1}{2} \mathrm{disc}_{\frac{\pi}{2}}\tilde{f}_\infty(y) \, , \quad y \in \IH \, . 
    \end{equation}
    Substituting this into the definition of the median resummation in Eq.~\eqref{eq: median2} for the asymptotic expansion $\tilde{f}_\infty(y)$ at the angle $\pi\ri/2$, Conjecture~\ref{conj:median-sum-S0} straightforwardly implies that the Borel--Laplace sums along the positive and negative real axes of $\tilde{f}_\infty(y)$ reproduce the original function $f_\infty(y)$ with a correction that is suitably encoded in the dual function $f_0(y)$, and vice versa. More precisely,
    \begin{subequations}
        \begin{align}
            s_0(\tilde{f}_\infty)(y)&= f_\infty(y) -f_0\big(-\tfrac{1}{3y}\big) \, , \quad \Re(y)>0 \, , \label{eq: BL-inf-0}\\
            s_\pi(\tilde{f}_\infty)(y)&=f_\infty(y) +f_0\big(-\tfrac{1}{3y}\big) \, , \quad \Re(y)<0 \, , \label{eq: BL-inf-pi}
        \end{align}
    \end{subequations}
    which represent the dual statement to Lemma~\ref{lemma:summability-logPhiNC}, are equivalent to Eq.~\eqref{eqn:median-sum0}.
\end{rmk}

Theorem~\ref{thm:median-sum-Sinf} and Conjecture~\ref{conj:median-sum-S0} can be placed within the context of the strong-weak resurgent symmetry discussed in Section~\ref{sec: strong-weak} as they provide additional arrows in the commutative diagram in Eq.~\eqref{diag: strong-weak2}. In particular, taking the median resummation of the asymptotic expansions represents the formal inverse of perturbatively expanding the generating functions. 
We show below the resulting completed diagram in terms of $f_0$, $f_\infty$ and their asymptotic expansions. 
\begin{equation} \label{diag: strong-weak-median}
\begin{tikzcd}[column sep = 2.8em, row sep=2.8em]
 \arrow[ddd] f_\infty(y) \arrow[rrr,shift left=0.5ex, "y \rightarrow 0"] & & & \arrow[lll, blue,dashed,shift left=0.5ex, "\text{median resummation}"] \tilde{f}_\infty(y) \arrow[sloped]{ddd}{\text{strong-weak}}[swap]{\text{symmetry}}
 \\ \\ \\
  \arrow[sloped]{uuu}{\text{strong-weak}}[swap]{\text{symmetry}} \tilde{f}_0(y) \arrow[rrr,blue,shift left=0.5ex, "\text{median resummation}"] & & & \arrow[uuu]  f_0(y) \arrow[lll,shift left=0.5ex, "y \rightarrow 0"]
\end{tikzcd}
\end{equation}
Note that all four arrows tracing the edges of the box are now invertible. We display Conjecture~\ref{conj:median-sum-S0} with a dashed arrow to distinguish it from the proven results that compose the rest of the diagram.

\begin{rmk}
The full content of the $q$, $\tilde{q}$-series in the block factorization of the spectral trace of local $\IP^2$ cannot be reconstructed by a simple Borel--Laplace resummation. Yet, and remarkably, the missing information is not lost. Instead, it is collected by the discontinuities that are then recovered by the median resummation.
The effectiveness of the median resummation as a summability tool is known in quantum Chern--Simons theory, as it was conjectured for the quantum invariants of knots and $3$-manifolds in~\cite[Conjecture~1.1]{costin-garoufalidis} and proven for the special cases of the trefoil knot~\cite{costin-garoufalidis} and the Seifert fibered homology sphere~\cite{qCS7}. We expect more evidence to be discovered.
Furthermore, as we discuss in the companion paper~\cite{FR1maths}, the effectiveness of the median resummation appears to be closely related to quantum modularity. 
\end{rmk}

\subsection{Quantum modularity of the generating functions}\label{sec:quantum-resurgent}
We will now present our results on the quantum modularity properties of the generating functions of the Stokes constants in both $\hbar$-regimes, that is, the periodic holomorphic functions $f_0(y)$ and $f_\infty(y)$ defined in Eqs.~\eqref{eq:f_0} and~\eqref{eq:f_inf}. Let us start our discussion by recalling the definition of a quantum modular form. We refer to~\cite[Chapter IV]{wheeler-thesis} for an overview. A first description of quantum modular forms was given by Zagier in~\cite{zagier_modular}. Namely, a function $f\colon\IQ_{\ne 0}\to\IC$ is called a weight-$\omega$ quantum modular form with respect to a subgroup $\Gamma\subseteq\mathsf{SL}_2(\IZ)$, where $\omega$ is a fixed integer or half-integer, if the cocycle
\begin{equation}\label{cocycle}
     h_\gamma[f](y):=(cy+d)^{-\omega} f\left(\frac{ay+b}{cy+d}\right) - f(y)
\end{equation} 
is \emph{better behaved} than $f(y)$ for all choices of $\gamma=\left(\begin{smallmatrix}
    a & b\\
    c & d
\end{smallmatrix}\right)\in\Gamma$. Conceptually, being better behaved means having better analyticity properties than the function $f$ itself---for instance, being a real analytic function over $\IR\setminus \{0,-d/c\}$. Quantum modular forms of weight zero are called quantum modular functions. Some examples of quantum modular forms are built from the usual holomorphic modular forms. In these cases, the cocycle in Eq.~\eqref{cocycle} is expressed in terms of the underlying modular form. See, \emph{e.g.},~\cite[Examples 0 to 4]{zagier_modular} and~\cite{bringmann-q-hyper,goswami2021quantum-theta,zagier2001vassiliev}. 

When the function $f$ is defined and analytic in the upper half-plane $\IH$, it already has good analyticity properties. In this case, the cocycle $h_\gamma[f]$ in Eq.~\eqref{cocycle} is required to be analytic in a domain larger than $\IH$---for instance, in $\IC'=\IC\setminus\IR_{\leq 0}$. The quantum modular forms in this class are called \emph{holomorphic quantum modular forms}. The following definition was proposed by Zagier~\cite[min 21:45]{zagier-talk}.
\begin{definition}
    An analytic function $f\colon\IH\to\IC$ is a weight-$\omega$ \emph{holomorphic quantum modular form} for a subgroup $\Gamma\subseteq\mathsf{SL}_2(\IZ)$, where $\omega$ is integer or half-integer, if the cocycle $h_\gamma[f]\colon\IH\to\IC$ in Eq.~\eqref{cocycle} extends holomorphically to\footnote{Note that $\IC_\gamma = \IC \setminus \left(-\infty; \, -d/c\right]$ when $c>0$ and $\IC_\gamma = \IC \setminus \left[-d/c; \, +\infty \right)$ when $c<0$. If $\gamma$ is the matrix $S = \left(\begin{smallmatrix}
        0 & -1 \\
        1 & 0
    \end{smallmatrix}\right)$, then $\IC_S= \IC'$.}
    \be
    \IC_\gamma:=\{y\in\IC\colon cy+d\in\IC'\}
    \ee
    for every $\gamma=\left(\begin{smallmatrix}
    a & b\\
    c & d
\end{smallmatrix}\right)\in\Gamma$.  
\end{definition}

Note that an analogous definition of holomorphic quantum modular form can be written for functions $f\colon\IH_{-}\to\IC$, where $\IH_{-}$ denotes the lower half of the complex plane. Remarkably, the resurgent analysis of the spectral trace of local $\IP^2$ in the weak and strong coupling regimes is a new source of examples of holomorphic quantum modular forms. As we will show in Theorem~\ref{thm:f0,finf-quantum modular}, the generating functions of the Stokes constants are holomorphic quantum modular functions for the congruence subgroup $\Gamma_1(3)$, which we have defined in Eq.~\eqref{eq: G13-definition}.   

Let us start by defining the holomorphic functions $F_R, F_S \colon \IH\to\IC$ as
\begin{subequations}
    \begin{align}
        F_R(y)&:=f_\infty\big(-\tfrac{1}{3y}\big)+f_0(y)\, , \label{eq: FR-def}\\
        F_S(y)&:=f_\infty(y)-f_0\big(-\tfrac{1}{3y}\big)\, . \label{eq: FS-def}
    \end{align}
\end{subequations}
Note that the function $F_R(y)$ is equal to the Borel--Laplace sum $s_0(\tilde{f}_0)(y)$ in Eq.~\eqref{eq: BL-zero-0} for $\Re(y)>0$, while the function $-F_S\big(-\tfrac{1}{3y}\big)$ is equal to the Borel--Laplace sum $s_\pi(\tilde{f}_0)(y)$ in Eq.~\eqref{eq: BL-zero-pi} for $\Re(y)<0$.  
They can, therefore, be expressed explicitly as functions of the strong coupling Stokes constants according to Eqs.~\eqref{eq: sumzero-cor-0} and~\eqref{eq: sumzero-cor-pi}, respectively. We find in this way that
\begin{subequations}
    \begin{align}
        F_R(y) &= -\frac{\pi \ri}{2}-\frac{3\mathcal{V}}{2\pi \ri y} +2\sum_{m\in\IZ_{\neq 0}} R_m\, \e\left(-\frac{m}{3y}\right) \, , \quad \Re(y)>0 \, , \label{eq: FR-Rn} \\
        F_S(y) &= \frac{\pi \ri}{2}-\frac{9\mathcal{V} y}{2\pi \ri} -2\sum_{m\in\IZ_{\neq 0}} R_m\, \e\left(my\right) \, , \quad \Re(y)>0\, ,
    \end{align}
\end{subequations}
where $\CV = 2 \Im\left(\mathrm{Li}_2(\re^{2 \pi \ri/3})\right)$ and $\e$ is the exponential integral in Eq.~\eqref{eq: e1-def} as before. Indeed, Theorem~\ref{thm:median-sum-Sinf} can be equivalently restated in terms of the functions above as
\be
\frac{F_R(y)-F_S\big(-\tfrac{1}{3y}\big)}{2} = f_0(y) \, , \quad y \in \IH \, . 
\ee
\begin{rmk}
    Assuming the validity of Conjecture~\ref{conj:median-sum-S0}, we can complete the above discussion as follows. We observe that the function $F_R\big(-\tfrac{1}{3y}\big)$ is equal to the Borel--Laplace sum $s_\pi(\tilde{f}_\infty)(y)$ in Eq.~\eqref{eq: BL-inf-pi} for $\Re(y)<0$, while the function $F_S(y)$ is equal to the Borel--Laplace sum $s_0(\tilde{f}_\infty)(y)$ in Eq.~\eqref{eq: BL-inf-0} for $\Re(y)>0$.  
    They can, therefore, be expressed explicitly as functions of the weak coupling Stokes constants according to Eqs.~\eqref{eq: suminf-cor-pi} and~\eqref{eq: suminf-cor-0}, respectively. Namely,
    \begin{subequations}
    \begin{align}
        F_R(y)&= -3\log\frac{\Gamma(2/3)}{\Gamma(1/3)}-\log(2\pi \ri /y)-2\sum_{m\in\IZ_{\neq 0}}S_m\, \e\left(my\right) \, , \quad \Re(y)>0 \,  , \\
        F_S(y)&= -3\log\frac{\Gamma(2/3)}{\Gamma(1/3)}-\log(-6\pi \ri y)-2\sum_{m\in\IZ_{\neq 0}}S_m\, \e\left(-\frac{m}{3y}\right) \, , \quad \Re(y)>0 \label{eq: FS-Sn} \, ,
    \end{align}
\end{subequations}
where $\e$ is again the exponential integral in Eq.~\eqref{eq: e1-def}.
Indeed, Conjecture~\ref{conj:median-sum-S0} is equivalent to
\be
\frac{F_S(y)+F_R\big(-\tfrac{1}{3y}\big)}{2} = f_\infty(y) \, , \quad y \in \IH \, . 
\ee
\end{rmk}
In addition, the functions $F_R(y)$ and $F_S(y)$ in Eqs.~\eqref{eq: FR-def} and~\eqref{eq: FS-def} can be written as the logarithm of a ratio of Faddeev's quantum dilogarithms. Indeed, setting $\mathsf{b}^2=y$, it follows from the infinite product representations in Eqs.~\eqref{eq:f_0-closed},~\eqref{eq:f_inf-closed}, and~\eqref{eq: seriesPhib} that
\begin{subequations}
\begin{align}
F_R(y)&=3\log\frac{\Phi_{\mathsf{b}}\big(\frac{2 \ri}{3\mathsf{b}}-c_\mathsf{b}\big)}{\Phi_{\mathsf{b}}(\frac{\ri}{3\mathsf{b}}-c_\mathsf{b})} - \ri\pi \label{eq: FR-phib} \, , \\
F_S(y/3)&=3\log\frac{\Phi_{\mathsf{b}}\big(\frac{2\ri\mathsf{b}}{3}-c_\mathsf{b}\big)}{\Phi_{\mathsf{b}}(\frac{\ri\mathsf{b}}{3}-c_\mathsf{b})} \label{eq: FS-phib} \, , 
\end{align}
\end{subequations} 
where $c_\mathsf{b}$ is defined in Eq.~\eqref{eq: cb-def}.
Thus, thanks to the analyticity properties of the Faddeev's quantum dilogarithm $\Phi_\mb(x)$ reviewed in Appendix~\ref{app: Faddeev}, $F_R(y)$ and $F_S(y)$ can be analytically continued to the whole cut complex plane $\IC'$.
We can now prove our statements on the quantum modularity of the generating functions. 
\begin{theorem}\label{thm:f0,finf-quantum modular}
   The generating functions $f_0,f_\infty\colon \IH\to\IC$, defined in Eqs.~\eqref{eq:f_0} and~\eqref{eq:f_inf}, respectively, are holomorphic quantum modular functions for $\Gamma_1(3)$.
\end{theorem}
\begin{proof}
Recall that the generators of $\Gamma_1(3)$ are the matrices $T$ and $\gamma_3$ in Eq.~\eqref{eq: G13-generators}. It is sufficient to prove that the corresponding cocycles, that is,
\begin{subequations}
    \begin{align}
        h_T[f](y)&=f(y+1)-f(y) \, , \\
        h_{\gamma_3}[f](y)&=f\big(\tfrac{y}{3y+1}\big)-f(y)\, ,
    \end{align}
\end{subequations}
for $f=f_0, f_\infty$, are analytic in $\IC'$.
It follows from the periodicity properties in Eqs.~\eqref{eq:f_0-symm} and~\eqref{eq:f_inf-symm} that the cocycles for $T$ are trivial, \emph{i.e.}, 
$h_T[f_0](y) = 0$ and $h_T[f_\infty](y)=0$.

Applying the definition of $F_R(y)$ in Eq.~\eqref{eq: FR-def} and the first formula in Eq.~\eqref{eq:f_inf-symm}, we find that
\be
f_0\big(\tfrac{y}{3y+1}\big)=F_R\big(\tfrac{y}{3y+1}\big)-f_\infty\big(-\tfrac{1}{3y}\big) =F_R\big(\tfrac{y}{3y+1}\big)-F_R(y)+f_0(y) \, .
\ee
Therefore, we obtain that the cocycle of $f_0$ for $\gamma_3$ is
\be \label{eq: f0-cocycle}
h_{\gamma_3}[f_0](y)=F_R\big(\tfrac{y}{3y+1}\big)-F_R(y) \, , 
\ee
which is analytic in $\IC'$ following Eq.~\eqref{eq: FR-phib}.
Analogously, applying the definition of $F_S(y)$ in Eq.~\eqref{eq: FS-def} and the first formula in Eq.~\eqref{eq:f_0-symm}, we find that
\be
f_\infty\big(\tfrac{y}{3y+1}\big)=
F_S\big(\tfrac{y}{3y+1}\big)+f_0\big(-\tfrac{1}{3y}\big)=F_S\big(\tfrac{y}{3y+1}\big)-F_S(y)+f_\infty(y) \, .
\ee
Therefore, the cocycle of $f_\infty$ for $\gamma_3$ is
\be \label{eq: finf-cocycle}
h_{\gamma_3}[f_\infty](y)=F_S\big(\tfrac{y}{3y+1}\big)-F_S(y) \, , 
\ee
which is analytic on $\IC'$ following Eq.~\eqref{eq: FS-phib}. 
\end{proof}
Let us stress that the generating functions $f_0(y)$ and $f_\infty(y)$, $y \in \IH$, have trivial cocycles for the generator $T$, while their cocycles for the generator $\gamma_3$ in Eqs.~\eqref{eq: f0-cocycle} and~\eqref{eq: finf-cocycle} are given by the functions $F_R(y)$ and $F_S(y)$ in Eqs.~\eqref{eq: FR-def} and~\eqref{eq: FS-def}, respectively. Specifically, we have that
\be
h_{\gamma_3}[f_0](y)=h_{\gamma_3}[F_R](y) \, , \quad h_{\gamma_3}[f_\infty](y)=h_{\gamma_3}[F_S](y) \, , 
\ee
which are determined by the strong and weak coupling Stokes constants $R_n$, $S_n$, $n \in \IZ_{>0}$, as a consequence of Eqs.~\eqref{eq: FR-Rn} and~\eqref{eq: FS-Sn}. Once more, note the exchange of information between the two regimes in $\hbar$. The non-trivial cocycle of the generating function of the Stokes constants in one regime is controlled by the Stokes constants in the dual regime.
\begin{rmk}
    As we learn from the proof of Theorem~\ref{thm:f0,finf-quantum modular}, the holomorphic quantum modular functions $f_0, f_\infty$ are deeply related to the Faddeev's quantum dilogarithm $\Phi_\mb$ in Eq.~\eqref{eq: seriesPhib}, which enters the formulae for the cocyles. This is indeed a general feature of the $q$-Pochhammer symbols, as discussed in~\cite[Section~8.11]{wheeler-thesis}. 
\end{rmk}
\subsubsection*{Quantum modularity and Fricke involution}
As we will show, a pair of ``dual'' holomorphic quantum modular functions for $\Gamma_1(3)$ is constructed by acting on the generating functions $f_0(y)$ and $f_\infty(y)$ with the transformation $y \mapsto -\frac{1}{3y}$, which is known as the Fricke involution of $\IH/\Gamma_1(3)$.
In particular, we define the holomorphic functions $f_0^{\star}, f_\infty^{\star}\colon \IH \to \IC$ as
\be \label{eq: fricke-fs}
f_0^{\star}(y):=f_0\big(-\tfrac{1}{3y}\big) \, , \quad  f^{\star}_\infty(y):=f_\infty\big(-\tfrac{1}{3y}\big) \, ,
\ee
where we have use the symbol $^{\star}$ to denote the action of the Fricke involution.
\begin{theorem}\label{thm:fricke-quantum modular}
The functions $f_0^{\star}, f_\infty^{\star}\colon \IH\to\IC$, defined in Eq.~\eqref{eq: fricke-fs} as the images of the generating functions under Fricke involution, are holomorphic quantum modular functions for $\Gamma_1(3)$.
\end{theorem}
\begin{proof}
Recall that the generators of $\Gamma_1(3)$ are the matrices $T$ and $\gamma_3$ in Eq.~\eqref{eq: G13-generators}. It is sufficient to prove that the corresponding cocycles, that is, 
\begin{subequations}
    \begin{align}
        h_T[f^\star](y)&=f^{\star}(y+1)-f^{\star}(y) = f\big(-\tfrac{1}{3y+3}\big) - f\big(-\tfrac{1}{3y}\big)\, , \\
        h_{\gamma_3}[f^\star](y)&=f^{\star}\big(\tfrac{y}{3y+1}\big)-f^{\star}(y) =f\big(-\tfrac{1}{3y}-1\big) - f\big(-\tfrac{1}{3y}\big) \, ,
    \end{align}
\end{subequations}
for $f=f_0, f_\infty$, are analytic in $\IC'$.
It follows from the periodicity properties in Eqs.~\eqref{eq:f_0-symm} and~\eqref{eq:f_inf-symm} that the cocycles for $\gamma_3$ are trivial, \emph{i.e.}, $h_{\gamma_3}[f_0^{\star}](y) = 0$ and $h_{\gamma_3}[f_\infty^{\star}](y) = 0$.

Applying the definition of $F_S(y)$ in Eq.~\eqref{eq: FS-def} and the first formula in Eq.~\eqref{eq:f_inf-symm}, we find that
\be
    \begin{aligned}
        f_0^{\star}(y+1)=f_0\big(-\tfrac{1}{3y+3}\big)&=f_\infty(y+1)-F_S(y+1)\\
        &=F_S(y)-F_S(y+1)+f_0\big(-\tfrac{1}{3y}\big)\\
        &=F_S(y)-F_S(y+1)+f_0^{\star}(y)\, .
    \end{aligned}
\ee
Therefore, we obtain that the cocycle of $f_0^{\star}$ for $T$ is
\be \label{eq: f0star-cocycle}
h_T[f_0^{\star}](y)= F_S(y)-F_S(y+1) \, , 
\ee
which is analytic on $\IC'$ following Eq.~\eqref{eq: FS-phib}.
Analogously, applying the definition of $F_R(y)$ in Eq.~\eqref{eq: FR-def} and the first formula in Eq.~\eqref{eq:f_0-symm}, we find that
\be
    \begin{aligned}
        f_\infty^{\star}(y+1)=f_\infty\big(-\tfrac{1}{3y+3}\big)
            &=F_R(y+1)-f_0(y+1)\\
            &=F_R(y+1)-F_R(y)+f_\infty\big(-\tfrac{1}{3y}\big)\\
             &=F_R(y+1)-F_R(y)+f_\infty^{\star}(y)\,,
    \end{aligned}
\ee
Therefore, the cocycle of $f_\infty^{\star}$ for $T$ is
\be \label{eq: finfstar-cocycle}
h_T[f_\infty^{\star}](y)=F_R(y+1)-F_R(y) \, , 
\ee
which is analytic on $\IC'$ following Eq.~\eqref{eq: FR-phib}.
\end{proof}
The Fricke involution has exchanged the roles played by the generators of the modular group $\Gamma_1(3)$ and, simultaneously, by the weak and strong coupling Stokes constants. Indeed, the dual functions $f^{\star}_0(y)$ and $f^{\star}_\infty(y)$, $y \in \IH$, have trivial cocycles for the generator $\gamma_3$, while their cocycles for the generator $T$ in Eqs.~\eqref{eq: f0star-cocycle} and~\eqref{eq: finfstar-cocycle} are given by the functions $F_S(y)$ and $F_R(y)$ in Eqs.~\eqref{eq: FS-def} and~\eqref{eq: FR-def}, respectively. Namely, we have that
\be
h_T[f_0^\star](y)=-h_T[F_S](y) \, , \quad h_T[f_\infty^\star](y)=h_T[F_R](y) \, , 
\ee
which are determined by the weak and strong coupling Stokes constants $S_n$, $R_n$, $n \in \IZ_{>0}$, as a consequence of Eqs.~\eqref{eq: FS-Sn} and~\eqref{eq: FR-Rn}. 

\begin{rmk}
   Our results on the quantum modularity of the generating functions and the role of the Fricke involution pave the way for a geometric interpretation of the Stokes constants. Indeed, as briefly discussed in Section~\ref{sec:geometry}, the group $\Gamma_1(3)$ is deeply related to the geometry of local $\IP^2$. The moduli space $\mathcal{M}_{\rm cpx}(\hat{X})$ parametrizing the complex structures of the mirror $\hat{X}$ is the compactification of the quotient $\IH/\Gamma_1(3)$. Furthermore, the generating functions of the Gromov--Witten invariants are known to be quasi-modular functions in the modular coordinate $\tau\in\mathcal{M}_{\rm cpx}(\hat{X})$ at the conifold point~\cite{CI}. This enticing connection opens a new direction of investigation aimed at studying an explicit relationship between the Stokes constants appearing in TS/ST correspondence and the Gromov--Witten invariants of toric CY threefolds, which we plan to address in upcoming work.
\end{rmk}
Finally, a large class of examples of quantum modular forms comes from the study of quantum invariants of $3$-manifolds and knots~\cite{qCS1,zagier_modular,DG,zagier2001vassiliev,goswami2021quantum-theta,CCFGH,Cheng1,Cheng2}. Thus, our results on the quantum modularity of the generating functions of the Stokes constants in Theorems~\ref{thm:f0,finf-quantum modular} and~\ref{thm:fricke-quantum modular} provide new evidence in addition to the many formal similarities that are shared by topological string theory and complex Chern--Simons theory~\cite{GuM,MR}, which we have briefly discussed in Section~\ref{sec: strong-weak}.

\section{Conclusions}\label{sec:conclusion}
The Stokes constants are analytic invariants capturing information about the non-perturbative corrections to a resurgent asymptotic series. Yet, in some remarkable cases, they are rational or integer numbers and possess an interpretation as enumerative invariants based on the counting of BPS states~\cite{kontsevich--analytic, IK1, IK2}. Examples come from various quantum theories, including $4d$ $\mathcal{N}= 2$ supersymmetric gauge theory in the NS limit of the Omega-background~\cite{GrGuM}, complex Chern--Simons theory on Seifert fibered homology spheres~\cite{qCS7} and on complements of hyperbolic knots~\cite{GGuM, GGuM2}, the standard and NS topological string theories on (toric) CY threefolds~\cite{Rella22, GuM2, GKKM, IM, GuM3, Gu23, GuM, PS, ASTT, GHN, AHT} and their refinement~\cite{AMP, GuGuo}, and Walcher’s real topological string~\cite{MS}.

In this paper, building on the key results of~\cite{Rella22}, we completed our understanding of the unique set of exact relations connecting the resurgent structures of the logarithm of the spectral trace of local $\IP^2$ in the weak and strong coupling limits. 
In particular, we uncovered a direct way of obtaining the perturbative expansion
of the generating function of the Stokes constants in one regime from the generating function in the other,
which amounts to performing an operation that is the formal inverse of taking the discontinuity of
the asymptotic series.
Thus, we argued that a \emph{strong-weak resurgent symmetry} is at the heart of the analytic number-theoretic duality of~\cite{Rella22}. Through this newly discovered exact symmetry, the two-way exchange of perturbative/non-perturbative information between the holomorphic and anti-holomorphic blocks in the factorization of the spectral trace takes a mathematically precise form. This is a realization of underlying physical mechanisms that can be intuitively traced back to the duality between the worldsheet and WKB contributions to the total grand potential of the topological string on local $\IP^2$. Moreover, the weak and strong coupling Stokes constants are related by an arithmetic twist, while the corresponding $L$-functions analytically continue each other through a functional equation.

Then, we proved that the $q$-series acting as generating functions of the Stokes constants are holomorphic quantum modular functions for the congruence subgroup $\Gamma_1(3) \subset \mathsf{SL}_2(\IZ)$. Finally, we showed the effectiveness of the median resummation in the weakly coupled regime, while the analogous statement at strong coupling is conjectured in light of numerical evidence.

Multiple questions and open problems follow from the investigation performed in this paper.
The geometric and physical meaning of the enumerative invariants and non-perturbative sectors of~\cite{GuM, Rella22} is yet to be understood. However, the remarkable arithmetic fabric underpinning the resurgent properties of the dual $\hbar$-expansions of $\log \mathrm{Tr}(\rho_{\IP^2})$ discovered in~\cite{Rella22}, and developed here into a global strong-weak resurgent symmetry, paves the way for new insights on the physical interpretation of the Stokes constants to be explored in upcoming work. 

Along the same lines, a study of the relation between the formalism of~\cite{GuM2, GuM3} and the framework proposed in~\cite{GuM, Rella22} and advanced here might help us achieve a more comprehensive understanding of the non-perturbative structure of topological string theory. Let us also mention that other techniques, which do not employ resurgence, have been recently used to address the problem of identifying and counting BPS states of type IIA compactifications on toric CY threefolds~\cite{DFR}. These include generalizations of the WKB approach~\cite{GMN, ESW, Longhi1, Longhi2} and techniques based on attractor flows~\cite{Bousseau22}. 

Simultaneously, our results on the quantum modularity of the generating functions suggest a new research direction as they naturally prompt us to investigate the geometric content of the weak and strong coupling Stokes constants in connection to the BPS spectrum of local $\IP^2$. In particular, there are promising links with some of the results of~\cite{Bousseau22}, which we plan to investigate in the future. 
 
Finally, it would be interesting to extend our results on the first fermionic spectral trace of local $\IP^2$ to other toric CY threefolds and possibly study higher-order fermionic spectral traces, where we expect that similar structures might show up, guiding us toward a generalization of the strong-weak resurgent symmetry.

\section*{Acknowledgements}
We thank 
Alba Grassi, 
Maxim Kontsevich, 
Marcos Mari\~no, 
Campbell Wheeler, 
and 
Don Zagier 
for many useful discussions.
This work has been supported by the ERC-SyG project ``Recursive and Exact New Quantum Theory'' (ReNewQuantum), which received funding from the European Research Council (ERC) within the European Union's Horizon 2020 research and innovation program under Grant No. 810573 and by the Swiss National Centre of Competence in Research SwissMAP (NCCR 51NF40-141869 The Mathematics of Physics).

\appendix

\section{Quantum dilogarithms} \label{app: Faddeev}
In this section, we recall some of the main properties of Faddeev's quantum dilogarithm~\cite{Faddeev2}. We refer to~\cite[Appendix~A]{AK} and~\cite[Section~8.10]{wheeler-thesis} for detailed expositions. We call \emph{quantum dilogarithm} the function of two variables defined by the series
\be \label{eq: dilog}
(x q^{\alpha}; \, q)_{\infty} = \prod_{n=0}^{\infty} (1- x q^{\alpha+n}) \, , \quad \alpha \in \IR \, ,
\ee
which is analytic in $x,q \in \IC$ with $|q| <1$ and has asymptotic expansions around $q$ a root of unity. It satisfies the relation
\be \label{eq: qdilog-rel}
(x; \, q^{-1})_\infty = (xq; \, q)_\infty^{-1} \, .
\ee
The \emph{$q$-Pochhammer symbols}, also known as $q$-shifted factorials, are denoted by
\be \label{eq: qFactor}
(x; \, q)_m = \prod_{n=0}^{m-1} (1- x q^n) \, , \quad (x; \, q)_{-m} = \frac{1}{(x q^{-m}; \, q)_m} \, , \quad m \in \IZ_{>0} \, ,
\ee
with $(x; \, q)_0=1$. Equivalently, we can write
\be \label{eq: qFactor2}
(x; \, q)_m = \frac{(x; \, q)_{\infty}}{(x q^m; \, q)_{\infty}} \, , \quad m \in \IZ \, .
\ee
The \emph{Faddeev's quantum dilogarithm} $\Phi_{\mb}(x)$ is defined in the strip $| \Im (x) | < | \Im (c_{\mb}) |$, where 
\be \label{eq: cb-def}
c_{\mb} = \ri (\mb + \mb^{-1})/2 \, , 
\ee
by the integral representation
\be \label{eq: intPhib}
\Phi_{\mb}(x) = \exp \left( \int_{\IR + \ri \epsilon} \frac{\re^{-2 \ri x z}}{4 \sinh(z \mb ) \sinh(z \mb^{-1})} \frac{d z}{z} \right) \, ,
\ee
which implies the symmetry properties
\be \label{eq: symmPhib}
\Phi_{\mb}(x) = \Phi_{-\mb}(x) = \Phi_{\mb^{-1}}(x) \, .
\ee
When $\Im(\mb^2) > 0$, the formula in Eq.~\eqref{eq: intPhib} is equivalent to
\be \label{eq: seriesPhib}
\Phi_{\mb}(x) = \frac{( \re^{2 \pi \mb (x + c_{\mb})}; \, q)_{\infty}}{(  \re^{2 \pi \mb^{-1} (x - c_{\mb})}; \, \tilde{q})_{\infty}} = \prod_{n=0}^{\infty} \frac{1-  \re^{2 \pi \mb (x + c_{\mb})} q^n}{1- \re^{2 \pi \mb^{-1} (x - c_{\mb})} \tilde{q}^n}\, ,
\ee
where 
\be
q = \re^{2 \pi \ri \mb^2} \, , \quad \tilde{q} = \re^{- 2 \pi \ri \mb^{-2}} \, .
\ee
Note that the function in Eq.~\eqref{eq: seriesPhib} can be extended to the region $\Im(\mb^2) < 0$ using Eq.~\eqref{eq: symmPhib} and further admits an analytic continuation to all values of $\mb$ such that $\mb^2 \notin \IR_{\le 0}$.
Moreover, $\Phi_{\mb}(x)$ can be extended to the whole complex $x$-plane as a meromorphic function with an essential singularity at infinity, poles at the points 
\be
x = c_{\mb} + \ri m \mb + \ri n \mb^{-1} \, ,
\ee
and zeros at the points 
\be
x = -c_{\mb} - \ri m \mb - \ri n \mb^{-1} \, ,
\ee
for $m,n \in \IN$. 
It satisfies the inversion formula
\be
\Phi_{\mb}(x) \Phi_{\mb}(-x) = \re^{\pi \ri x^2} \Phi_{\mb}(0)^2 \, , \quad \Phi_{\mb}(0) = \left( \frac{q}{\tilde{q}} \right)^{1/48} = \re^{\pi \ri (\mb^2 + \mb^{-2})/24} \, ,
\ee
the complex conjugation formula
\be
\overline{\Phi_{\mb}(x)} = \frac{1}{\Phi_{\overline{\mb}}(\overline{x}) } \, ,
\ee
and the quasi-periodicity relations
\begin{subequations} \label{periodicity}
\begin{align} 
\Phi_{\mb}(x \pm \ri \mb) &= \Phi_{\mb}(x) \left( 1 + \re^{2 \pi \mb x \pm \pi \ri \mb^2} \right)^{\mp 1} \, , \\
\Phi_{\mb}(x \pm \ri \mb^{-1}) &= \Phi_{\mb}(x) \left( 1 + \re^{2 \pi \mb^{-1} x \pm \pi \ri \mb^{-2}} \right)^{\mp 1} \, .
\end{align}
\end{subequations}
In addition, when $\mathsf{b}^2=M/N\in\IQ$ with $M, N \in \IZ_{>0}$ coprime, the expression for $\Phi_{\mathsf{b}}(x)$ in Eq.~\eqref{eq: seriesPhib} simplifies into~\cite[Theorem~1.9]{Faddeev--roots}
\begin{equation} \label{eq: faddev-rationals}
    \Phi_{\mathsf{b}}\left(\frac{x}{2\pi\sqrt{MN}}-c_{\mathsf{b}}\right)=\frac{\re^{\tfrac{\ri}{2\pi MN}\mathrm{Li}_2(\re^x)} \, (1-\re^x)^{1+\tfrac{\ri x}{2\pi MN}}}{D_N(\re^{x/N}; \, \re^{2\pi \ri M/N}) \, D_M(\re^{x/M}; \, \re^{2\pi \ri N/M})} \, ,
\end{equation}
where 
\be
D_N(z;\, q)=\prod_{k=1}^{N-1}(1-z q^k)^{\tfrac{k}{N}}
\ee
is the cyclic quantum dilogarithm.\footnote{We use the simplified formula for Faddeev's quantum dilogarithm in Eq.~\eqref{eq: faddev-rationals} to perform numerical checks at rational points of the results presented in Section~\ref{sec:median_resummation}.}

Finally, we recall the known asymptotic behavior of Faddeev's quantum dilogarithm. We refer to Appendix~\ref{app: resurgence} for an introduction to resurgence.
In the limit of $\mb \rightarrow 0$ with $\Im(\mb^2) >0$, $\Phi_{\mb}(x)$ gives the asymptotic series~\cite{AK}
\be \label{eq: logPhib}
\log \Phi_{\mb}\left( \frac{x}{2 \pi \mb} \right)\sim \sum_{k=0}^{\infty} (2 \pi \ri \mb^2)^{2k-1} \frac{B_{2k}(1/2)}{(2k)!} \mathrm{Li}_{2-2k}(-\re^x) \, ,
\ee
where $x$ is kept fixed, $\mathrm{Li}_n(z)$ is the polylogarithm of order $n$, and $B_n(z)$ is the $n$-th Bernoulli polynomial. Similarly, for $\mb \rightarrow 0$ with $\Im(\mb^2) >0$ and $x$ fixed, the following special cases of the quantum dilogarithm in Eq.~\eqref{eq: dilog} have the asymptotic expansions~\cite{Katsurada}
\begin{subequations}
\begin{align}
\log (x ; \, q)_{\infty} \sim & \frac{1}{2} \log(1-x) + \sum_{k=0}^{\infty} (2 \pi \ri \mb^2)^{2k-1} \frac{B_{2k}}{(2k)!} \mathrm{Li}_{2-2k}(x) \, , \label{eq: logPhiNC} \\
\log (q^{\alpha} ; \, q)_{\infty} \sim & - \frac{\pi \ri}{12 \mb^2} - B_1(\alpha) \log(- 2 \pi \ri \mb^2) - \log \frac{\Gamma(\alpha)}{\sqrt{2 \pi}}  \label{eq: logPhiK} \\
& - B_2(\alpha) \frac{\pi \ri \mb^2}{2} - \sum_{k=2}^{\infty} (2 \pi \ri \mb^2)^k \frac{B_k B_{k+1}(\alpha)}{k (k+1)!} \, , \quad \alpha > 0 \, , \nonumber
\end{align}
\end{subequations}
where $B_n = B_n(0)$ is the $n$-th Bernoulli number and $\Gamma(\alpha)$ is the gamma function.
Besides, $\Phi_{\mb}(x)$ is Borel--Laplace summable~\cite[Theorem 1.3]{GK}.

\section{Basics of resurgence}\label{app: resurgence}
The resurgent analysis of factorially divergent formal power series unveils a universal mathematical structure involving a set of numerical data called Stokes constants. See~\cite{ABS, diver-book, Dorigoni} for a formal introduction to the theory of resurgence and~\cite{lecturesM, bookM} for its application to gauge and string theories. In this section, we review the basics of resurgence.
Let $z$ be a formal variable. The Borel transform is a map acting on $z$-monomials as
\begin{equation}
    \borel[z^{n-\alpha}]:=\frac{\zeta^{n-\alpha-1}}{\Gamma(n-\alpha)}\, , \quad n \in \IZ_{\ge 0} \, , \quad \alpha \in \IR \backslash \IZ_{\geq 0} \, ,
\end{equation}
where $\zeta$ is a new formal variable and $\Gamma(n-\alpha)$ is the gamma function. The action of the Borel transform above is extended by linearity to all formal power series in $z^{-\alpha}\IC[\![z]\!]$, while the Borel transform of the identity is conventionally denoted with the symbol $\delta:=\borel[1]$.
A Gevrey-1 asymptotic series $\phi(z)$ is defined as
\be \label{eq: phi}
\phi(z) = z^{-\alpha} \sum_{n=0}^{\infty} a_n z^n \in z^{-\alpha} \IC[\![z]\!] \, , \quad |a_n|\le \CA^{-n} n! \quad n \gg 1 \, ,
\ee
where $\alpha \in \IR \backslash \IZ_{\ge 0}$ and $\CA \in \IR_{>0}$.\footnote{The factorial growth of the coefficients $a_n$ as $n \to \infty$ dictates the Gevrey order.} 
Its Borel transform, which we denote for simplicity by $\hat{\phi}(\zeta):=\borel\phi(\zeta)$, is by definition
\be \label{eq: phihat}
\hat{\phi}(\zeta) = \sum_{k=0}^{\infty} \frac{a_k}{\Gamma(k-\alpha)} \zeta^{k-\alpha-1} \in \zeta^{- \alpha} \IC\{\zeta\}
\ee
and is a holomorphic function in an open neighborhood of $\zeta=0$ of radius $|\CA|$. 
When extended to the complex $\zeta$-plane, known as Borel plane, $\hat{\phi}(\zeta)$ shows a (possibly infinite) set of singularities $\zeta_{\omega} \in \IC$, which we label by the index $\omega \in \Omega$. 
A ray in the Borel plane of the form 
\be \label{eq: ray}
\CC_{\theta_{\omega}} = \re^{\ri \theta_{\omega}} \IR_{\ge 0} \, , \quad \theta_{\omega} = \arg (\zeta_{\omega}) \, ,
\ee 
which starts at the origin and passes through a singularity $\zeta_{\omega}$, is called a Stokes ray. 
The Borel plane is partitioned into sectors bounded by the Stokes rays in such a way that the Borel transform converges to a generally different holomorphic function in each sector.

We recall that the Gevrey-1 asymptotic series $\phi(z)$ is called \emph{resurgent} if its Borel transform $\hat{\phi}(\zeta)$ can be endlessly analytically continued. Namely, for every $L>0$, there is a finite set of points $\Omega_L$ in the Riemann surface of $\zeta^{\alpha}$ such that $\hat{\phi}(\zeta)$ can be analytically continued along any path that avoids $\Omega_L$ and has length at most $L$. If, additionally, its Borel transform has only simple poles and logarithmic branch points, then it is called \emph{simple resurgent}. Let us assume that the formal power series $\phi(z)$ in Eq.~\eqref{eq: phi} is simple resurgent.
If the singularity $\zeta_{\omega}$ is a simple pole, the local expansion of the Borel transform in Eq.~\eqref{eq: phihat} around it has the form 
\be \label{eq: Stokes0}
\hat{\phi}(\zeta) = - \frac{S_{\omega}}{2 \pi \ri (\zeta - \zeta_{\omega})} + \text{regular in $\zeta-\zeta_\omega$} \, ,
\ee
where $S_{\omega} \in \IC$ is the Stokes constant at $\zeta_{\omega}$.
Whereas, if the singularity $\zeta_{\omega}$ is a logarithmic branch point, the local expansion of the Borel transform in Eq.~\eqref{eq: phihat} around it has the form
\be \label{eq: Stokes-log}
\hat{\phi}(\zeta) = - \frac{S_{\omega}}{2 \pi \ri} \log(\zeta - \zeta_{\omega}) \hat{\phi}_{\omega}(\zeta - \zeta_{\omega}) + \text{regular in $\zeta-\zeta_\omega$} \, ,
\ee
where again $S_{\omega} \in \IC$ is the corresponding Stokes constant. If we introduce the variable $\xi = \zeta - \zeta_{\omega}$, the function
\be \label{eq: phihat2}
\hat{\phi}_{\omega}(\xi) = \sum_{k=0}^{\infty} \hat{a}_{k, \omega} \xi^{k-\beta} \in \xi^{-\beta} \IC\{\xi \}\, ,
\ee
where $\beta \in \IR \backslash \IZ_{\ge 0}$, is locally analytic at $\xi = 0$ and can be regarded as the Borel transform of the Gevrey-1 asymptotic series 
\be \label{eq: phi2}
\phi_{\omega}(z) = z^{-\beta} \sum_{n=0}^{\infty} a_{n, \omega} z^n \in z^{-\beta} \IC[\![z]\!] , \, \quad a_{n, \omega} = \Gamma(n-\beta+1) \, \hat{a}_{n, \omega} \, .
\ee
Note that the value of the Stokes constant $S_{\omega}$ depends on the normalization of $\phi_{\omega}(z)$.

If the analytic continuation of the Borel transform $\hat{\phi}(\zeta)$ in Eq.~\eqref{eq: phihat} does not grow too fast at infinity\footnote{Roughly, we require that the Borel transform grows at most exponentially in an open sector of the Borel plane containing the angle $\theta$.}, its Laplace transform at an arbitrary angle $\theta$ in the Borel plane gives the Borel--Laplace sum of the original, divergent formal power series $\phi(z)$, denoted by $s_{\theta}(\phi)(z)$.
Explicitly,
\be \label{eq: Laplace}
s_{\theta}(\phi)(z) = \int_0^{\re^{\ri \theta} \infty} \re^{-\zeta/z} \hat{\phi}(\zeta) \, d \zeta = z \int_0^{\re^{\ri \theta} \infty} \re^{-\zeta} \hat{\phi}(\zeta z) \, d \zeta \, , 
\ee
whose asymptotics near the origin reconstructs $\phi(z)$. 
If the Borel--Laplace sum in Eq.~\eqref{eq: Laplace} exists in some region of the complex $z$-plane, we say that the series $\phi(z)$ is Borel--Laplace summable along the direction $\theta$. Note that the Borel--Laplace sum inherits the sectorial structure of the Borel transform. It is a locally analytic function with discontinuities across the special rays identified by 
\be
\arg(z)=\arg(\zeta_{\omega}) \, , \quad \omega \in \Omega \, .
\ee 

The discontinuity across an arbitrary ray $\CC_{\theta} = \re^{\ri \theta} \IR_{\ge 0}$ in the complex $z$-plane is the difference between the Borel--Laplace sums along two rays in the complex $\zeta$-plane that lie slightly above and slightly below $\CC_{\theta}$. Namely,
\be \label{eq: disc}
\mathrm{disc}_{\theta}\phi(z) = s_{\theta_+}(\phi)(z) - s_{\theta_-}(\phi)(z) = \int_{\mathcal{C}_{\theta_+} - \, \mathcal{C}_{\theta_-}} \re^{-\zeta/z} \hat{\phi}(\zeta) \,  d \zeta \, , 
\ee
where $\theta_{\pm}= \theta \pm \epsilon$ for some small positive angle $\epsilon$.
A standard contour deformation argument shows that the two lateral Borel--Laplace sums differ by exponentially small terms. More precisely, if the Borel transform $\hat{\phi}(z)$ has only simple poles, then we have that
\be \label{eq: Stokes1-poles}
\mathrm{disc}_{\theta}\phi(z) = \sum_{\omega  \in \Omega_{\theta}} S_{\omega} \re^{-\zeta_{\omega}/z}  \, ,
\ee
where the index $\omega \in \Omega_{\theta}$ labels the singularities $\zeta_{\omega}$ such that $\arg (\zeta_{\omega}) = \theta$ and the complex numbers $S_{\omega}$ are the same Stokes constants that appear in Eq.~\eqref{eq: Stokes0}.
Similarly, if there are only logarithmic branch points, we have that
\be \label{eq: Stokes1}
\mathrm{disc}_{\theta}\phi(z) = \sum_{\omega  \in \Omega_{\theta}} S_{\omega} \re^{-\zeta_{\omega}/z} s_{\theta_-}(\phi_{\omega})(z) \, ,
\ee
where again the complex numbers $S_{\omega}$ are the Stokes constants appearing in Eq.~\eqref{eq: Stokes-log}, while $\phi_{\omega}(z)$ is the formal power series in Eq.~\eqref{eq: phi2}.

Let us now focus on the case of logarithmic singularities. If we regard the lateral Borel--Laplace sums as operators acting on formal power series, the Stokes automorphism $\mathfrak{S}_{\theta}$ at an arbitrary angle $\theta$ is defined by the composition
\be \label{eq: autStokes}
s_{\theta_+} = s_{\theta_-} \circ \mathfrak{S}_{\theta} \, ,
\ee
and the discontinuity formula in Eq.~\eqref{eq: Stokes1} has the equivalent form
\be
\mathfrak{S}_{\theta}(\phi) = \phi +  \sum_{\omega  \in \Omega_{\theta}} S_{\omega} \re^{-\zeta_{\omega}/z} \phi_{\omega} \, .
\ee
An interesting feature of resurgence is that it allows us to build a whole family of new asymptotic series $\phi_\omega$, $\omega \in \Omega$, from the knowledge of the singularity structure of the Borel transform of one single resurgent series $\phi$. 
We can then repeat the procedure above with each new series obtained this way, assuming they are simple resurgent.
For each $\omega \in \Omega$, let us define the basic trans-series
\be
\Phi_{\omega}(z) = \re^{-\zeta_{\omega}/z} \phi_{\omega}(z) \, ,
\ee
such that its Borel--Laplace sum at the angle $\theta$ is given by
\be
s_{\theta}(\Phi_{\omega})(z) = \re^{-\zeta_{\omega}/z} s_{\theta}(\phi_{\omega})(z) \, .
\ee
The corresponding Stokes automorphism acts on the basic trans-series as
\be
\mathfrak{S}_{\theta}(\Phi_{\omega}) = \Phi_{\omega} + \sum_{\omega' \in \Omega_{\theta}} S_{\omega \omega'} \Phi_{\omega'} \, ,
\ee
where we have denoted by $S_{\omega \omega'} \in \IC$ the Stokes constants of the secondary asymptotic series $\phi_{\omega}(z)$. 
The minimal resurgent structure associated with $\phi(z)$ is the smallest collection of basic trans-series that resurge from it and form a closed set under Stokes automorphisms. It is denoted by~\cite{GuM}
\be
\mathfrak{B}_{\phi} = \left\{ \Phi_{\omega}(z) \right\}_{\omega \in \bar{\Omega}} \, , \quad \bar{\Omega} \subseteq \Omega \, .
\ee
We stress that $\mathfrak{B}_{\phi}$ does not necessarily include all the basic trans-series arising from $\phi(z)$. 
Finally, we construct the (possibly infinite-dimensional) matrix of Stokes constants 
\be
\mathcal{S}_{\phi} = \{S_{\omega \omega'} \}_{\omega, \omega' \in \bar{\Omega}} \, ,
\ee
which is indexed by the distinct basic trans-series in the minimal resurgent structure of $\phi(z)$ and incorporates additional information about the non-analytic corrections to the original asymptotic series. In the case of simple poles, the same definitions and considerations above apply after imposing that $\phi_\omega(z)=1$, $\omega \in \Omega$.

Finally, we recall that the median resummation of the Gevrey-1 asymptotic series $\phi(z)$ in Eq.~\eqref{eq: phi} across an arbitrary ray $\CC_{\theta}$ in the complex $z$-plane is the average of the two lateral Borel--Laplace sums $s_{\theta_\pm}(\phi)(z)$, that is,
\be \label{eq: median}
\mathcal{S}^{\text{med}}_{\theta}\phi(z) = \frac{s_{\theta_+}(\phi)(z) + s_{\theta_-}(\phi)(z)}{2} \, ,
\ee
which is an analytic function for $\arg z\in (\theta-\frac{\pi}{2},\theta+\frac{\pi}{2})$. Equivalently, we can write
\begin{equation} \label{eq: median2}
    \mathcal{S}^{\text{med}}_{\theta}\phi(z) =\begin{cases}
        s_{\theta_-}(\phi)(z)+\frac{1}{2}\,\mathrm{disc}_{\theta}\phi(z) \, , \quad & \Re \left( \re^{-\ri\theta_{-}}z \right)>0 \, , \\
        & \\
        s_{\theta_+}(\phi)(z)-\frac{1}{2}\,\mathrm{disc}_{\theta}\phi(z) \, , \quad & \Re \left( \re^{-\ri\theta_{+}}z \right)>0 \, ,
    \end{cases}
\end{equation}
where $\Re \left( \re^{-\ri\theta} z \right)>0$.
The expression above highlights how the discontinuities can be suitably interpreted as corrections to the standard Borel--Laplace transform. In addition, by varying the contour of integration of the Laplace transform, the domain of analyticity of the median resummation can be extended, which is helpful in some applications. 

\addcontentsline{toc}{section}{References}
\bibliographystyle{JHEP}
\linespread{0.4}
\bibliography{localP2-biblio-3}

\end{document}